\documentclass[preprint,12pt,authoryear]{elsarticle}

\usepackage{amssymb}
\usepackage{amsmath}
\usepackage{booktabs}
\usepackage{graphicx} 

\usepackage{pdflscape}
\usepackage{pifont}
\usepackage{hyperref}
\usepackage{mathtools} 
\usepackage{xcolor}
\usepackage{tabularx}
\usepackage{float}      
\usepackage{placeins}   
\usepackage{array}
\usepackage{makecell}
\newcommand{\cmark}{\ding{51}} 
\usepackage{siunitx}    
\usepackage{amsthm}
\newtheorem{theorem}{Theorem}
\newtheorem{proposition}{Proposition}   
\usepackage{enumitem}
\usepackage[ruled,vlined,linesnumbered]{algorithm2e}

\journal{Computers \& Chemical Engineering}

\begin{document}

\begin{frontmatter}



\title{Privacy-Preserving Coordinated Operation of Multi-Player Industrial Network Using Secure Aggregation} 


\author[aff1]{Akshdeep Singh Ahluwalia}
\author[aff2]{Zachary Wilson}
\author[aff3]{Jeffrey E.~Arbogast}
\author[aff1]{Can Li\corref{cor1}}

\cortext[cor1]{Corresponding author}
\ead{canli@purdue.edu}

\affiliation[aff1]{organization={Davidson School of Chemical Engineering, Purdue University},
            city={West Lafayette},
            state={IN},
            postcode={47907},
            country={USA}}

\affiliation[aff2]{organization={Air Liquide Innovation Campus, Air Liquide},
            city={Newark},
            state={DE},
            postcode={19702},
            country={USA}}

\affiliation[aff3]{organization={Digital \& AI, Air Liquide},
            city={Paris},
            postcode={75011},
            country={France}}


\begin{abstract}
Electrified chemical industries with operational flexibility can reduce operating costs by shifting production and distribution decisions in response to time-varying electricity prices. However, chemical plants are rarely isolated; they operate within process networks where coordinated demand response can exploit flexibility across multiple stakeholders. A centralized coordination scheme would require access to stakeholders' local scheduling models and proprietary operational data, which is often incompatible with data-privacy requirements. Distributed optimization with an independent central coordinator (ICC) avoids direct model sharing, but the iterative exchange of coupling variables can still reveal information about private model parameters.

We propose a privacy-preserving distributed coordination framework for coordinated demand response in industrial networks. The framework integrates secure aggregation with an ICC-based alternating direction method of multipliers (ADMM) algorithm, so that plant-level messages are numerically masked and become useful to the ICC only after aggregation. We test the framework on a multi-plant industrial gas network in which three air-separation units jointly schedule production and shipment decisions to shared customer regions. To support stable participation, we incorporate a two-phase revenue-sharing mechanism that reallocates realized savings so that every plant improves relative to its decentralized status quo. In a 31-day rolling-horizon simulation using synthetic data representing heterogeneous electricity prices and demand, the coordinated policy reduces total network cost by 19.77\% relative to decentralized operation and achieves a full-month cost within 3.08\% of that obtained under a centralized social-welfare-maximization benchmark. We further quantify a conservative worst-case collusion mode, showing how unmasked iterates and auxiliary information can expose private objective parameters.
\end{abstract}




\end{frontmatter}

\newpage

\section{Introduction}
\label{sec:introduction}

The urgent need to address climate change has driven global initiatives like the Paris Agreement of 2015 to set stringent carbon emission standards. Chemical industries, contributing 7\% to global greenhouse gas emissions, are increasingly focusing on electrification as an established strategy for swift decarbonization \citep[]{Tickner2021TransitioningCrises}. However, integrating electrification faces challenges due to volatile electricity prices caused by the intermittent nature of renewable energy in the grid \cite[]{Kyritsis2017ElectricityImplications}. Demand Response (DR) refers to the mechanism by which electricity consumers, including residential, commercial, and industrial users, modify their electricity consumption patterns in response to fluctuations in electricity prices \cite[]{Zhang2016}. A wide range of chemical industries, including air separation, steel and cement production, aluminum processing, and chlor-alkali, possess inherent potential for demand response participation \citep[]{Allman2022}. Their operational flexibility allows them to increase production when electricity prices are low and reduce consumption during high-price periods. Moreover, their large scale and inherent product storage capabilities make them particularly well suited for demand response programs \cite[]{Strbac2008,Tsay2019,Zhang2015}. However, optimizing process schedules in isolation may limit the effectiveness of DR \citep[]{Klaucke2020}. This is primarily because these processes are typically components of a larger interconnected network. Coordinated operations across this network can optimize overall energy usage and cost, leveraging synergies and shared capacities that individual operations cannot achieve alone \citep[]{Wassick2009,Allman2020,Allman2022}.

Early studies on coordinated optimization in industrial supply chain networks established that substantial cost savings can be achieved through system-wide planning across multiple plants \citep{Marchetti2014,Neiro2022}. However, these studies assumed either the existence of a trusted decision maker with access to all plant-level scheduling models or an institutional setting in which all assets belong to the same parent organization. Such assumptions are often unrealistic in practice, since supply chain networks commonly involve independent stakeholders that regard their process models, cost structures, and operational parameters as proprietary. Consequently, any practically relevant framework for coordinated demand response must address two fundamental challenges: (i) preserving the privacy of local scheduling models and associated operational data, and (ii) allocating coordination gains in a way that provides all participants with sufficient incentive to cooperate.

Recently, \citet{Allman2022} proposed a fairness-guided framework for coordinated demand response that combines distributed optimization for operational coordination with game-theoretic mechanisms for benefit allocation. Their framework uses the alternating direction method of multipliers (ADMM) to coordinate plant-level scheduling decisions without requiring full model sharing, and subsequently applies mechanisms such as Nash bargaining to distribute the realized benefits fairly. In this setting, ADMM enables each plant to solve a local scheduling problem iteratively while exchanging only limited coordination messages, such as shipping quantities, with an independent central coordinator (ICC) \citep[]{Boyd2010}. At each iteration, the ICC aggregates plant-level shipping estimates and returns the residual quantity required to satisfy network-wide demand. This approach has demonstrated the economic potential of coordinated demand response across several case studies \citep{Allman2020,Allman2022}. However, while game-theoretic mechanisms address the allocation of realized coordination gains, the privacy of the information exchanged during distributed coordination remains a distinct and unresolved challenge. Although ADMM avoids direct disclosure of complete plant-level scheduling models, the sequence of communicated iterates may still encode sensitive information about local feasible regions and private cost parameters. Indeed, prior work has shown that sensitive parameters of entities participating in ADMM-based coordination can be inferred under adversarial attacks \citep{Zhang2019,Dvorkin2020}. This leakage risk has neither been explicitly assessed nor incorporated into prior coordinated demand response studies. Therefore, privacy-by-limited-disclosure does not provide formal privacy protection, leaving unresolved the practical challenge of ensuring privacy while leveraging ADMM for coordinated demand response.

To facilitate the practical adoption of coordinated demand response, ADMM-based coordination frameworks must therefore incorporate explicit privacy protection for the messages exchanged during iterative coordination. Our framework modifies the ICC--ADMM architecture by requiring each plant to numerically mask its shared shipping estimates so that the masks cancel only upon aggregation. As a result, the coordinator recovers only the aggregate quantity required for the ADMM update, while individual plant-level shipping information remains concealed from the coordinator. This privacy-enhancing mechanism is an instance of secure aggregation, whose purpose is to ensure that communicated information remains useful only in aggregated form. Secure aggregation has also been employed in other privacy-sensitive distributed settings, including federated learning and smart-metering systems \citep[]{bonawitz2017practical,thoma2012secure}. After coordination is achieved through secure-aggregation-based ADMM in the first phase, the realized savings are redistributed in a second phase based on Nash bargaining. Our main contributions are summarized as follows:

\begin{itemize}
    \item \textbf{Privacy-aware ICC--ADMM coordination mechanism:} We propose an ICC--ADMM coordination scheme in which plants exchange only masked messages via Diffie--Hellman-based neighbor masking, enabling secure aggregation of the quantities required for coordination while reducing plant-level information disclosure.

    \item \textbf{Benchmarking on an industrial-scale case study:} We demonstrate the practical efficacy of the proposed architecture on a simulated multi-plant industrial gas network using industrial-scale mixed-integer linear programming (MILP) scheduling models. The network jointly optimizes production and shipping decisions over a 31-day horizon while preserving privacy during coordination through secure aggregation.
    
    \item \textbf{Limitation analysis under worst-case collusion:} We evaluate privacy under an extreme and practically unlikely collusion regime in which all but one stakeholder collude against the remaining participants. Although masking-based secure aggregation protects against collusion as long as at least two plants remain honest, this worst-case setting falls outside that protection regime and therefore serves as a conservative stress test. We use it to quantify residual privacy leakage and show that, when revealed iterates are combined with plausible side information, an adversary may formulate an empirical data-fitting recovery model to infer private cost parameters of the remaining participant.
\end{itemize}

\paragraph{Paper organization}
Section~\ref{sec:literature-review} provides a detailed literature review and summarizes the state of the art in privacy-aware industrial coordination. Section~\ref{sec:background} gives an overview of secure aggregation methods and motivates our choice of neighbor masking. Section~\ref{sec:methods} presents the system model, the status quo and social-welfare formulations, and the ICC--ADMM coordination scheme with secure aggregation. Section~\ref{sec:case-study-and-results} describes the industrial gas-network-based case study, the simulation setup, and the coordination and revenue-sharing results. Section~\ref{sec:adv_model} introduces the worst-case collusion threat model and the adversarial recovery model, followed by numerical evidence of privacy leakage. Section~\ref{sec:conclusion} concludes and outlines directions for strengthening privacy and improving the robustness of coordinated industrial demand response.

\section{Related Works}
\label{sec:literature-review}


Energy-intensive chemical and metallurgical industries such as air separation plants \citep[]{Ierapetritou2002,Zhang2015,Tsay2019}, chlor-alkali processes \citep[]{Bre2019,Otashu2019}, and steel plants \citep[]{Castro2020} have emerged as important application domains for demand response. In the context of air separation units (ASUs), \citet{Ierapetritou2002} employed a two-stage stochastic programming framework to account for uncertainty in electricity price realizations during scheduling. Subsequently, \citet{Zhang2015} investigated the integration of cryogenic energy storage (CES) with ASUs to enhance load-shifting flexibility and create additional opportunities for participation in ancillary service markets. \citet{Tsay2019} further proposed a data-driven optimization framework that integrates scheduling and control for demand response in industrial ASUs. They also performed Monte Carlo sensitivity analysis to quantify the impact of electricity price forecast volatility, evaluating participation in both day-ahead and real-time electricity markets and demonstrating significant economic benefits. Similarly, \citet{Bre2019} developed a mode-switching MILP framework for the chlor-alkali process, motivated by the sector's large installed capacity, high electricity intensity, and widespread industrial deployment. They considered two operating modes with distinct power demands and showed that, despite downtime penalties associated with switching, equipment oversizing combined with mode switching can yield substantial long-term cost savings. In contrast, \citet{Otashu2019} presented a dynamic membrane-cell model for the chlor-alkali process and demonstrated, through an industrial-scale simulation and optimization case study, the provision of fast demand response. Although cell temperature limits the achievable response speed and available demand response capacity, the process was still shown to provide substantial load curtailment during peak-price periods. Beyond chemical processes, \citet{Castro2020} proposed a novel MILP-based demand response framework for steel production that explicitly incorporates alternative electric arc furnace operating modes together with electrode degradation and replacement. Their formulation highlights the trade-off between shifting production toward low-price hours and the associated impacts on energy efficiency and electrode consumption.

The capacity to modulate production rates and maintain product inventories enhances the feasibility of demand response participation in chemical industries. Beyond single-plant optimization, additional value can be unlocked when multiple plants or stakeholders coordinate their decisions, thereby improving system-wide flexibility and reducing overall cost. Such gains may arise, for example, when geographically separated but operationally similar plants face different locational marginal prices (LMPs) for electricity. This spatial price heterogeneity can be exploited through coordinated production shifting, inventory repositioning, and distribution routing while satisfying shared demand. In the industrial gases sector, coordination across plants owned by a single organization has been studied primarily through integrated production-distribution planning frameworks that are closely related to coordinated demand response. For example, \citet{Marchetti2014} proposed a multi-period MILP for simultaneous production and distribution planning in industrial gas supply chains. Their model coordinates operating modes, production rates, inventories, and routing decisions across multiple ASUs under time-varying electricity prices. Although the study is framed as integrated supply-chain optimization rather than explicit electricity-market demand response, it is closely related to coordinated demand response because production and sourcing decisions are jointly adjusted across plants in response to dynamic power prices. More recently, \citet{Neiro2022} developed an integrated production--distribution planning framework for regional industrial gas supply chains with multiple ASUs. Their formulation simultaneously balances production across sites while satisfying customer demand and service requirements at minimum total cost. The model incorporates plant-specific power costs, startup and shutdown decisions, inter-plant argon transfers by rail, customer inventory management, and both short- and long-haul routing decisions. While this work is also framed as enterprise-wide supply-chain coordination rather than explicit demand response, it highlights the broader value of multisite coordination architectures in industrial supply chains. In particular, it shows that fully coordinated production--distribution planning can outperform production decisions taken in isolation.

While the studies reviewed above primarily consider coordination within a single organization, similar system-level benefits may also be achievable through coordinated demand response across plants owned by distinct organizations. In such settings, however, centralized scheduling formulations are often impractical. They require stakeholders to disclose local scheduling models and operational data to a central planner, while also producing large-scale integrated optimization problems. This creates both a confidentiality concern, because process characteristics, cost structures, and production parameters may be proprietary, and a computational concern, because the resulting multi-plant models can be difficult to solve at industrial scale.

Coordinated demand response among independent stakeholders requires a distributed architecture in which each participant retains its local scheduling model and exchanges only the information required to implement the coordination algorithm. Recent studies have begun to develop such formulations for self-interested and operationally coupled stakeholders. In particular, \citet{Allman2020} studied cooperative industrial demand response between an energy-intensive producer and its downstream customers. In their formulation, customer demand profiles are decision variables that may deviate from reference schedules, and compensation payments are co-optimized so that the producer reduces electricity-driven operating costs while each participating customer improves relative to its status-quo solution. The resulting model treats downstream demand flexibility as an explicit source of system-level value, rather than as an exogenous input to a single-plant demand response problem. \citet{Allman2022} extended this setting to process networks with multiple interconnected stakeholders coupled through material flows, showing that coordinated scheduling can improve system-wide operating outcomes relative to decentralized operation. In both studies, the coordination problem is solved using ADMM, where stakeholders exchange selected coupling variables rather than complete process models.

However, limited model disclosure should not be interpreted as formal privacy protection. In ADMM-based coordination, the sequence of exchanged primal, dual, or coupling-variable iterates is generated by solving local optimization problems whose objectives, constraints, and parameters are private. These iterates may therefore encode information about proprietary feasible regions, operating limits, and cost parameters. Prior work has shown that, under adversarial settings and with suitable side information, sensitive parameters of participants in ADMM-based coordination can be inferred from the communicated iterates \citep{Zhang2019,Dvorkin2020}. This attack surface has not been explicitly assessed or incorporated into prior coordinated demand response studies. Thus, although distributed optimization reduces the need for centralized model sharing, it does not by itself provide a formal privacy guarantee for the information exchanged during coordination.

Taken together, the aforementioned literature shows that coordinated demand response is economically compelling, but that its practical deployment among independent stakeholders requires exchanging information that may reveal sensitive plant data. This paper addresses the challenge by proposing a secure aggregation-based approach for coordinated demand response.


\section{Background}
\label{sec:background}

The core of this work is to mask the ADMM messages exchanged during multi-plant coordination using secure aggregation. We first introduce a motivating example based on the aggregation step performed by the independent central coordinator (ICC) in ICC--ADMM. This example illustrates why the coordinator may need an aggregate quantity for coordination while individual plant-level messages remain commercially sensitive. We then present three secure aggregation mechanisms and define the threat models associated with distributed communication. Finally, we analyze the security properties of these mechanisms under the considered threat models and provide intuition for how secure aggregation can be embedded within ADMM for coordinated demand response.

To illustrate the role of secure aggregation in ICC--ADMM, consider one aggregation step within a coordinated demand response problem involving a set of plants $\mathcal{P}$, with $|\mathcal{P}| = P$. At a given ADMM iteration, each plant $p \in \mathcal{P}$ solves a local scheduling problem and obtains a demand-related coordination message
\[
d_p := (d_p^1,\dots,d_p^T) \in \mathbb{R}^T,
\]
where $d_p^t$ denotes plant $p$'s contribution to the network-level demand balance in time period $t$. The ICC does not need to observe each individual vector $d_p$ separately. Rather, to perform the coordination update, it only requires the aggregate quantity
\begin{equation}
D \;=\; \sum_{p \in \mathcal{P}} d_p \in \mathbb{R}^T .
\label{eq:aggregate_demand}
\end{equation}
This aggregate vector is sufficient for computing the network-level residuals or consistency updates required in the ADMM coordination step, as discussed in subsequent sections. However, the individual message $d_p$ may be commercially sensitive. Across the scheduling horizon, its temporal structure may reveal information about plant-level production schedules, operating limits, outages, ramping behavior, and operational flexibility. Therefore, requiring each plant to disclose $d_p$ directly to the ICC may expose proprietary operational information, even though the ICC only requires the aggregate quantity $D$. This motivates the use of secure aggregation mechanisms, which allow the ICC to recover $D = \sum_{p \in \mathcal{P}} d_p$ exactly while preventing direct observation of any individual plant-level vector. In this way, secure aggregation addresses the mismatch between the aggregate information required for coordination and the proprietary plant-level information that should remain concealed during coordination.

\subsection{Secure aggregation methods}
\label{subsec:secagg_methods}

Secure aggregation refers to a class of cryptographic protocols that enable a collection of agents to compute a prescribed aggregate of their local messages while preventing disclosure of the individual messages themselves \cite[]{bonawitz2017practical}. In other words, each plant transforms its local message into a privacy-protected representation, and the protocol is designed so that only the target aggregate can be recovered, not any single plant's contribution. In the present setting, for example, the coordinator should be able to recover only the desired aggregate (e.g., $\sum_{p\in\mathcal{P}} d_p$) while being unable to recover any individual demand profile. As summarized in Figure~\ref{fig:secaggmethods}, three commonly used families of secure aggregation mechanisms are: (i) neighbor masking, (ii) homomorphic encryption, and (iii) secret-sharing-based protocols.

These three families differ in both cryptographic structure and communication requirements. Since this work integrates secure aggregation into an iterative coordination procedure, the practical communication overhead of repeated aggregation is a central consideration. For this reason, the main text focuses on masking-based secure aggregation, which is the mechanism adopted in our implementation. Homomorphic encryption and secret-sharing-based approaches are introduced briefly here for context, while additional technical details for these two approaches are deferred to~\ref{app:app-sec-methods}.

\begin{figure}[h]
  \centering
  \includegraphics[width=0.9\linewidth]{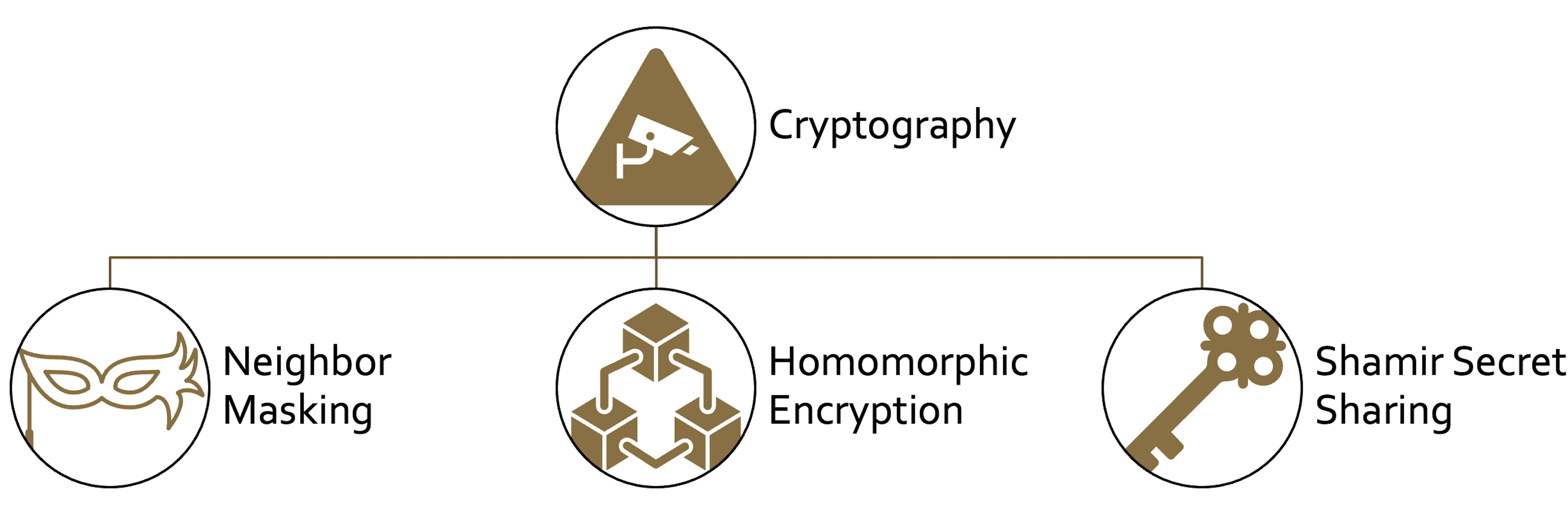}
  \caption{Secure aggregation methods.}
  \label{fig:secaggmethods}
\end{figure}

\paragraph{Masking-based secure aggregation (neighbor masking)}
Masking-based protocols hide each plant's message by adding random masks that cancel out when aggregated. In neighbor masking, pairs of plants establish shared secrets (e.g., via Diffie--Hellman key exchange \cite[]{Diffie1976}) and use these to generate equal-and-opposite masks. Each plant transmits a masked value to the coordinator; when the coordinator sums the masked messages, the pairwise masks cancel, leaving exactly the desired aggregate while the individual terms remain concealed.

We return to the ICC--ADMM aggregation setting introduced above, where the ICC requires only the aggregate demand-related vector $D$ in \eqref{eq:aggregate_demand} and does not require direct access to any individual plant-level vector $d_p$. We describe neighbor masking, a lightweight secure aggregation mechanism widely used in privacy-preserving distributed learning and coordination \citep{bonawitz2017practical}. Neighbor masking proceeds in three conceptual steps: (i) pairwise seed establishment via Diffie--Hellman, (ii) pseudorandom mask generation via a mask generator, and (iii) mask cancellation under summation.

\paragraph{(i) Pairwise seed establishment via Diffie--Hellman}
Diffie--Hellman (DH) key exchange allows two parties to establish a shared secret over a public channel \citep{Diffie1976}. Concretely, for each communicating pair $(p,q)$, plants execute DH to derive a shared seed $s_{pq}=s_{qp}$ that is computationally infeasible for an eavesdropper to compute (assuming standard hardness assumptions underlying DH). This seed is never revealed to the coordinator.

\paragraph{(ii)  Deterministic mask generation via a PRG}
Given a shared seed $s_{pq}$, plants derive a pseudorandom mask vector using a pseudorandom generator (PRG) \citep[]{Goldreich2001}.
A PRG is a deterministic algorithm that expands a random seed into a longer sequence that is computationally indistinguishable from uniform randomness \citep{KatzLindell}. We denote this expansion as
\begin{equation}
m_{pq} \;=\; \mathrm{PRG}\!\left(s_{pq}; T\right) \in \mathbb{R}^T,
\label{eq:prg_mask}
\end{equation}
where $T$ is the time horizon. In practice, one also includes context such as the day index, iteration counter, and plant identifiers in the PRG input to avoid mask reuse across rounds (e.g., by hashing these into the seed); for clarity, we suppress that notation here.

\paragraph{(iii)  Masking rule and cancellation in the aggregate}
Let $\mathcal{N}(p)$ denote the set of plants with which plant $p$ shares pairwise seeds (e.g., $\mathcal{N}(p)=\mathcal{P}\setminus\{p\}$ in the fully connected case). Each plant sends a single masked message to the coordinator:
\begin{equation}
y_p \;=\; d_p
\;+\;\sum_{q\in\mathcal{N}(p):\, p<q} m_{pq}
\;-\;\sum_{q\in\mathcal{N}(p):\, q<p} m_{qp}
\label{eq:neighbor_masking_message}
\end{equation}
The $p<q$ ordering is a notational convention that ensures each pairwise mask appears once with a plus sign and once with a minus sign across the network.
Upon receiving $\{y_p\}_{p\in\mathcal{P}}$, the coordinator computes
\begin{align}
\sum_{p\in\mathcal{P}} y_p
&=\sum_{p\in\mathcal{P}} d_p
+\sum_{p\in\mathcal{P}}\sum_{q\in\mathcal{N}(p):\, p<q} m_{pq}
-\sum_{p\in\mathcal{P}}\sum_{q\in\mathcal{N}(p):\, q<p} m_{qp} \nonumber\\
&=\sum_{p\in\mathcal{P}} d_p \;=\; D,
\label{eq:mask_cancellation}
\end{align}
because each pairwise mask $m_{pq}$ appears exactly once with a $+$ sign (in $y_p$) and exactly once with a $-$ sign (in $y_q$), and therefore cancels in the sum. Consequently, the coordinator recovers the required aggregate $D$ in \eqref{eq:aggregate_demand}, while any individual original message $d_p$ remains masked by terms that are pseudorandom to any party lacking the corresponding pairwise seeds \citep{bonawitz2017practical,KatzLindell}.

\paragraph{Homomorphic encryption}
Homomorphic encryption (HE) enables computation directly on encrypted values. In an additive HE scheme such as the Paillier cryptosystem \citep{Paillier1999}, a public key is used to encrypt local plaintext values, while the corresponding private key is used to decrypt. Thus, each plant can encrypt its local quantity (e.g., $d_p$) using the public key and send only the resulting ciphertext to the coordinator. The coordinator can then combine these ciphertexts algebraically so that the encrypted result corresponds to the sum of the original plaintext values. In this way, the coordinator can form an encrypted aggregate without observing any plant's plaintext demand profile.

A key architectural question is where the decryption capability resides. In the simplest Paillier setting, a trusted party generates the key pair, distributes the public key for encryption, and retains the private key for decryption. After receiving encrypted messages from the plants and aggregating them in ciphertext form, the coordinator sends the encrypted aggregate to this trusted party, which decrypts it and reveals the aggregate plaintext. This approach is conceptually simple, but it places substantial trust in a single entity that can decrypt the aggregate.

A more decentralized alternative is threshold Paillier, in which the decryption capability is split across multiple plants in the network \citep{pallier_threshold_2023}. A trusted setup authority generates the public key for encryption and distributes private-key shares to selected plants. The coordinator can still aggregate ciphertexts, but no single plant can decrypt the aggregate independently. Instead, at least a prescribed threshold $\tau$ of these plants must provide partial decryptions, which are combined to recover the aggregate plaintext. The threshold $\tau$ can be set so that all participating plants are required for decryption, i.e. \(\tau =  |\mathcal P| \). This reduces reliance on a single trusted decryptor, but increases communication and protocol complexity because key setup and decryption require additional interactive steps. The main cryptographic details of encryption, ciphertext aggregation, and threshold decryption are summarized in~\ref{app:app-pallier}.


\paragraph{Shamir secret sharing} Secret sharing provides privacy by splitting a private value into multiple pieces, called shares, that are individually uninformative. In Shamir's secret sharing \citep{Shamir1979}, the basic idea is to represent a scalar secret as the constant term of a randomly generated polynomial. The remaining polynomial coefficients are chosen at random, which hides the scalar within a random polynomial representation. Shares are then created by evaluating this polynomial at a collection of distinct labels assigned to the participants. The key property is that the original secret can be reconstructed only if at least \(k\) shares are collected, whereas any collection of fewer than \(k\) shares reveals no information about it. In the coordinator-mediated implementation considered here, we set the reconstruction threshold to \(k = |\mathcal{P}|\), so that recovering any one plant's secret requires all shares of that plant's polynomial.

In the present setting, each plant in the set \(\mathcal{P}\) is assigned a distinct scalar label, say \(\alpha_p\) for plant \(p \in \mathcal{P}\). Each plant then forms its own random polynomial, with its private scalar value placed in the constant term. When the local message is a vector, such as the demand profile \(q_p\), the same construction is applied one component at a time. Since the reconstruction threshold is taken as \(k = |\mathcal{P}|\), each plant uses a polynomial of degree \(k-1 = |\mathcal{P}|-1\). Each plant evaluates this polynomial at the labels of all plants in the network, including its own label. It retains the share corresponding to its own label and sends the remaining \(|\mathcal{P}|-1\) recipient-labeled shares to the coordinator for forwarding. In this way, a single private scalar is transformed into a collection of shares, with one share generated for each participant.

The first communication step is then as follows. Each plant sends to the coordinator the $|\mathcal{P}|-1$ shares intended for the other plants, together with the corresponding labels needed to identify the intended recipients. The coordinator does not learn the underlying secret from these shares; its role is only to forward each received share to the correct plant. After this redistribution step, every plant holds one retained self-share and one forwarded share from each of the other plants, all evaluated at its own label. Thus, plant $p$ possesses a complete set of shares at $\alpha_p$: one from itself and one from every other plant in the network.

The key observation for secure aggregation is that these received shares can be summed locally. Once plant $p$ sums all shares available at its own label, it obtains a single share of a new polynomial whose constant term equals the sum of all plants' private values. In other words, each plant now holds one share of the aggregate secret rather than a share of any individual secret. Each plant then sends this aggregated share to the coordinator. Because the coordinator receives one aggregated share from each plant, it can reconstruct only the aggregate value through interpolation. At no point does the coordinator obtain enough shares to reconstruct the secret of any individual plant. Hence, the aggregate is revealed, while each plant's private contribution remains hidden. The mathematical details of this procedure are summarized in~\ref{app:app-shamir}.


\subsection{Threats models}
\label{subsec:threat_models}

In this section, we briefly discuss the relevant privacy attacks and then clarify the level of protection that secure aggregation can provide against each of them. Treating plants as independent entities in the network, the adversarial attacks can be categorized as follows \citep{Liu2024}:

\begin{itemize}
    \item \textit{Honest-but-curious plants or central coordinator}: These plants or the central coordinator execute the protocol as prescribed for communication and computation. However, they may log intermediate messages and variables to infer sensitive information about other plants.

    \item \textit{Colluding entities}: A subset of plants may collude with the central coordinator or among themselves to recover private information about a targeted victim plant. For encryption-based schemes with a single trusted decryptor, this threat model also includes possible collusion or compromise of that decryptor, since it holds the full private key. 

    \item \textit{Outside eavesdropper}: This attacker can observe and intercept all exchanged messages during the execution of the protocol, but does not inject false messages or disrupt message delivery.
\end{itemize}

\subsection{Secure aggregation against threat models}
\label{subsec:threats_secagg}

This section summarizes the adversarial robustness and privacy guarantees offered by the secure aggregation primitives discussed in Section \ref{subsec:secagg_methods} under the previously considered threat models. Table~\ref{tab:secagg_threats} provides a concise comparison.

\begin{table}[t]
\centering
\small
\setlength{\tabcolsep}{4pt}
\renewcommand{\arraystretch}{1.15}
\caption{Threat-model coverage of secure aggregation primitives ($P : = |\mathcal P|$; number of plants). Here, $\tau$ denotes the decryption threshold in threshold encryption.}
\label{tab:secagg_threats}

\begin{tabularx}{\linewidth}{p{4.2cm} *{3}{>{\centering\arraybackslash}X}}
\toprule
\textbf{Method} &
\makecell{\textbf{Honest-}\\\textbf{but-curious}\\\textbf{coordinator or}\\\textbf{plants}} &
\makecell{\textbf{Colluding}\\\textbf{entities}} &
\makecell{\textbf{Outside}\\\textbf{eavesdropper}} \\
\midrule
Neighbor masking                           & \cmark & \makecell{Yes\\(up to $P-2$)} & \cmark \\
Homomorphic encryption (trusted dealer setup) & \cmark & \makecell{Yes\\(if decryptor\\remains trusted)}     & \cmark \\
Threshold homomorphic encryption ($\tau = \mathcal P$) & \cmark & \makecell{Yes\\(up to $P-2$)}     & \cmark \\
Shamir secret sharing                      & \cmark & \makecell{Yes\\(up to $P-2$)} & \cmark \\
\bottomrule
\end{tabularx}
\end{table}

\paragraph{Honest-but-curious coordinator or plants}
In the honest-but-curious setting, parties follow the protocol but may log all intermediate messages in an attempt to infer private plant information. Neighbor masking protects privacy because each transmitted message is offset by pseudorandom masks unknown to the coordinator; only the aggregate is revealed due to algebraic cancellation in \eqref{eq:mask_cancellation}. Similarly, Paillier encryption prevents the coordinator from observing plaintext values, since only ciphertexts are received and aggregation is performed in encrypted space. Shamir secret sharing also provides protection in this setting because individual shares are information-theoretically uninformative, and only the reconstructed aggregate is revealed.

\paragraph{Colluding entities}
The collusion setting is more subtle: a subset of plants may share their local information (e.g., masks, seeds, or shares) and potentially cooperate with the coordinator to infer a targeted plant's value. For neighbor masking and Shamir secret sharing, privacy is retained as long as at least two honest plants remain outside the colluding set. Intuitively, even if $|\mathcal P|-2$ plants collude against a victim plant, there remain at least two unknown contributions in the aggregate, and the victim's individual value cannot be isolated from the aggregate alone. Under the same assumption (i.e., fewer than $ |\mathcal P|-1$ parties collude), both protocols can therefore tolerate collusion by up to $|\mathcal P|-2$ plants.

For conventional Paillier with a single trusted decryptor, collusion with the party that holds the private decryption key breaks confidentiality. Once the coordinator obtains decryption capability, it can decrypt individual ciphertexts rather than only the aggregate ciphertext. Thus, the basic Paillier deployment relies critically on the trusted decryptor remaining honest and uncompromised. This limitation can be mitigated using threshold Paillier, where the private key is split across multiple plants and decryption requires partial decryptions from at least $\tau$ key-share-holding plants. If $\tau= \mathcal |\mathcal P|$, no strict subset of plants can decrypt an individual ciphertext on its own, which removes the single trusted-decryptor bottleneck. However, as with all secure aggregation mechanisms, once the aggregate itself is revealed, collusion by all but one plant can still infer the remaining plant's contribution by subtraction.

\paragraph{Outside eavesdropper}
An outside eavesdropper that can observe messages but cannot alter them does not learn any individual plant values when secure aggregation is used. In neighbor masking, intercepted messages remain protected by pairwise masking terms. In Paillier-based aggregation, intercepted ciphertexts do not reveal plaintext because the eavesdropper does not possess the private decryption key, under the standard security assumptions of the scheme. In Shamir secret sharing, intercepted shares do not reveal the underlying secret unless the reconstruction threshold is met; in the coordinator-mediated implementation considered here, all shares of any individual plant's secret are never exposed in transit, so the threshold required for reconstruction is not reached.

\paragraph{Summary}
Overall, all three primitives protect against honest-but-curious behavior and passive eavesdropping when implemented correctly. Their resistance to collusion depends on the trust structure of the protocol. Neighbor masking and Shamir secret sharing tolerate plant-only collusion of up to $|\mathcal P|-2$ plants, which is the strongest possible guarantee once the aggregate itself is revealed. In contrast, Paillier with a single trusted decryptor additionally requires that the private-key holder remain honest and uncompromised. Comparable robustness can be obtained using threshold Paillier when private-key shares are distributed among the plants and the decryption threshold is set to $\tau=|\mathcal P|$, so that no strict subset of plants can decrypt independently. We next provide a brief quantitative comparison of the communication efficiency of these secure aggregation methods.

\subsection{Secure aggregation for iterative coordination}
\label{subsec:sec_agg_for_iter_coord}

In this section, we compare the online communication complexity of the secure aggregation schemes introduced above. Since secure aggregation is invoked repeatedly within an iterative coordination algorithm such as ADMM, the relevant metric is the per-iteration online communication overhead. We quantify this overhead using two simple measures: the number of online communication stages and the total number of message transmissions in one aggregation step. 

Neighbor masking has the lightest online communication pattern. After one-time pairwise seed establishment, each plant generates its local mask and sends one masked message to the coordinator. Thus, each aggregation step requires only one communication stage and a total of $|\mathcal{P}|$ transmissions. By contrast, Paillier-based aggregation requires three online stages: plants first upload encrypted messages to the coordinator, the coordinator then sends the aggregated ciphertext for decryption, and finally the decrypted aggregate is returned. In the trusted-decryptor setting, this gives a total of $|\mathcal{P}|+2$ transmissions per aggregation step. In the threshold variant, the coordinator must instead collect $\tau$ partial decryptions, giving $|\mathcal{P}|+2\tau$ transmissions. Coordinator-mediated Shamir secret sharing also requires three online stages: plants first send recipient-labeled shares to the coordinator, the coordinator forwards them to the intended plants, and each plant finally uploads one aggregate share. This yields a total of $2|\mathcal{P}|(|\mathcal{P}|-1)+|\mathcal{P}|$ transmissions per aggregation step.

Table~\ref{tab:secagg_compare} summarizes these communication patterns. Although all three families protect individual messages, masking-based protocols are the most practical for iterative optimization because they minimize online interaction at each iteration. For the same reason, masking-based secure aggregation is also widely used in large-scale federated learning systems \citep{bonawitz2017practical}. Accordingly, this work adopts neighbor masking as the secure aggregation mechanism for privacy-preserving coordinated demand response among plants.

\newcolumntype{C}{>{\centering\arraybackslash}X}
\begin{table}[t]
\centering
\caption{Online communication complexity of secure aggregation mechanisms for one aggregation step. Here $P := |\mathcal{P}|$, and $\tau$ denotes the decryption threshold in threshold encryption. One-time setup costs are excluded.}
\label{tab:secagg_compare}

\begin{tabularx}{\textwidth}{CCC}
\toprule
\textbf{Method} &
\textbf{Online communication stages} &
\textbf{Total online transmissions} \\
\midrule
Neighbor masking & 1 & $P$ \\
Homomorphic encryption (trusted dealer setup) & 3 & $P + 2$\\
Threshold  homomorphic encryption & 3 & $P + 2\tau$ \\
Shamir secret sharing  & 3 & $2P(P-1) + P$ \\
\bottomrule
\end{tabularx}
\end{table}

\subsection{Intuition behind using neighbor masking in ADMM}
\label{subsec:intuition_sec_agg_admm}

Thus far, secure aggregation has been introduced through a market-facing example, such as joint bidding, to illustrate privacy-preserving information exchange. The focus of this work, however, is coordinated demand response on the industrial side to improve system-wide operational efficiency. As shown by \citet{Allman2022}, coordinated demand response across multiple plants can be implemented through distributed optimization, in which each plant solves a local scheduling problem and shares only limited, scheduling-relevant summaries, such as shipment variables, with a coordinator. In a standard coordinator-based ADMM architecture for a given network demand, plants iteratively solve local subproblems, transmit their shipment-estimate iterates to a central coordinator, and receive updated coordination signals, such as residual shipments and dual multipliers, computed from aggregated information \citep{Boyd2010,Franke2024}. The coordinator's role is therefore primarily to perform the linear aggregation of shipment variables required by the ADMM updates and to broadcast the resulting quantities to all plants. However, repeated exchanges of coordination messages can still leak sensitive information, leaving ADMM vulnerable to inference attacks on private plant parameters. This motivates the central idea of this paper: to mask coordination messages in a way that preserves the aggregation structure required by ADMM. Specifically, we use neighbor masking so that each plant transmits a masked version of its local iterate; the coordinator can still recover the aggregate sum or average needed for the ADMM update because the masks cancel under aggregation, while individual plant iterates remain concealed.

In the remainder of this paper, we formalize this integration by (i) specifying the optimization model and the coordinator-based ADMM updates, and (ii) replacing the coordinator's plaintext aggregation step with a secure aggregation primitive based on neighbor masking. We then study privacy under a worst-case collusion scenario in which $|\mathcal{P}|-1$ plants collude against a victim plant. This case should be viewed as a limitation benchmark rather than the typical operating regime: masking-based secure aggregation remains protected against collusion of up to $|\mathcal{P}|-2$ plants, whereas the $|\mathcal{P}|-1$ case represents the extreme setting in which only one honest plant remains. Although such a scenario is unlikely in practice, it provides a conservative benchmark for quantifying the residual information leakage of the proposed approach over the course of coordination.

\section{Methods}
\label{sec:methods}

\subsection{Problem Setting}
Consider a set of plants $\mathcal{P}$ with $|\mathcal{P}| = P$ that manufacture the same product family $\mathcal{I}$ and ship to a common set of customers. Customers are indifferent to the producing plant as long as orders are fulfilled. At the start of each day, each plant $p\in\mathcal{P}$ receives a day-ahead electricity price profile from the grid and computes a production schedule by solving a local scheduling problem (e.g., a mixed-integer model) that trades off electricity cost, process constraints, and shipment commitments. When plants act independently, each plant can only shift its own load. However, when products are substitutable across plants, there is an additional coordination opportunity: plants can reallocate production and shipments so that more energy-intensive production is executed at plants during hours with lower-carbon or lower-cost electricity, while still meeting customer demand. Such coordinated demand response can reduce operating cost and support decarbonization by exploiting spatial and temporal differences in electricity prices and renewable availability.

We consider coordinated demand response among the set of plants with the help of an independent central coordinator (ICC). Plants are assumed to have the same portfolio of products $\mathcal{I}$ (index $i$) and to ship to the same set of customer regions $\mathcal{R}$ (index $r$). The electricity price profile is taken as an exogenous parameter of the model, and plant-level operating or coordination decisions are not assumed to influence market prices.

We first describe the status-quo optimization solved by each plant independently, and then formulate the ideal social-welfare problem. Since the centralized social-welfare formulation is not directly deployable in our setting, we subsequently solve it in a distributed manner using ICC-coordinated ADMM (ICC--ADMM) with neighbor masking.

\subsection{Status quo optimization}
We assume each plant $p$ is required to satisfy a feasible status-quo demand over its product portfolio $\mathcal{I}$, customer regions $\mathcal{R}$, and shipping epochs $\mathcal{W}$ (index $w$). Customer regions aggregate the demand of local customers near a plant site; this aggregation is reasonable when transportation costs can be represented using an average distance from the plant to the region.

We define the plant-level (status-quo) demand vector as $\tilde d_p \in \mathbb{R}^{m}$, where
\[
m := |\mathcal{I}|\,|\mathcal{R}|\,|\mathcal{W}|.
\]
Hereafter, for any plant $p$ and any variable indexed across multiple sets, e.g., $x_{p i r w}$ with
$i \in \mathcal{I}$, $r \in \mathcal{R}$, and $w \in \mathcal{W}$, we use $x_p$ (the variable name without indices) to denote the stacked vector containing all components:
\[
x_p := \big(x_{p i r w}\big)_{i \in \mathcal{I},\, r \in \mathcal{R},\, w \in \mathcal{W}} .
\]
The status-quo optimization solved independently by each plant is the following mixed-integer program:
\begin{equation}
\begin{aligned}
\min_{x_p}\quad & f_p(x_p) \\
\text{s.t.}\quad & P^{u}_{p}\, x_p = \tilde d_p, \\
& x_p \in \mathcal{F}_p (S_p^{\text{init}}) \cap\left(\mathbb R_+^{r_p}\times\{0,1\}^{b_p}\right)  \\
\label{eq:sq_model}
\end{aligned}
\end{equation}

Here, $x_p$ is a mixed-integer decision vector. The integers $r_p$ and $b_p$ denote, respectively, the number of nonnegative continuous variables, e.g., production quantities, energy consumption, etc., and the number of binary scheduling variables for plant $p$, so that $x_p$ has $r_p + b_p$ components. The plant model is presented in~\ref{app:asumodel} and is adapted from \citet{Zhang2015}.

Let $P^{u}_{p}$ denote the linear projection that extracts the shipping subvector from $x_p$. Define
\begin{equation}
\begin{aligned}
u_p &:= P^{u}_{p}\, x_p, \\
u_p &\in \mathbb{R}_+^{m}.
\end{aligned}
\end{equation}

The feasible set $\mathcal{F}_p$ encodes all plant-side constraints, such as production rates, capacities, material balances, bounds, and inventory dynamics. It also depends on the initial plant state (e.g., existing inventory levels and mode-switch history). For example, if a plant has operational modes such as off, startup, and production, then a plant previously in the off mode may have reduced effective production capability due to startup requirements. We denote the initial state compactly as $S_p^{\text{init}}$. Hence, the feasible set is essentially a function of the initial plant state. We denote the status quo cost of the plant by \(\widetilde{f}_p\). Additionally, we note that parameters such as electricity prices at the plant sites affect only the cost and, as such, do not affect the feasible region of the optimization problem.

\subsection{Social welfare optimization}
In the social-welfare case, plants collectively satisfy the total network demand, i.e., aggregated across all plants and regions. For each plant, we make a copy of the product shipping vector of the plant \( P^{u}_{p}\, x_p\) as \(u_p\), and then using \(u_p\) in the global optimization problem. The centralized social-welfare problem is:
\begin{equation}
\begin{aligned}
\min_{\{x_p \in \mathcal{F}_p\},\,\{u_p\}}\quad & \sum_{p\in\mathcal{P}} f_p(x_p) \\
\text{s.t.}\quad & P^{u}_{p}\, x_p = u_p, \qquad \forall p\in\mathcal{P},\\
& \sum_{p\in\mathcal{P}} u_p = D, \\
& x_p \in \mathcal{F}_p (S_p^{\text{init}}) \cap\left(\mathbb R_+^{r_p}\times\{0,1\}^{b_p}\right)  \qquad \forall p\in\mathcal{P} \\
& u_p \in \mathbb{R}_+^{m}  \qquad \forall p\in\mathcal{P} 
\end{aligned}
\label{eq:icc-consensus}
\end{equation}
where
\[
D := \sum_{p\in\mathcal{P}} \tilde d_p \in \mathbb{R}^{m}.
\]
Since solving the above using an off-the-shelf solver such as Gurobi would require sharing all the plant models. We resort to ICC-governed ADMM with neighbor masking instead to solve the above instance as discussed next. We denote the social welfare cost of the plant by \(\bar{f}_p\).

\subsection{Distributed optimization using secure aggregation}
\label{subsec:distributed-secure-agg}

Solving the social-welfare problem \eqref{eq:icc-consensus} is complicated by the global aggregation constraint on the plants' shipping vectors. To enable a distributed solution, we introduce scaled dual variables $\lambda_p \in \mathbb{R}^{m}$ for the local coupling constraints $P^{u}_{p}x_p = u_p$ and $\nu \in \mathbb{R}^{m}$ for the global balance constraint $\sum_{p}u_p = D$. For a penalty parameter $\rho>0$, the scaled augmented Lagrangian is
\begin{equation}
\mathcal{L}_\rho\bigl(\{x_p\},\{u_p\},\{\lambda_p\},\nu\bigr)
=
\sum_{p\in\mathcal{P}}
\Bigl[
f_p(x_p)
+\tfrac{\rho}{2}\,\bigl\|P^{u}_{p}x_p - u_p + \lambda_p\bigr\|_2^2
\Bigr]
+\tfrac{\rho}{2}\,
\Bigl\|\textstyle\sum_{p\in\mathcal{P}}u_p - D + \nu\Bigr\|_2^2 
\label{eq:icc-al}
\end{equation}

With \eqref{eq:icc-al}, the vanilla ADMM iterations (for $k=0,1,2,\dots$) are:
\begin{subequations}
\begin{align}
x_p^{k+1}
&:= \arg\min_{x_p\in\mathcal{F}_p}
\mathcal{L}_\rho\bigl(\{x_p\},\{u_p^{k}\},\{\lambda_p^{k}\},\nu^{k}\bigr),
\qquad \forall p\in\mathcal{P},
\label{eq:plant-x-update}\\[2mm]
\{u_p^{k+1}\}_{p\in\mathcal{P}}
&:= \arg\min_{\{u_p\}}
\mathcal{L}_\rho\bigl(\{x_p^{k+1}\},\{u_p\},\{\lambda_p^{k}\},\nu^{k}\bigr),
\label{eq:icc-u-update}\\[2mm]
\lambda_p^{k+1}
&:= \lambda_p^{k} + \bigl(P^{u}_{p}x_p^{k+1} - u_p^{k+1}\bigr),
\qquad \forall p\in\mathcal{P},
\label{eq:lambda-update-icc}\\[1mm]
\nu^{k+1}
&:= \nu^{k} + \Bigl(\sum_{p\in\mathcal{P}}u_p^{k+1} - D\Bigr)
\label{eq:nu-update-icc}
\end{align}
\end{subequations}

\paragraph{Residuals and stopping criteria}
Define the primal residuals associated with the two constraint blocks:
\begin{equation}
r_{p}^{k+1} := P^{u}_{p}x_p^{k+1} - u_p^{k+1},\qquad
r_{0}^{k+1} := \sum_{p\in\mathcal{P}}u_p^{k+1} - D
\label{eq:admm_[primal_residuals]}
\end{equation}

We use the aggregated primal residual norm:
\begin{equation}
\|r^{k+1}\|_2
:=
\sqrt{\sum_{p\in\mathcal{P}}\|r_{p}^{k+1}\|_2^2 + \|r_{0}^{k+1}\|_2^2 }
\label{eq:primal_residual}
\end{equation}
For the dual residual, we use the standard ADMM form induced by changes in the consensus variables:
\begin{equation}
\|s^{k+1}\|_2
:=
\rho\,
\sqrt{\sum_{p\in\mathcal{P}}\|u_p^{k+1}-u_p^{k}\|_2^2}
\label{eq:dual_residual}
\end{equation}
The algorithm terminates either upon reaching the maximum iteration limit (e.g., $k=K_{\max}-1$) or when
$\|r^{k+1}\|_2\le \varepsilon_{\mathrm{pri}}$ and $\|s^{k+1}\|_2\le \varepsilon_{\mathrm{dual}}$,
where $\varepsilon_{\mathrm{pri}}$ and $\varepsilon_{\mathrm{dual}}$ are prescribed tolerances. The explicit plant and coordinator subproblems, as well as the derivation of the residual expressions, are provided in~\ref{app:app-admm-subproblems}.

\paragraph{Closed-form central coordinator update}
The coordinator update \eqref{eq:icc-u-update} admits a closed form. Define the shifted local shipments
\[
\bar{z}_p^{\,k}:=P^{u}_{p}x_p^{k+1}+\lambda_p^{k},
\qquad
\bar{z}^{\,k}:=\frac{1}{|\mathcal{P}|}\sum_{p\in\mathcal{P}}\bar{z}_p^{\,k},
\qquad
s^{k}:=D-\nu^{k}
\]
Minimizing \eqref{eq:icc-u-update} over $\{u_p\}_{p\in\mathcal{P}}$ yields, component-wise,
\begin{equation}
u_p^{k+1} \;=\; \bar{z}_p^{\,k} - \delta^{k},
\qquad
\delta^{k}:=\frac{1}{|\mathcal{P}|+1}\Bigl(|\mathcal{P}|\,\bar{z}^{\,k}-s^{k}\Bigr),
\qquad \forall p\in\mathcal{P}
\label{eq:icc-u-closed}
\end{equation}

While \eqref{eq:icc-u-closed} is computationally simple, it still requires the coordinator to (i) aggregate $\{\bar{z}_p^{\,k}\}_{p\in\mathcal{P}}$ to compute $\bar{z}^{\,k}$, and (ii) broadcast $\delta^{k}$ (or equivalently $u_p^{k+1}$) back to all plants. This dependence on aggregated information motivates the use of secure aggregation (via neighbor masking) so that the coordinator can compute the required sums without learning individual plant shipments. Moreover, we can further restructure the iterations so that the coordinator's role becomes purely aggregative by maintaining a local copy of the coordinator dual variable at each plant, which we describe next.

\paragraph{Secure communication}
Algorithm~\ref{alg:icc-admm-masked} is the main working engine of our coordination framework. Its distinctive feature is that all plant-to-ICC vector communications are protected via neighbor masking: every plant message is transmitted in masked form and denoted by $\widetilde{(\cdot)}$. Individually, $\widetilde{(\cdot)}$ has no physical meaning; it becomes useful only through aggregation, where pairwise masks cancel under summation. Consequently, the ICC only ever observes masked vectors and can recover only the aggregates required by the algorithm (e.g., sums/means), not any plant's individual shipment trajectory. Since Algorithm~\ref{alg:icc-admm-masked} is a distributed optimization algorithm, it invokes several supporting routines, including local plant subproblem solves, coordinator-side aggregation steps, and ADMM update rules. The auxiliary algorithms called within Algorithm~\ref{alg:icc-admm-masked} are provided in~\ref{app:app-icc-algorithms}. We now provide a step-by-step description of Algorithm~\ref{alg:icc-admm-masked}, beginning with its objective and incumbent upper-bound state.

\setcounter{topnumber}{5}
\setcounter{bottomnumber}{5}
\setcounter{totalnumber}{10}
\renewcommand{\topfraction}{0.95}
\renewcommand{\bottomfraction}{0.95}
\renewcommand{\textfraction}{0.05}
\renewcommand{\floatpagefraction}{0.85}

\FloatBarrier
\SetAlFnt{\small}
\SetAlgoSkip{smallskip}

\begin{algorithm}[!htbp]
\caption{\textsc{ICC--ADMM}}
\label{alg:icc-admm-masked}
\DontPrintSemicolon
\SetKwInOut{Input}{Inputs}
\SetKwInOut{Output}{Outputs}

\Input{$D\in\mathbb{R}^{m}$; damping $\beta\in[0,1)$;  $(\tau,\theta, \rho^f)$; $(\varepsilon_{\mathrm{pri}},\varepsilon_{\mathrm{dual}})$; $K_{\max}$;
$\{\rho^{0}_p>0, \{ u_p^{0},\lambda_p^{0}, \nu_p^{0}, z_p^{*,0}\}=\mathbf{0}_m, \tilde{f}_p\}_{p\in\mathcal{P}}$;\\
}

\Output{UB information $ \{ J^\star,h^\star, \texttt{feas}\star, k\star \}$, $\{J_p^\star,x_p^\star, , S_p^{\text{new}} \}_{p\in\mathcal{P}} $.}

\For{$k=0,1,2,\dots,K_{\max}-1$}{
    \tcp{\emph{Local plant solves (in parallel)  cf.\ \eqref{eq:plant-x-update} }}
    \ForPar{$p\in\mathcal{P}$}{
        $(\widetilde{z}_p^{k+1},\widetilde{\overline{z}}_p^{k}) 
\leftarrow \textsc{LocalPlantADMMSolve}(p,k)$\;
  SendToICC($\widetilde{{z_p}}^{k+1}, \widetilde{\overline{z_p}}^{k}$)
    }

    $\bar z^{\,k} \leftarrow \frac{1}{|\mathcal{P}|}\sum_{p\in\mathcal{P}} \widetilde{\overline z}_p^{k}$ \ \ \ \ \ \ \ \  \tcp{\emph{ICC aggregation (masks cancel)}}
    $\{\epsilon_p^{(h,k)}\}_{p\in\mathcal{P},\,h\in\mathcal{H}} \leftarrow \textsc{ICCUBHeuristics}(D,\{\widetilde{z}_p^{k+1}\}_{p\in\mathcal{P}}, k) $\;
    BroadcastToPlants$(\bar z^{\,k}, \{\epsilon_p^{(h,k)}\}_{p\in\mathcal{P},\,h\in\mathcal{H}} ) $\;

    \tcp{\emph{Plants perform \eqref{eq:icc-u-update} locally and run Heuristics}}
    \ForPar{$p\in\mathcal{P}$}{
    $(\widetilde u_p^{k+1},\{(J_{p}^{(h,k)},\texttt{feas}_{p}^{(h,k)})\}_{h\in\mathcal{H}})
    \leftarrow
    \textsc{PlantCoordinatorUpdate\&UBEval}(p,k,\bar z^{\,k},D,\{\epsilon_p^{(h,k)}\}_{h\in\mathcal{H}})$

    SendToICC$\bigl(\widetilde u_p^{k+1},\{(J_{p}^{(h,k)},\texttt{feas}_{p}^{(h,k)})\}_{h\in\mathcal{H}}\bigr)$\;
}

    $\bar u^{\,k+1} \leftarrow \frac{1}{|\mathcal{P}|}\sum_{p\in\mathcal{P}} \widetilde u_p^{k+1}$ \ \ \ \ \tcp{\emph{ICC aggregation (masks cancel)}}

    $(J^\star, h^\star,k^{\star},\texttt{feas}^\star) \leftarrow
    \textsc{ICCUBEval}\bigl(k, \{(J_{p}^{(h,k)},\texttt{feas}_{p}^{(h,k)})\}_{p\in\mathcal{P},\,h\in\mathcal{H}}, \{ \tilde{f}_p\}_{p\in\mathcal{P}} \bigr)$\;
    BroadcastToPlants$(h^\star,k^{\star},\texttt{feas}^\star,\bar u^{\,k+1})$\;
    
    \tcp{\emph{Update state, residual and dual ascent (at plants)}}
    \ForPar{$p\in\mathcal{P}$}{
        $\bigr( (S_p^{\text{new}}, x_p^{\star}, J_p^{\star}), (\lambda_p^{k+1}, \nu_p^{k+1}), (\alpha_p^{k+1},\gamma_p^{k+1}) \bigl)
        \leftarrow
        \textsc{PlantUpdateState\&Dual\&Residuals}(p,k,\bar u^{\,k+1},h^\star,\texttt{feas}^\star, k^\star)$
        SendToICC$\bigl(\alpha_p^{k+1},\gamma_p^{k+1}\bigr)$\;
    }

    \tcp{\emph{ICC forms global residual norms}}
    $\|r^{k+1}\|_2 \leftarrow \sqrt{ \| (|\mathcal{P}|\,\bar u^{\,k+1}-D) \|_2^2 +   \sum_{p\in\mathcal{P}}\alpha_p^{k+1} }$; \;
    $\|s^{k+1}\|_2 \leftarrow \sqrt{\sum_{p\in\mathcal{P}}\gamma_p^{k+1}}$\;
 
    BroadcastToPlants$(\|r^{k+1}\|_2,\|s^{k+1}\|_2)$\;
    \ForPar{$p\in\mathcal{P}$}{

        $(\rho_p^{k+1}, \lambda_p^{k+1}, \nu_p^{k+1})
        \leftarrow
        \textsc{PlantRhoUpdateAndDualRescale}(p,k,\|r^{k+1}\|_2,\|s^{k+1}\|_2)$
    }

    \If{$\|r^{k+1}\|_2 \le \varepsilon_{\mathrm{pri}}$ \textbf{and} $\|s^{k+1}\|_2 \le \varepsilon_{\mathrm{dual}} \ \texttt{and}  \ J^{\star} \leq \sum_p \tilde{f}_p $}{
        \textbf{break}
    }    
}
\end{algorithm}

\subsection{Step-by-step description of Algorithm~\ref{alg:icc-admm-masked}}
\label{subsec:icc-admm-algo-descp}

The goal of Algorithm~\ref{alg:icc-admm-masked} is to identify a deployable feasible operating point. This corresponds to plant schedules and shipment decisions that satisfy network demand and plant-side constraints while achieving a total system cost strictly below the network status quo. The total system cost is tracked through an incumbent network upper-bound (UB) state $ (J^\star,\,h^\star,\,k^\star,\,\texttt{feas}^\star),$ where $J^\star$ denotes the best feasible network cost found so far, $h^\star$ is the heuristic that generated it, and $k^\star$ is the iteration at which it was attained. The flag $\texttt{feas}^\star\in\{0,1\}$ indicates whether a network-feasible, status-quo-improving UB has been found. In parallel, each plant maintains a consistent local UB record $ \{ \bigl(J_p^\star,\,x_p^\star,\,S_p^{\text{new}}\bigr) \}_{\ p\in\mathcal{P}},$ corresponding to the incumbent schedule. Both the ICC and plants initialize their UB state to the status quo scenario $ \tilde{f}_p$ (equivalently $J^\star=\sum_p \tilde{f}_p$).

A key safeguard is that the ICC updates $J^\star$ only if the best UB at iteration $k$ is (i) feasible across all plants and (ii) improves the network status quo $\sum_{p\in\mathcal{P}}\tilde f_p$. Thus, any committed incumbent is guaranteed to be system-improving relative to baseline.

\paragraph{Initialization and broadcast}
At initialization, the plants broadcast their current status quo demands \( \tilde d_p \in \mathbb{R}^m  \) and cost $ \tilde{f}_p$ to the ICC. The ICC broadcasts the total network demand $D$ and all required hyperparameters, including damping $\beta$, residual-balancing parameters $(\tau,\theta)$, penalty update frequency $\rho^{f}$, and initial penalties $\{\rho_p^0\}_{p\in\mathcal{P}}$. Plants are instructed to initialize their primal and dual iterates consistently as $\{u_p^{0},\lambda_p^{0},\nu_p^{0},z_p^{\star,0}\}_{p\in\mathcal{P}}=\mathbf{0}_m$. Neighbor-masking seeds are established once using Diffie--Hellman mechanism by the plants, and used to generate iteration-indexed masks with a PRG.

\paragraph{Iterative structure and communication overview}
Algorithm~\ref{alg:icc-admm-masked} then repeats the following steps for $k=0,1,\dots$:

\begin{enumerate}[leftmargin=*, label=(\roman*)]
\item \textbf{Local plant ADMM step (masked reporting):}
Based on its current internal state \(S_p^{\text{init}}\) from the previous day, each plant first solves its local MIQP subproblem, defined in Equation~\ref{eq:plant-x-update}, to obtain \(x_p^{k+1}\). We define the extracted shipment vector \(P_p^u x_p^{k+1}\) as \(z_p^{k+1}\). This shipment vector may be optionally damped using the previous iterate to obtain \(z_p^{\star,k+1}\). The plant then forms the shifted shipment
\[
\overline z_p^{k}=z_p^{\star,k+1}+\lambda_p^{k}.
\]
This shifted shipment is sent to the ICC because the ICC must aggregate it across plants and return the aggregate quantity needed for each plant to perform the coordinator update locally (\( u_p^{k+1}\)), as defined in Equation~\ref{eq:icc-u-closed}. This local implementation is possible because each plant maintains its own local copy of the ICC dual variable \(\nu_p^k\), which is also updated locally in subsequent steps. Finally, the plant generates the neighbor mask \(mask_p^k\) and transmits only the following masked vectors to the ICC:
\begin{equation}
\widetilde z_p^{k+1}=z_p^{\star,k+1}+mask_p^k,\qquad
\widetilde{\overline z}_p^{k}=\overline z_p^{k}+mask_p^k.
\label{eq:alg-exp-1}
\end{equation}
These are the only shipment-related vectors revealed to the ICC, and the masking preserves privacy while still allowing aggregate quantities to be recovered. Algorithm~\ref{alg:local-plant-solve} in~\ref{app:app-icc-algorithms} specifies the plant-side ADMM step, including the local MIQP solve, damping hyperparameter, neighbor-mask generation, and the masked shipment vectors revealed to the ICC. This step reflects Lines 2--3 of Algorithm~\ref{alg:icc-admm-masked}.

\item \textbf{Purely aggregative ICC step:}
After receiving \( \widetilde z_p^{k+1} \) and \( \widetilde{\overline z}_p^{k} \) from all plants, the ICC performs only aggregate-level computations, marked by Lines 5--6 of Algorithm~\ref{alg:icc-admm-masked}. First, using secure aggregation, it computes
\begin{equation}
\bar z^{\,k}=\frac{1}{|\mathcal{P}|}\sum_{p\in\mathcal{P}}\widetilde{\overline z}_p^{k},
\label{eq:alg-exp-2}
\end{equation}
which equals the plaintext mean of $\{\overline z_p^{k}\}_{p\in\mathcal{P}}$ because the neighbor masks cancel in aggregate, i.e., $\sum_{p\in\mathcal{P}} mask_p^k=0$. The ICC broadcasts $\bar z^{\,k}$ back to the plants, where it is used to carry out the local coordinator update as also mentioned in step~(i).

The ICC also computes the aggregate shipment mismatch
\begin{equation}
\epsilon^k =\Bigl(\sum_{p\in\mathcal{P}}\widetilde z_p^{k+1}\Bigr)-D,
\label{eq:alg-exp-3}
\end{equation}
where the summation again reveals only an aggregate quantity. Based on this mismatch, the ICC instantiates a set $\mathcal H$ of aggregate-only upper-bound heuristics using Algorithm~\ref{alg:icc-ub-aggregate} given in~\ref{app:app-icc-algorithms}. Each heuristic $h\in\mathcal H$ generates a candidate offset allocation
\(\{\epsilon_p^{(h,k)}\}_{p\in\mathcal{P}}\) satisfying \(\sum_{p\in\mathcal{P}}\epsilon_p^{(h,k)}=\epsilon^k \). Thus, if all plants apply their assigned offsets, the aggregate shipment is corrected to satisfy the network demand exactly. For a given heuristic $h$, the vector $\epsilon_p^{(h,k)}$ specifies the shipment perturbation assigned to plant $p$. A simple instance-independent rule is \(h=\textsc{SplitEpsilonEqual}\), which assigns $\epsilon^k /|\mathcal P|$ to every plant, thereby distributing the mismatch uniformly across the network. Moreover, instance-dependent heuristics can also be implemented using public problem-specific information, such as electricity prices at the plant sites, to guide the split. A detailed discussion of the upper-bound heuristics used in this work is provided in~\ref{app:app-ub-heuristics}. Importantly, these candidate perturbations are demand-correcting at the aggregate level, but they are not guaranteed to be feasible for each plant's local scheduling model. Therefore, each plant next evaluates the local feasibility and cost of its assigned candidate perturbations, as described below.

\item \textbf{Local coordinator update and UB evaluation at plants:}

This step reflects Lines 8--9 of Algorithm~\ref{alg:icc-admm-masked}. After receiving $\bar z^{\,k}$ and the candidate offset allocations $\{\epsilon_p^{(h,k)}\}_{h\in\mathcal H}$ for the current iteration $k$, each plant $p$ performs two operations. First, it carries out the coordinator update locally using $\bar z^{\,k}$ and its local copy $\nu_p^{k}$ of the coordinator dual:
\begin{equation}
\delta_p^{k}=
\frac{1}{|\mathcal{P}|+1}
\Bigl(|\mathcal{P}|\,\bar z^{\,k}-(D-\nu_p^{k})\Bigr),
\qquad
u_p^{k+1}=\overline z_p^{k}-\delta_p^{k}.
\label{eq:alg-exp-4}
\end{equation}
Equation~\eqref{eq:alg-exp-4} is the local-copy implementation of the coordinator update in \eqref{eq:icc-u-closed}. The distinction is that \eqref{eq:icc-u-closed} is written using the coordinator dual $\nu^k$, whereas \eqref{eq:alg-exp-4} uses the plant-local copy $\nu_p^k$. Maintaining these local copies allows the ICC operations to remain purely aggregative. The resulting $u_p^{k+1}$ is needed by plant $p$ in the next local ADMM subproblem solve. A masked version of $u_p^{k+1}$ is also sent to the ICC so that the ICC can compute the aggregate quantity $\bar u^{\,k+1}$ required for the residual evaluation.

Second, for each heuristic $h\in\mathcal H$, plant $p$ evaluates the candidate shipment schedule induced by its assigned offset $\epsilon_p^{(h,k)}$. The candidate shipment is defined as
\begin{equation}
\hat d_p^{(h,k)} = z_p^{\star,k+1} - \epsilon_p^{(h,k)} .
\label{eq:cand_ub_shipment}
\end{equation}
Plant $p$ then solves the local feasibility and cost-evaluation problem in \eqref{eq:sq_model} with $\tilde d_p$ replaced by the candidate shipment $\hat d_p^{(h,k)}$. This step is detailed in Algorithm~\ref{alg:plant-upperbound} presented in~\ref{app:app-icc-algorithms}. For each heuristic $h$, the plant obtains
\[
\bigl(J_{p}^{(h,k)}, \texttt{feas}_{p}^{(h,k)}, S_p^{\text{new},(h,k)}\bigr),
\]
where $J_{p}^{(h,k)}$ is the local production-and-shipping cost, $\texttt{feas}_{p}^{(h,k)}$ indicates whether the candidate shipment is feasible for the plant's scheduling model, and $S_p^{\text{new},(h,k)}$ is the terminal internal state that would result if the candidate schedule were committed. This terminal state is used to initialize the plant's next-day optimization only if the corresponding heuristic and iteration are selected as the incumbent network schedule. The candidate production schedule and terminal state remain local to the plant; the plant returns to the ICC only the feasibility flag and local cost for each heuristic candidate. Algorithm~\ref{alg:plant-coord-ub-eval} summarizes all of the operations performed in this step.

\item \textbf{ICC aggregation and UB selection:}

After receiving the masked coordinator-update vectors $\{\widetilde u_p^{k+1}\}_{p\in\mathcal P}$ and the per-plant upper-bound summaries
$\{(J_p^{(h,k)},\texttt{feas}_p^{(h,k)})\}_{p\in\mathcal P,\,h\in\mathcal H}$, the ICC again performs only aggregate-level operations. First, using secure aggregation, it computes
\begin{equation}
\bar u^{\,k+1}
=
\frac{1}{|\mathcal P|}
\sum_{p\in\mathcal P}\widetilde u_p^{k+1},
\label{eq:alg-exp-5}
\end{equation}
which equals the plaintext mean of $\{u_p^{k+1}\}_{p\in\mathcal P}$ because the neighbor masks cancel in aggregate. This quantity is broadcast back to the plants and is used in the subsequent local dual and residual updates.

The ICC then evaluates the candidate upper bounds using Algorithm~\ref{alg:icc-ub-eval} as shown in~\ref{app:app-icc-algorithms}. For each heuristic $h\in\mathcal H$, the ICC first checks network-wide feasibility by aggregating the feasibility flags across plants:
\begin{equation}
\texttt{feas}^{(h,k)}
=
\prod_{p\in\mathcal P}\texttt{feas}_p^{(h,k)}.
\label{eq:alg-exp-6}
\end{equation}
Thus, a heuristic is globally feasible only if every plant reports that its assigned candidate shipment is locally feasible. For each such heuristic, the corresponding network upper-bound cost is computed as
\begin{equation}
J^{(h,k)}
=
\sum_{p\in\mathcal P} J_p^{(h,k)} ,
\label{eq:alg-exp-7}
\end{equation}
whereas infeasible candidates are assigned $J^{(h,k)}=+\infty$. The ICC then selects the best feasible heuristic at iteration $k$,
\begin{equation}
h^k \in \arg\min_{h\in\mathcal H} J^{(h,k)},
\qquad
J^k = J^{(h^k,k)} .
\label{eq:alg-exp-8}
\end{equation}
As mentioned the global incumbent upper bound is initialized as \(J^\star=\sum_{p\in\mathcal P}\tilde f_p\), with \(\texttt{feas}^\star=0\), indicating that no improving coordinated schedule has yet been found. At iteration \(k\), the ICC compares the best feasible heuristic candidate \(J^k\) against the current incumbent \(J^\star\). If \(J^k<J^\star\), then the incumbent is updated as
\[
J^\star \leftarrow J^k, 
\qquad
h^\star \leftarrow h^k,
\qquad
k^\star \leftarrow k.
\]
Since \(J^\star\) is initialized at the status-quo cost, any such improvement implies that a deployable coordinated schedule with cost below the network status quo has been identified; hence, \(\texttt{feas}^\star\) is set to one. Thus, \((J^\star,h^\star,k^\star,\texttt{feas}^\star)\) records the best improving network schedule found over all ADMM iterations and heuristic candidates. The ICC broadcasts the incumbent label \((h^\star,k^\star)\), the feasibility flag \(\texttt{feas}^\star\), and \(\bar u^{\,k+1}\) to the plants. This step reflects Lines 11--13 of Algorithm~\ref{alg:icc-admm-masked}.

\item \textbf{Plant state, dual, and residual updates:}

After the ICC broadcasts $\bar u^{\,k+1}$ together with the incumbent label $(h^\star,k^\star)$ and feasibility flag $\texttt{feas}^\star$, each plant performs the local state, dual, and residual updates in Algorithm~\ref{alg:plant-dual-state-residual}. First, plant $p$ updates its local copies of the scaled dual variables as defined by Equation~\ref{eq:lambda-update-icc} and~\ref{eq:nu-update-icc}:
\begin{equation}
\nu_p^{k+1}
=
\nu_p^{k}
+
\Bigl(|\mathcal P|\,\bar u^{\,k+1}-D\Bigr),
\qquad
\lambda_p^{k+1}
=
\lambda_p^{k}
+
\Bigl(z_p^{\star,k+1}-u_p^{k+1}\Bigr).
\label{eq:alg-exp-9}
\end{equation}
The update of $\nu_p$ uses only the aggregate demand residual $|\mathcal P|\bar u^{\,k+1}-D$, while the update of $\lambda_p$ enforces local consistency between plant $p$'s shipment iterate and its coordinator-update variable.

The same routine also handles incumbent state commitment. If $\texttt{feas}^\star=1$ and $k^\star=k$, then the current iteration has produced the best improving network upper bound so far. In that case, each plant commits the locally stored candidate associated with the selected heuristic $h^\star$:
\[
S_p^{\mathrm{new}} \leftarrow S_p^{\mathrm{new},(h^\star,k)},
\qquad
x_p^\star \leftarrow x_p^{(h^\star,k)},
\qquad
J_p^\star \leftarrow J_p^{(h^\star,k)}.
\]
If this condition is not satisfied, the plant retains its previous incumbent state. 
Finally, each plant computes scalar residual contributions

\begin{equation}
\alpha_p^{k+1}
=
\bigl\|z_p^{*,k+1}-u_p^{k+1}\bigr\|_2^2,
\qquad
\gamma_p^{k+1}
=
(\rho_p^k)^2
\bigl\|u_p^{k+1}-u_p^{k}\bigr\|_2^2,
\label{eq:alg-exp-10}
\end{equation}
and sends only $(\alpha_p^{k+1},\gamma_p^{k+1})$ to the ICC. The ICC then forms the global residual norms as defined by Equation~\ref{eq:primal_residual} and~\ref{eq:dual_residual}
\begin{equation}
\|r^{k+1}\|_2
=
\sqrt{
\bigl\||\mathcal P|\,\bar u^{\,k+1}-D\bigr\|_2^2
+
\sum_{p\in\mathcal P}\alpha_p^{k+1}},
\qquad
\|s^{k+1}\|_2
=
\sqrt{\sum_{p\in\mathcal P}\gamma_p^{k+1}}.
\label{eq:alg-exp-11}
\end{equation}
These residual norms are broadcast back to the plants, which then update their penalty parameters and rescale the dual variables using Algorithm~\ref{alg:plant-rho-update-rescale} in~\ref{app:app-icc-algorithms}. This step corresponds to Lines 14--18 of Algorithm~\ref{alg:icc-admm-masked}.

\item \textbf{Penalty update and stopping criterion:}

After receiving the global residual norms from the ICC, each plant updates its local ADMM penalty parameter using the residual-balancing rule in Algorithm~\ref{alg:plant-rho-update-rescale} shown in~\ref{app:app-icc-algorithms}. The penalty parameter is updated only every $\rho^f$ iterations. If the primal residual is large relative to the dual residual, i.e., $\|r^{k+1}\|_2>\tau\|s^{k+1}\|_2$, the penalty is increased as
\[
\rho_p^{k+1}=\theta\rho_p^k.
\]
Conversely, if the dual residual is large relative to the primal residual, i.e., $\|s^{k+1}\|_2>\tau\|r^{k+1}\|_2$, the penalty is decreased as
\[
\rho_p^{k+1}=\rho_p^k/\theta.
\]
Otherwise, the penalty remains unchanged. Whenever $\rho_p$ is updated, the scaled dual variables are rescaled by the factor
\[
c_{\rho,p}^{k+1}=\frac{\rho_p^k}{\rho_p^{k+1}},
\]
so that the scaled-dual representation remains consistent after the penalty change:
\[
\lambda_p^{k+1}\leftarrow c_{\rho,p}^{k+1}\lambda_p^{k+1},
\qquad
\nu_p^{k+1}\leftarrow c_{\rho,p}^{k+1}\nu_p^{k+1}.
\]
This update is performed locally by each plant and corresponds to Lines~19--20 of Algorithm~\ref{alg:icc-admm-masked}. The ICC then applies the stopping test in Line~21, terminating only if the primal and dual residual norms satisfy their tolerances and the incumbent upper bound improves upon the network status quo.

\end{enumerate}

To conclude the algorithm description, we again state that neighbor masking ensures ICC learns only the aggregate quantities required for ADMM and UB-heuristic construction (e.g., $\sum_p \widetilde z_p^{k+1}$, $\sum_p \widetilde{\overline z}_p^{k}$, and $\sum_p \widetilde u_p^{k+1}$). Under the standard threat model for neighbor masking, this provides confidentiality against an honest-but-curious coordinator, external eavesdropper and against collusion of up to $|\mathcal{P}|-2$ plants (i.e., privacy holds as long as at least two non-colluding plants remain). Finally, if no status-quo-improving upper-bound cost is derived for the system, ICC fetches new hyperparameters, and the algorithm is implemented again. Hyperparameter selection and tuning remain instance-dependent and difficult to generalize as reported in the literature \citep[]{non_convex_admm_one,pmlr-v54-xu17a}.

\subsection{Revenue sharing and fair payoff allocation}
\label{sec:rev-sharing}

Let $h^\star$ denote the globally selected upper-bound (UB) heuristic returned by
\textsc{ICCUBEval} (Algorithm~\ref{alg:icc-ub-eval}). Each plant $p\in\mathcal{P}$ obtains a feasible
plant-side operating point under $h^\star$ and reports its corresponding UB cost $J_p^\star$.
The associated network UB cost is
\begin{equation}
J^\star := \sum_{p\in\mathcal{P}} J_p^\star.
\label{eq:Jstar-net}
\end{equation}
Although $J^\star$ is designed to be lower than the network status-quo cost $\sum_{p\in\mathcal{P}}\tilde f_p$,
it may still occur that an individual plant experiences a higher cost than its status quo, i.e.,
$J_p^\star > \tilde f_p$ for some $p$. Without compensation, such a plant has no incentive to remain in coordination. 
We therefore introduce a revenue-sharing transfer vector 
$z\in\mathbb{R}^{|\mathcal{P}|}$, where $z_p$ denotes the monetary transfer 
credited to plant $p$; a negative value of $z_p$ corresponds to a payment made by 
plant $p$, as in \citet{Allman2022}. We impose
budget balance,
\begin{equation}
\sum_{p\in\mathcal{P}} z_p = 0,
\label{eq:budget-balance-main}
\end{equation}
and define plant $p$'s realized cost after transfers as $J_p^\star - z_p$.
A transfer rule is called feasible if it guarantees individual rationality for all plants:
\begin{equation}
J_p^\star - z_p \le \tilde f_p,\qquad \forall p\in\mathcal{P}.
\label{eq:IR-cost}
\end{equation}
We show that whenever the coordination produces positive total savings denoted as \(TG \), i.e.,
\begin{equation}
TG \;:=\; \sum_{p\in\mathcal{P}}\tilde f_p \;-\; J^\star \;>\; 0,
\label{eq:TG-main}
\end{equation}
there exists at least one feasible transfer vector satisfying \eqref{eq:budget-balance-main}--\eqref{eq:IR-cost}.
See~\ref{app:app-feasible-transfers} for a constructive proof.
Having established feasibility, we next address fairness in the distribution of the benefits of coordination.
Following the two-phase approach of \citet{Allman2022}, we (i) quantify the total savings $TG$ achieved through coordination
and (ii) distribute these savings using a cooperative game-theoretic solution concept, namely Nash bargaining.
Nash bargaining determines transfers by maximizing the product of players' surpluses over the disagreement point,
subject to budget balance. Under Nash bargaining, each plant receives an equal surplus, i.e.,

\begin{equation}
\Delta_p \;=\; \tilde f_p - (J_p^\star - z_p^\star) \;=\; \frac{TG}{|\mathcal{P}|},
\qquad \forall p\in\mathcal{P}.
\label{eq:icc-compensation}
\end{equation}
Equivalently, Nash bargaining yields the closed-form transfer
\begin{equation}
z_p^\star
=
\bigl(J_p^\star-\tilde f_p\bigr)
+\frac{TG}{|\mathcal{P}|},
\qquad \forall p\in\mathcal{P},
\label{eq:zstar-main}
\end{equation}
which is budget-balanced and satisfies individual rationality whenever $TG>0$.
The derivation is provided in~\ref{app:nash-bargaining}.

\section{Case study and results}
\label{sec:case-study-and-results}

\begin{figure}[t]
  \centering
  \includegraphics[width=0.9\linewidth]{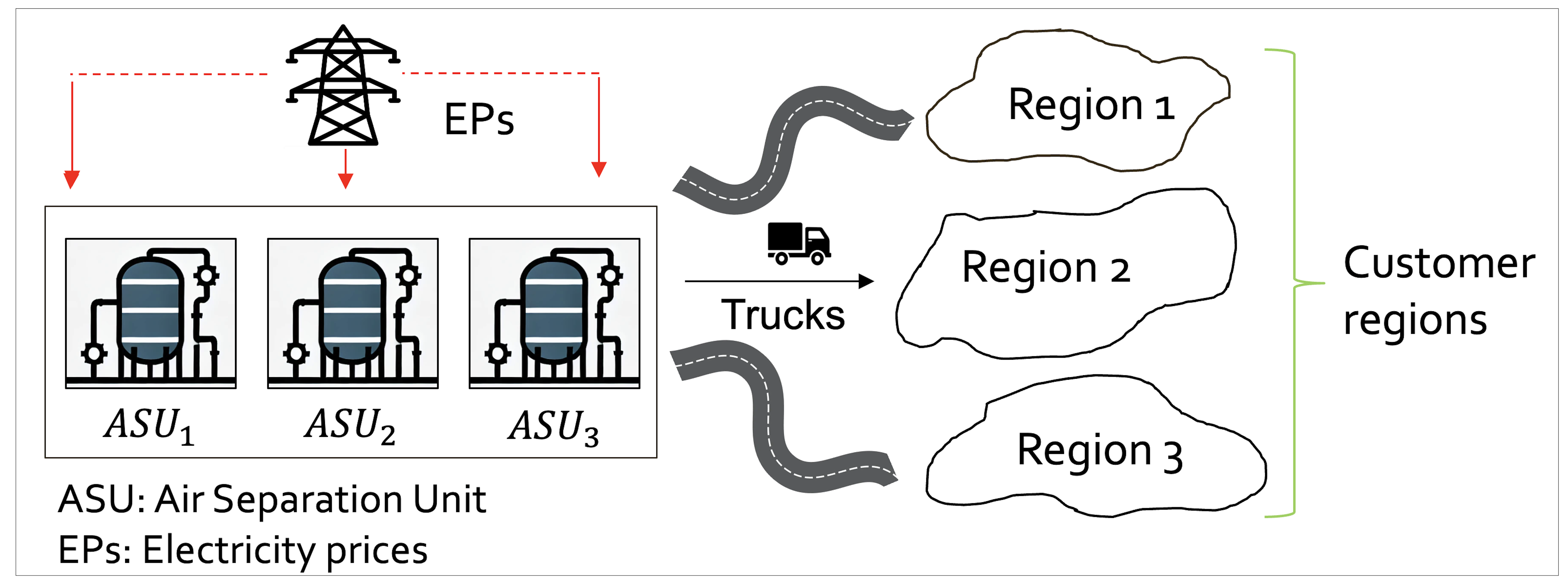}
  \caption{Map of industrial network.}
  \label{fig:networkmap}
\end{figure}

We consider a multi-stakeholder industrial gas network comprising $|\mathcal{P}|=3$ air-separation units (ASUs), denoted $\{\mathrm{ASU}_1,\mathrm{ASU}_2,\mathrm{ASU}_3\}$, coordinated by a central coordinator. The plants are geographically separated but operate within the same regional market and serve a common set of customer regions. Figure~\ref{fig:networkmap} summarizes the network topology: each ASU ships product by truck to three aggregated customer regions $\mathcal{R}=\{R_1,R_2,R_3\}$. Following standard practice in industrial gas operations, we model production and distribution decisions jointly~\citep{Marchetti2014}.

Each plant produces the same portfolio of liquid industrial gases,
$\mathcal{I}=\{\mathrm{LIN},\mathrm{LOX},\mathrm{LAR}\}$, where LIN, LOX, and LAR denote liquid nitrogen, liquid oxygen, and liquid argon, respectively. Each plant is operated in one of three scheduling modes,
$\mathcal{M}=\{\textsc{Off}, \textsc{Startup}, \textsc{Production}\}$. The ASU plant model is formulated as a mixed-integer linear program (MILP) adapted from \citet{Zhang2015}, with the full formulation provided in~\ref{app:asumodel}. The appendix details the modeling of operating modes, transitions, production, inventories, and mass and energy balances.

Customer regions are assumed to have daily product demands. Plants schedule their operations at the start of each day on an hourly grid with a rolling-horizon lookahead of one week (168 hours). Shipments are dispatched at the end of each day, i.e., after 24 hours of production. At the start of the next day, the plant state is updated (initial inventories and mode-switch history), and the optimization is re-solved in a rolling fashion.

\subsection{Simulation setup}
\label{sec:simulation-setup}

\begin{figure}[t]
  \centering
  \includegraphics[width=0.9\linewidth]{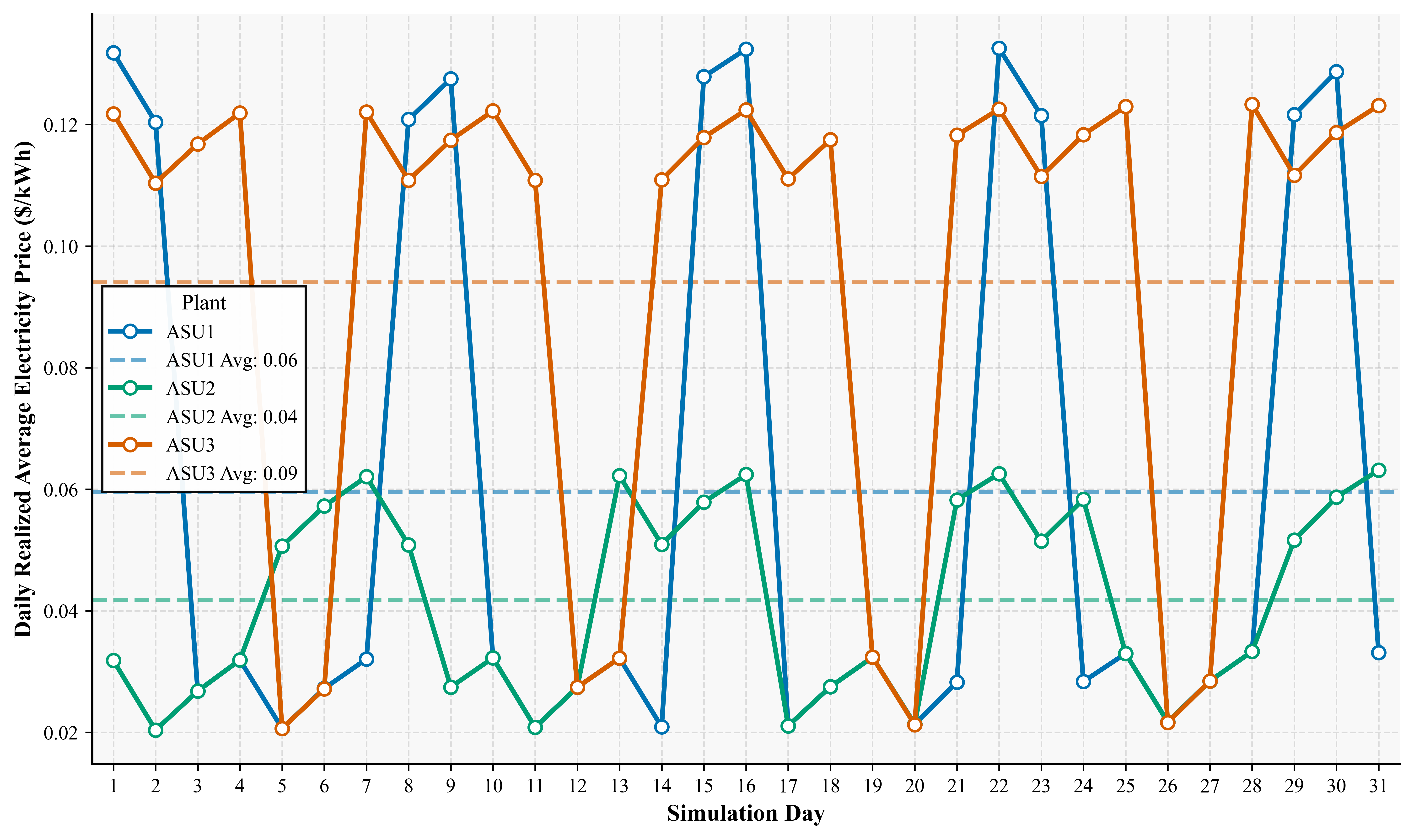}
  \caption{Average daily electricity prices received by plants over the month.}
  \label{fig:elpmonth}
\end{figure}

We simulate 31 consecutive operating days. Each day, plants forecast product demands for the 7-day horizon using historical demand (synthetic in this study). Day-ahead electricity prices are assumed available for the next 24 hours, while prices for the remaining six days are forecasted. The rolling-horizon policy then implements only the first-day decisions and updates the plant state at day end.

We benchmark coordination potential under three scenarios:
\begin{itemize}
    \item \textbf{Status quo (decentralized):} Each plant $p$ independently solves its local scheduling MILP to satisfy its own baseline forecasted demand $\tilde d_p$ (assumed feasible by construction) and incurs cost $\tilde f_p$.
    \item \textbf{Social welfare (centralized benchmark):} A centralized optimizer solves the full network problem to meet total forecasted demand $D=\sum_{p\in\mathcal{P}}\tilde d_p$, providing a system-wide lower-cost benchmark. In this case, plant incurs a cost $\bar f_p$.
    \item \textbf{ICC--ADMM (distributed):} Plants coordinate via ICC--ADMM while preserving privacy through secure aggregation. Since the underlying subproblems are MILPs, they are nonconvex, and convergence is not guaranteed. Consequently, in practice, the coordinator may need to explore multiple hyperparameter settings, upper-bound heuristics, or both to obtain feasible schedules with costs lower than the status-quo solution. The coordinated cost for plant $p$ is denoted by $J_p^\star$.
\end{itemize}

All data used in the simulation are synthetic and are intended to benchmark (i) the feasibility of distributed coordination on industrial-scale MILPs and (ii) the magnitude of achievable coordination benefits under heterogeneous electricity price profiles.

\paragraph{Air separation units and electricity prices}
All the air separation units (ASUs) are assumed to be of moderate size, with aggregated capacity (across liquid products) between 300 and 450 tons per day. We consider a case study in which $\mathrm{ASU}_1$ has the highest maximum production capacity, followed by $\mathrm{ASU}_2$ and $\mathrm{ASU}_3$. The average daily demand over the simulation horizon is assumed to be highest for $\mathrm{ASU}_1$, followed by $\mathrm{ASU}_3$, with $\mathrm{ASU}_2$ exhibiting the lowest average demand. Despite having a higher production capacity than $\mathrm{ASU}_3$, $\mathrm{ASU}_2$ is engineered to be underutilized in the simulation. Electricity prices differ across plants to reflect spatial variation in electricity market prices and congestion effects. We construct a case study in which $\mathrm{ASU}_3$ experiences the highest average electricity prices over the simulation horizon (Figure~\ref{fig:elpmonth}). This setting creates a clear incentive for coordination, enabling production and shipment responsibilities to be shifted across plants when feasible.

\begin{figure}[htpb]
  \centering
  \includegraphics[width=0.9\linewidth]{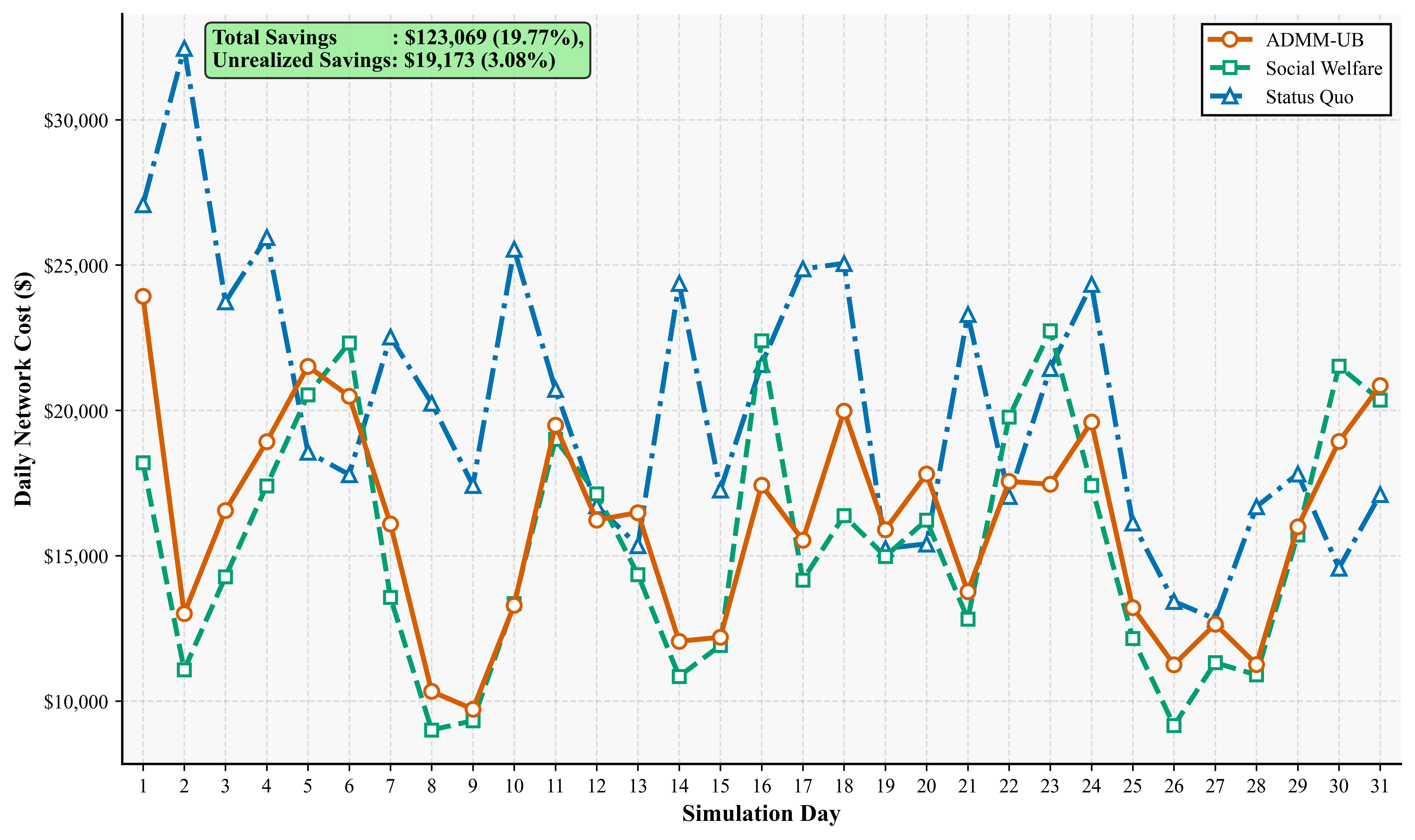}
  \caption{Daily network cost comparison across simulation period.}
  \label{fig:dayonecostcomp}
\end{figure}

\subsection{Simulation benchmarking results}
\label{sec:sim-benchmark}

As mentioned previously, we benchmark three operational trajectories for our simulation, namely, the status-quo trajectory, the coordinated ADMM-UB trajectory, and the centralized social-welfare rolling-horizon trajectory. Figure~\ref{fig:dayonecostcomp} shows these three trajectories over the 31-day rolling-horizon simulation. The ADMM-UB and centralized social-welfare cases are both implemented in closed-loop rolling-horizon manner: on each day, a week-ahead optimization problem is solved, only the first 24 hours of decisions are implemented, and the resulting plant states are propagated to the next day. Since these two policies may implement different first-day schedules, they generally induce different boundary inventories, operating modes, and mode-transition histories. Consequently, the daily optimization instances encountered by the ADMM-UB and centralized social-welfare trajectories need not be identical after the first day. 

Across the full 31-day simulation, the coordinated ADMM-UB strategy reduces the total network cost by $19.77\%$ relative to the status quo. Moreover, its realized full-month cost is within $3.08\%$ of the centralized social-welfare rolling-horizon trajectory, indicating that the distributed coordination policy captures most of the available system-level savings while preserving plant-level privacy.

As described in Section~\ref{sec:simulation-setup}, the upper-bound (UB) heuristics are used to construct a feasible week-ahead operating schedule. Plants then implement only the first 24 hours of decisions and repeat the procedure daily after updating their boundary inventories, operating modes, and mode-switch histories. Hyperparameter selection and other implementation details of Algorithm~\ref{alg:icc-admm-masked} are discussed in~\ref{app:admm-implement-details}.

\begin{table}[htbp]
\centering
\caption{Cost comparison and revenue sharing results (costs in k\$)}
\label{tab:cost_comparison_kusd_transposed}

\setlength{\tabcolsep}{6pt}
\renewcommand{\arraystretch}{1.15}

\begin{tabular}{lrrr}
\toprule
 & \textbf{ASU1} & \textbf{ASU2} & \textbf{ASU3} \\
\midrule
\makecell{SQ Cost (k\$)}                  & 171.44 & 111.30 & 339.83 \\
\makecell{UB Cost (k\$)}                     & 158.22 & 225.30 & 115.99 \\
\makecell{Cost reduction (\%, SQ$\to$UB)} & 7.71 \%   & -102.43\% & 65.87\%   \\
\makecell{Payment (+pay / -receive) (k\$)}   & -27.80 & -155.02& 182.82 \\
\makecell{Cost after revenue sharing (k\$)}  & 130.42 & 70.28  & 298.81 \\
\makecell{Cost reduction (\%, from SQ)\\after revenue sharing} & 23.93 \%  & 36.86  \% & 12.07 \%   \\
\bottomrule
\end{tabular}
\end{table}

\begin{table}[htbp]
\centering
\small
\setlength{\tabcolsep}{6pt}
\renewcommand{\arraystretch}{1.2}
\caption{Aggregated costs over the simulation horizon, summed over all days, for each plant. 
Op. cost and ship. cost denote operational cost and shipping cost, respectively. 
All costs are reported in k\$.}
\label{tab:costs-sq-ub}
\begin{tabular}{lrrrrrr}
\toprule
& \multicolumn{2}{c}{\textbf{Status quo}}
& \multicolumn{2}{c}{\textbf{UB derived}}
& \multicolumn{2}{c}{\textbf{\% reduction}} \\
\cmidrule(lr){2-3}\cmidrule(lr){4-5}\cmidrule(lr){6-7}
& \makecell{\textbf{Op}\\\textbf{cost}}
& \makecell{\textbf{Ship}\\\textbf{cost}}
& \makecell{\textbf{Op}\\\textbf{cost}}
& \makecell{\textbf{Ship}\\\textbf{cost}}
& \makecell{\textbf{Op}\\\textbf{cost}}
& \makecell{\textbf{Ship}\\\textbf{cost}} \\
\midrule
\textbf{ASU1} & 97.02  & 74.42 & 88.74  & 69.47 & 8.53 \%     & 6.65 \%    \\
\textbf{ASU2} & 56.11  & 55.19 & 133.07 & 92.23 & -137.17 \%  & -67.11 \%  \\
\textbf{ASU3} & 277.89 & 61.94 & 85.20  & 30.79 & 69.34 \%    & 50.30 \%  \\
\midrule
\textbf{Network} & 431.02 & 191.55 & 307.02 & 192.49 & 28.77 \%  & -0.49 \%  \\
\bottomrule
\end{tabular}
\end{table}

Table~\ref{tab:cost_comparison_kusd_transposed} reports plant costs over the simulation (costs in k\$), comparing the status quo (SQ) against the UB-derived coordinated schedule, and the resulting costs after revenue sharing. Coordination redistributes operational burden across plants: $\mathrm{ASU}_3$, which faces the highest electricity prices, achieves a large cost decrease (65.87\% from SQ to UB), indicating that a substantial fraction of its production is shifted to other plants. In contrast, $\mathrm{ASU}_2$ experiences a cost increase under the UB schedule (SQ$\to$UB reduction of $-102.43\%$), consistent with $\mathrm{ASU}_2$ being underutilized in the engineered status quo and therefore absorbing additional production and shipment obligations under coordination. $\mathrm{ASU}_1$, which has the largest capacity and already serves high baseline demand by construction, exhibits only a modest decrease (7.71\%), reflecting limited remaining flexibility.

The payment row quantifies transfers under the proposed revenue-sharing mechanism. Despite heterogeneous raw impacts under the UB schedule, the post-sharing costs show that each plant improves relative to its status quo baseline (23.93\%, 36.86\%, and 12.07\% reductions for $\mathrm{ASU}_1$, $\mathrm{ASU}_2$, and $\mathrm{ASU}_3$, respectively), consistent with the individual rationality guarantees discussed in the game-theoretic analysis.

Table~\ref{tab:costs-sq-ub} reports the cumulative operating and shipping costs for each plant over the full simulation horizon, comparing the status-quo policy against the coordinated solution. The plant-level cost shifts observed over the full horizon are consistent with the allocation patterns induced by the best-performing upper-bound heuristics on individual daily instances. In particular, the operating cost decreases substantially for $\mathrm{ASU}_3$ by 69.34\%, while it increases for $\mathrm{ASU}_2$ by 137.17\%. This indicates a coordinated production reallocation away from the high-electricity-price plant and toward the more underutilized plant. For $\mathrm{ASU}_1$, both operating and shipping costs decrease modestly, by 8.53\% and 6.65\%, respectively, suggesting incremental efficiency gains rather than a major structural shift in production allocation.

At the network level, coordination reduces total operating cost by 28.77\%, while total shipping cost is essentially unchanged in this instance (-0.49\%). More broadly, coordination can increase network shipping cost when electricity-price-driven production arbitrage leads to longer delivery routes, trading higher transportation expenditure against lower production cost. The decomposition in Table~\ref{tab:costs-sq-ub} makes this trade-off explicit.

\subsection{ICC-ADMM algorithm performance results}
\label{sec:admm-performance}

Since the benchmarking results in Section~\ref{sec:sim-benchmark} are evaluated along closed-loop rolling-horizon trajectories, they reflect the cumulative effect of repeatedly implementing only the first 24 hours of schedules generated by week-ahead optimization problems. This trajectory-level comparison does not, by itself, isolate the quality of the feasible upper-bound schedules produced by the Algorithm ~\ref{alg:icc-admm-masked} on each daily optimization instance. We therefore separately evaluate the day-wise performance of the upper-bound recovery procedure. For each simulation day, we compare the ADMM-derived feasible week-ahead schedule with a centralized week-ahead optimum solved using the same initial plant states and forecasts. This state-matched comparison quantifies the optimality gap of the ADMM-derived feasible schedule relative to the centralized optimum for the same daily optimization instance. Specifically, at the beginning of each day, we fix the initial plant states to those realized by the ADMM-UB trajectory and solve the corresponding centralized week-ahead social-welfare problem under the same forecasts. The resulting benchmark is used only for day-wise UB gap calculations and is distinct from the centralized social-welfare rolling-horizon trajectory shown in Figure~\ref{fig:dayonecostcomp}.

\begin{table}[htbp]
\centering
\caption{Statistics of the day-wise state-matched upper-bound (UB) gap. For each simulation day, the ADMM-derived feasible UB cost is compared with the centralized social-welfare optimum initialized from the same plant states realized by the ADMM-UB trajectory on that day. This benchmark is used only for day-wise gap evaluation and is distinct from the centralized social-welfare rolling-horizon trajectory shown in Figure~\ref{fig:dayonecostcomp}.}
\label{tab:ub_gap_stats}
\begin{tabular}{r r r r}
\hline
Max gap (\%) & Min gap (\%) & Avg gap (\%) & Std. dev. of gap (p.p.) \\
\hline
19.33 & 4.55 & 8.87 & 3.63 \\
\hline
\end{tabular}
\end{table}

Table~\ref{tab:ub_gap_stats} summarizes the resulting day-wise state-matched UB gap statistics. Over the 31-day simulation, the gap ranges from $4.55\%$ to $19.33\%$, with an average of $8.87\%$ and a standard deviation of $3.63\%$ percentage points. These values indicate moderate day-to-day variability in the quality of the ADMM-derived feasible schedules, driven by changes in ADMM-induced plant states, demand forecasts, and electricity-price forecasts across the rolling-horizon simulation. The full day-wise gap trajectory is provided in~\ref{app:app-state-matched-ub-gap} (Figure~\ref{fig:trueubgap}).

We highlight that the average day-wise state-matched gap is larger than the full-month unrealized savings of $3.08\%$ because the two quantities evaluate different notions of performance. The state-matched gap measures the suboptimality of the ADMM-derived feasible week-ahead schedule for each daily optimization instance, relative to a centralized optimum initialized from the same ADMM-induced plant states. In contrast, the full-month unrealized savings are evaluated along the closed-loop rolling-horizon trajectory, where only the first 24 hours of each week-ahead schedule are implemented before the problem is re-solved with updated plant states and revised forecasts. As a result, suboptimality in later-horizon decisions is not fully realized operationally, since those decisions are replaced by subsequent rolling-horizon optimizations. This receding-horizon feedback effect mitigates the impact of day-wise ADMM suboptimality on the accumulated monthly cost.

\begin{figure}[htpb]
  \centering
  \includegraphics[width=0.9\linewidth]{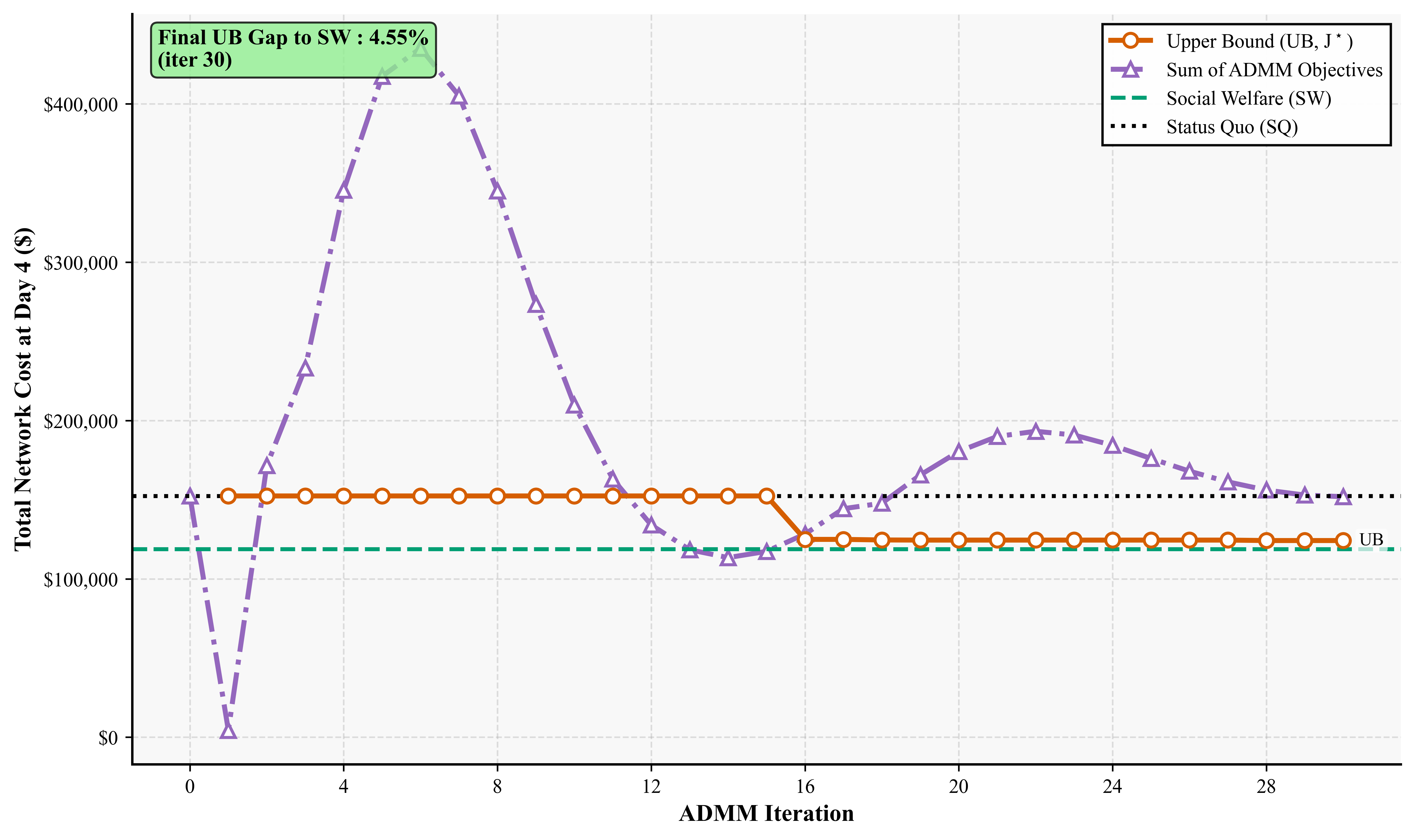}
  \caption{Exemplar ICC-ADMM convergence using Algorithm~\ref{alg:icc-admm-masked} for the network on day 4.}
  \label{fig:dayfourconvg}
\end{figure}

The smallest observed state-matched gap is $4.55\%$, occurring on day~4 with ADMM terminated after 30 iterations. Figure~\ref{fig:dayfourconvg} shows representative convergence behavior of Algorithm~\ref{alg:icc-admm-masked} on this day. The ICC begins to observe nontrivial feasible upper bounds around iteration~16, after the ADMM iterates have stabilized sufficiently for the UB recovery procedure to construct feasible schedules across all plants. The evolution of UB values across heuristic candidates is reported in~\ref{app:app-day4-heuristic-ub-evolution} (Figure~\ref{fig:iccobservation}). As discussed in Subsection~\ref{subsec:icc-admm-algo-descp}, the ICC maintains a portfolio of heuristics used to construct candidate schedules for the plants. A detailed comparison of their performance over the full simulation is provided in~\ref{app:app-heur-performance}.

In the next section, we characterize the leakage of plant-level private information under a worst-case collusion model, where all but one of the (\(|\mathcal{P}|)\) agents in the network are assumed to collude.

\section{Adversarial attack under worst-case collusion}
\label{sec:adv_model}

\begin{figure}[t]
  \centering
  \includegraphics[width=0.9\linewidth]{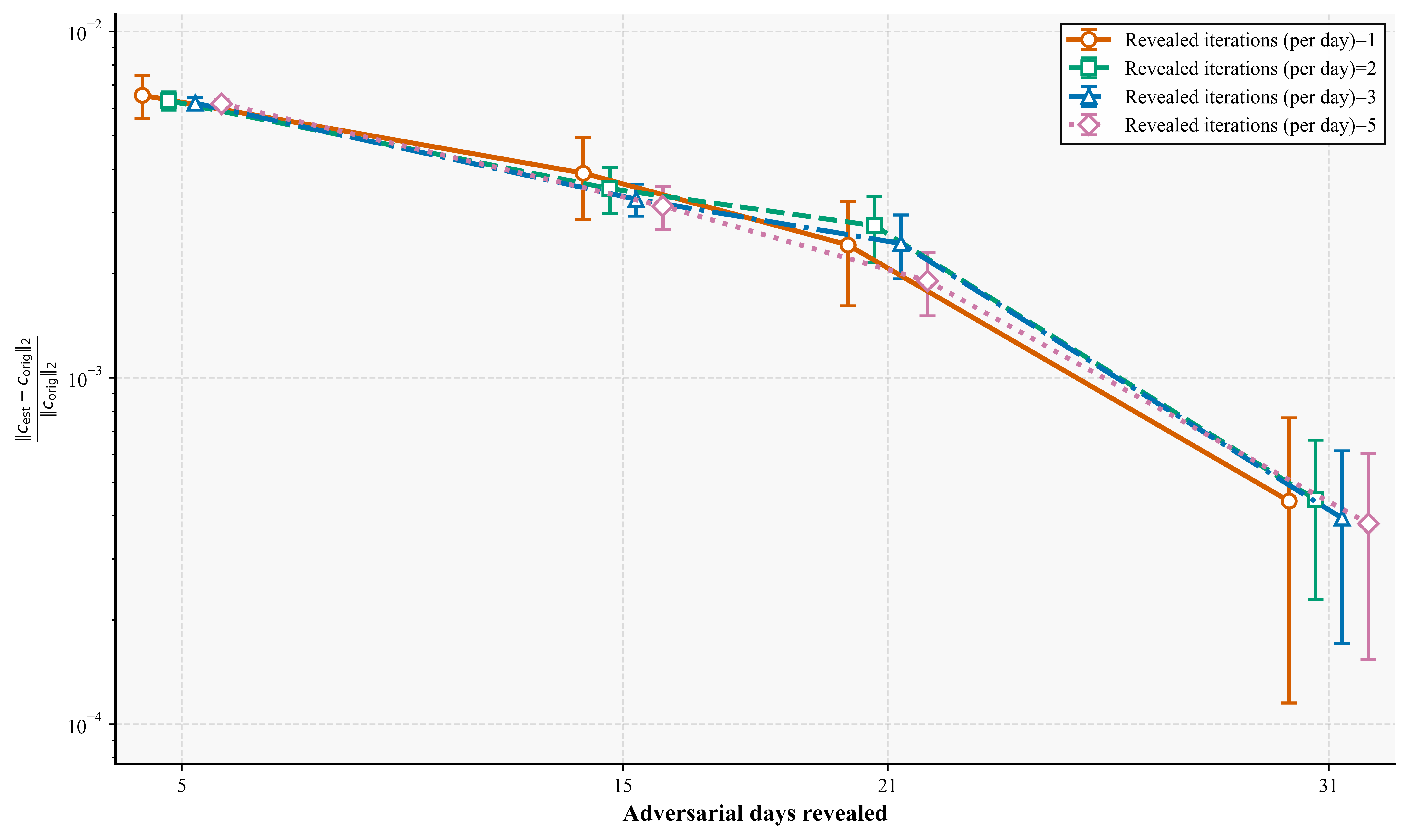}
    \caption{Relative reconstruction error of the adversarial ERM-based recovery attack under randomized ADMM-iterate leakage for $\mathrm{ASU}_1$. For each adversarial-day setting and leakage level, 100 randomized leakage realizations were collected by sampling revealed iterations from the last ten ADMM iterations of each day. Markers denote the sample mean of the relative coefficient-reconstruction error over 100 feasible randomized leakage realizations, and vertical bars denote one sample standard deviation across these realizations. Markers and error bars corresponding to different leakage levels are slightly horizontally offset around each adversarial-day value for visual clarity. The y-axis is shown on a logarithmic scale.}
  \label{fig:advcoeffnorm}
\end{figure}

Algorithm~\ref{alg:icc-admm-masked} protects plant-level messages by neighbor masking: plants transmit masked quantities
$\bigl(\widetilde z_p^{k+1},\,\widetilde{\overline z}_p^{k},\,\widetilde u_p^{k+1}\bigr)$ to the ICC, where masks cancel only under aggregation. In the worst-case threat model, the ICC colludes with all other $|\mathcal{P}|-1$ plants. In this setting, the colluding parties can reconstruct the victim plant's masks and recover its unmasked iterates, revealing
$\bigl(z_p^{k+1},\,\overline z_p^{k},\,u_p^{k+1}\bigr)$ for each iteration and day. These iterates encode plant-specific primal and dual information. With modest side information, the adversary can then formulate a parameter-recovery problem to estimate private objective parameters, such as fixed and variable production costs, thereby exposing sensitive operating margins and process efficiencies.

We instantiate this attack by targeting $\mathrm{ASU}_1$ and demonstrate privacy failure under a worst-case regime with $|\mathcal{P}|-1$ colluding plants. The attack is constructed around the local ADMM solve in Algorithm~\ref{alg:local-plant-solve}, where plant $p$ computes the next local iterate $x_p^{k+1}$ at each ADMM iteration. The adversary observes the revealed coordination signals over the $K_{\max}$ ADMM iterations performed on each simulation day and seeks to infer private objective parameters of the target plant.

The lack of access to the target plant's complete feasible region and full operational trajectories makes the private cost parameters generally non-identifiable from ADMM iterates alone. We therefore pose the adversarial recovery task as a structured empirical risk minimization (ERM) problem. In general, ERM estimates unknown parameters by minimizing an empirical loss over observed samples~\citep[]{Dvorkin2020}. In the present setting, the samples are the leaked day--iteration ADMM transcripts, and the empirical loss penalizes disagreement between the revealed ADMM behavior and the behavior induced by candidate private cost parameters under the adversary's available approximation of the target plant's local optimization model. Thus, the adversary seeks cost coefficients and latent operational trajectories that make the observed ADMM iterates appear as consistent as possible with the revealed coordination signals and the assumed side information.

The assumed side information is introduced only to make this ERM-based adversarial recovery model identifiable and computationally tractable. The mathematical intuition is as follows: given the local iterate sequence implied by the revealed primal and dual coordination signals, the adversary fits objective parameters so that the induced plant-level trajectories match the revealed iterates as closely as possible. Prior works commonly consider two regimes: worst-case revealment, in which the adversary is endowed with extensive auxiliary knowledge, and minimal revealment, in which only limited information is available \citep[]{Dwork2014,Jayaraman2021Revisiting}. To better reflect realistic deployments, where complete side information is unlikely to be disclosed, we adopt the minimal-revealment regime and benchmark privacy under partial disclosure rather than an overly pessimistic worst-case setting. This choice yields a more challenging recovery task, as the adversary must reconstruct hidden plant-level quantities from the smallest set of revealed signals.

To make the ERM-based recovery task identifiable and computationally tractable, we assume the adversary has access to the following side information:
\begin{enumerate}[label=\roman*.]
    \item \textbf{Mode schedule:} the binary mode decisions at each hour in the plant are revealed (e.g., \textsc{Off}, \textsc{Startup}, \textsc{Production}), reducing the adversary's recovery task to a continuous program on each day and ADMM iteration.
    \item \textbf{Boundary inventories:} the initial inventory state of the plant is revealed at the beginning of each simulation day.
    \item \textbf{Bounds on hourly production capacity:} the adversary knows lower and upper bounds on hourly production quantities for each product.
    \item \textbf{Cost ratios:} the adversary knows the ratio between variable-cost and fixed-cost components, as well as the relative proportions of fixed costs across operating modes.
\end{enumerate}

The cost-ratio assumptions are required for identifiability. The revealed iterates are primarily informative about continuous production, and hence variable-cost parameters, whereas fixed-cost terms are only weakly coupled through mode selection. Without additional structural constraints, such as known cost ratios, the ERM formulation can admit degenerate solutions that under-estimate fixed-cost components.

We now clarify which plant-level quantities involved in the local solve remain hidden from the adversary. The adversary does not observe product-wise intra-day production trajectories or intra-day inventory trajectories; it sees only day-boundary inventory levels. The adversary also does not learn the plant's ADMM objective values or any cost breakdowns, with only the ratio information revealed as discussed above. The mathematical details of the adversarial ERM model used for the attack are provided in~\ref{app:app-adv-model}.

\paragraph{Adversarial attack setup}
Adversarial analyses of distributed optimization algorithms often assume that only a limited number of intermediate iterates are leaked from each solve~\citep[]{Dvorkin2020}. Accordingly, we evaluate the ERM-based adversarial recovery model in~\eqref{eq:adv_qp} under randomized leakage of up to five ADMM iterations per simulation day. The adversary is assumed to solve the ERM model after accumulating leaked data over progressively longer exposure windows,
\[
D \in \{5,15,21,31\},
\]
corresponding to short, intermediate, three-week, and approximately month-long adversarial observation horizons. These horizons are chosen to evaluate how reconstruction fidelity evolves as the adversary accumulates information across increasingly diverse day-level operating conditions, rather than to resolve every incremental day. This is particularly relevant because consecutive within-day ADMM iterates are strongly correlated, whereas additional days expose the ERM model to different initial states, forecasts, and operating regimes.

For a given number $D$ of adversarial days, let $\mathcal D_D=\{1,\dots,D\}$. On each day $d\in\mathcal D_D$, the adversary observes a randomly sampled subset $\mathcal A_j^d$ of ADMM iterations from the final ten iterations of that day's coordination transcript, with
\[
|\mathcal A_j^d|=K,
\qquad
\mathcal A_j^d \subseteq \{K_{\max}-9,\dots,K_{\max}\},
\qquad
K\in\{1,2,3,5\}.
\]
Here, $j$ indexes a randomized leakage realization; that is, for each fixed $(D,K)$, the realization $j$ corresponds to one sampled collection
\[
\mathfrak A_{D,K}^{(j)}=\{\mathcal A_j^d\}_{d\in\mathcal D_D}.
\]
This randomized leakage model captures uncertainty in which communication rounds are exposed.

Thus, the experiments vary two disclosure dimensions: the length of the adversarial observation window, measured by the number of attacked days $D$, and the leakage intensity within each day, measured by the number of revealed ADMM iterates $K$. For each $(D,K)$ setting, we collect 100 feasible randomized leakage realizations, indexed by $j=1,\dots,100$. For each realization $j$, the adversary solves the ERM-based recovery model using the revealed iterates indexed by $\mathfrak A_{D,K}^{(j)}$ and returns an estimated coefficient vector $c_{\mathrm{est}}^{(j)}$. We then compute the relative reconstruction error
\[
e_{D,K}^{(j)}
=
\frac{\|c_{\mathrm{est}}^{(j)}-c_{\mathrm{orig}}\|_2}
{\|c_{\mathrm{orig}}\|_2},
\]
where $c_{\mathrm{orig}}$ is the true private objective coefficient vector of $\mathrm{ASU}_1$. 

\paragraph{Attack outcome and interpretation}
Figure~\ref{fig:advcoeffnorm} reports the mean relative reconstruction error $e_{D,K}^{(j)}$ over 100 feasible randomized leakage realizations, with error bars denoting one sample standard deviation. The error is shown as a function of the number of attacked days $D$ and the number of leaked ADMM iterates per day $K$; lower values indicate more accurate reconstruction of the target plant's private objective coefficients. Two observations are salient. First, the number of revealed iterations within a fixed day does not translate linearly into improved identifiability: even when more iterates are disclosed per day, the additional samples are strongly correlated because they arise from the same day-level instance and rolling-horizon state. Second, meaningful improvement is driven primarily by revealing iterates across different days, which correspond to different initial states and forecast realizations and therefore provide more informative excitation for the ERM-based recovery model.

Across all disclosure rates considered, from one to five iterations revealed per day, the adversary requires approximately one month of day-level exposure to reconstruct the private objective coefficients with high accuracy. This indicates that repeated within-day disclosures can be partially redundant, whereas disclosures spanning distinct operating conditions accumulate substantially more information. These results are consistent with the protection limits of neighbor masking: it protects against non-colluding coordinators and limited collusion, but not against the extreme all-but-one collusion regime studied here. Mitigating this residual risk may require mechanisms beyond exact secure aggregation, such as differential privacy (DP), which adds calibrated noise to ADMM messages to limit the informativeness of revealed iterates \citep{Dwork2014}. However, this would introduce a privacy--utility trade-off, since noise may affect convergence, feasibility recovery, and solution quality. Developing DP mechanisms that preserve operational reliability while providing formal privacy guarantees is left for future work.

\section{Conclusion}
\label{sec:conclusion}

We studied coordinated demand response in a multi-plant industrial gas network, where plants jointly satisfy aggregate customer demand while co-optimizing production and shipment decisions under time-varying electricity prices. The setting creates opportunities for system-level cost reduction, but also introduces practical challenges: plants cannot disclose local scheduling models or proprietary operational data, distributed coordination messages may leak private information, and participants require incentives to deviate from their decentralized status-quo schedules.

To address these challenges, we developed a privacy-aware ICC--ADMM coordination framework that combines coordinator-based distributed optimization with secure aggregation based on Diffie--Hellman neighbor masking. The masking scheme enables the ICC to recover only the aggregate quantities required for ADMM updates and upper-bound construction, while concealing individual plant-level messages. Since convergence guarantees for ADMM do not generally hold for the resulting nonconvex mixed-integer subproblems, we use a portfolio of secure-aggregation-compatible upper-bound heuristics to construct feasible network-wide incumbent schedules. We also incorporated a two-phase revenue-sharing mechanism to redistribute realized savings and ensure that each plant improves relative to its status quo.

In a 31-day rolling-horizon simulation of a three-plant industrial gas network, the coordinated policy reduced total network cost by $19.77\%$ relative to decentralized operation and achieved a full-month cost within $3.08\%$ of the centralized social-welfare benchmark. The day-wise state-matched UB gap ranged from $4.55\%$ to $19.33\%$, with an average of $8.87\%$, indicating that the heuristic recovery procedure produced feasible and competitive schedules under rolling-horizon execution. Revenue sharing further ensured net cost reductions for all plants.

Finally, we quantified a residual privacy limitation under a conservative worst-case collusion regime. Neighbor masking protects against an honest-but-curious coordinator, external eavesdroppers, and collusion as long as at least two plants remain non-colluding. However, if all but one plant collude, the remaining plant's unmasked iterates can be reconstructed and used, together with side information, to infer private objective parameters through an empirical data-fitting recovery model. Future work will (i) develop more systematic and less instance-dependent penalty-parameter selection, restart strategies, and stopping rules for nonconvex ICC--ADMM, (ii) broaden secure-aggregation-compatible upper-bound recovery heuristics and explore learning-based heuristic selection, (iii) strengthen privacy under extreme collusion through differential privacy, and (iv) evaluate scalability and robustness in larger industrial networks.

\newpage
\appendix

\section{Homomorphic Encryption and Secret Sharing for Secure Aggregation}
\label{app:app-sec-methods}

This appendix summarizes two classical alternatives to neighbor masking for securely computing the aggregate coordination message required by the ICC. Consistent with the notation in Section~\ref{sec:background}, let
\[
d_p := (d_p^1,\dots,d_p^T)\in\mathbb{R}^T
\]
denote the demand-related coordination message sent by plant $p\in\mathcal{P}$ at a given ADMM iteration. The coordinator does not require the individual plant-level vectors $d_p$; it only needs their aggregate
\begin{equation}
D \;=\; \sum_{p\in\mathcal{P}} d_p \in \mathbb{R}^T .
\label{eq:app_Dsum}
\end{equation}
The goal of secure aggregation is therefore to allow the ICC to recover $D$ exactly, while preventing direct observation of any individual plant message $d_p$.

We focus on two standard cryptographic alternatives to neighbor masking: (i) additive homomorphic encryption via the Paillier cryptosystem and (ii) Shamir secret sharing. These mechanisms differ in their computational cost, communication structure, trust assumptions, and robustness to collusion, but they share the same aggregation objective: the ICC should learn only the aggregate coordination signal \(D\), not the individual plant-level contributions \(d_p\). Throughout this appendix, the aggregation is understood componentwise over the scheduling horizon \(t=1,\dots,T\).

\subsection{Paillier Cryptosystem for Additively Homomorphic Aggregation}
\label{app:app-pallier}

\paragraph{Idea}
Paillier is a public-key cryptosystem that supports additive homomorphism: ciphertexts can be combined so that, after decryption, the plaintext equals the sum of the original plaintexts \citep{Paillier1999}. This property allows the coordinator to compute an encrypted aggregate from encrypted plant messages without observing any individual plaintext load.

\paragraph{Keys and notation}
Let $(\mathsf{pk},\mathsf{sk})$ denote the public and private keys. We write $\mathsf{Enc}_{\mathsf{pk}}(\cdot)$ for encryption and $\mathsf{Dec}_{\mathsf{sk}}(\cdot)$ for decryption. Paillier operates over integers modulo a large composite $n$ (details in \citep{Paillier1999}). Since our values are real-valued (for example, ton or kg), plants typically use a fixed-point encoding, e.g.,
\[
\tilde{d}_p^t = \lfloor s \, d_p^t \rceil \in \mathbb{Z},
\]
for a scaling factor $s>0$ (and optionally a signed encoding). After decryption, the recovered integer sum is rescaled by $1/s$.

It is useful to distinguish the coordinator from the trusted dealer. The coordinator is the entity that receives encrypted messages and combines them homomorphically. The trusted dealer, when assumed, is a separate setup entity whose role is only to generate keys and distribute the required cryptographic material. In the simplest deployment, the trusted dealer may also retain the decryption capability; in threshold variants, it instead distributes shares of the private key and then exits the protocol.

\paragraph{Protocol for secure aggregation of $D$}
For each time index $t\in\{1,\dots,T\}$ (or vectorized across $t$), the protocol proceeds as follows:
\begin{enumerate}[label=\Roman*.]
    \item \textbf{Key setup:} A trusted dealer generates the Paillier key pair $(\mathsf{pk},\mathsf{sk})$. The public key $\mathsf{pk}$ is distributed to all plants for encryption. The decryption capability associated with $\mathsf{sk}$ is either retained by a designated trusted decryptor or split into shares across multiple authorized parties for threshold decryption.

    \item \textbf{Encryption at plants:} Each plant $p$ encrypts its encoded value $\tilde{d}_p^t$ using the public key and sends
    \[
    c_p^t \;=\; \mathsf{Enc}_{\mathsf{pk}}(\tilde{d}_p^t)
    \]
    to the coordinator.

    \item \textbf{Homomorphic aggregation at the coordinator:} The coordinator multiplies ciphertexts using Paillier's group operation to obtain
    \begin{equation}
    c_{\mathrm{agg}}^t \;=\; \prod_{p\in\mathcal{P}} c_p^t,
    \label{eq:app_paillier_agg}
    \end{equation}
    which, by additive homomorphism, is an encryption of the sum:
    \[
    c_{\mathrm{agg}}^t \;=\; \mathsf{Enc}_{\mathsf{pk}}\!\Big(\sum_{p\in\mathcal{P}} \tilde{d}_p^t \Big).
    \]
    \item \textbf{Recovery of the aggregate:}
    \begin{itemize}
        \item \emph{Single decryptor setting:} If a trusted decryptor holds the full private key $\mathsf{sk}$, then the coordinator sends $c_{\mathrm{agg}}^t$ to that party, which computes
        \[
        \mathsf{Dec}_{\mathsf{sk}}(c_{\mathrm{agg}}^t)
        \;=\;
        \sum_{p\in\mathcal{P}} \tilde{d}_p^t,
        \]
        and the result is then rescaled to obtain $D^t$.

        \item \emph{Threshold setting:} If the private key is split across multiple authorized parties, then no single party can decrypt on its own. Instead, each authorized party $j$ uses its private-key share $\mathsf{sk}_j$ to compute a \emph{partial decryption} of $c_{\mathrm{agg}}^t$, which we denote abstractly by
        \[
        q_j^t \;=\; \mathsf{PDec}_{\mathsf{sk}_j}(c_{\mathrm{agg}}^t).
        \]
        Once a prescribed number of such partial decryptions are collected, they are combined to recover the aggregate plaintext:
        \[
        \sum_{p\in\mathcal{P}} \tilde{d}_p^t
        \;=\;
        \mathsf{Combine}(q_1^t,\dots,q_m^t),
        \]
        where $m$ satisfies the threshold requirement of the scheme. The recovered integer sum is then rescaled to obtain $D^t$.
    \end{itemize}
\end{enumerate}

\paragraph{Deployment variants and practical remarks}
The simplest Paillier deployment assumes that a trusted dealer generates the key pair and that a trusted decryptor holds the full private key. This makes the protocol easy to describe and implement, but it concentrates decryption trust in a single entity. For one aggregation of a scalar quantity, the online communication consists of $|\mathcal{P}|$ plant-to-coordinator ciphertext uploads, followed by one coordinator-to-decryptor transmission carrying the aggregate ciphertext, and one decryptor-to-coordinator return message carrying the decrypted aggregate. Thus, ignoring one-time key setup, the online phase requires $|\mathcal{P}|+2$ transmissions.

A more decentralized alternative is threshold Paillier \citep{pallier_threshold_2023}, in which the private key is split across multiple plants in the network. Let $\tau$ denote the number of partial decryptions required to recover the aggregate plaintext. In that case, the coordinator still receives $|\mathcal{P}|$ encrypted messages from the plants and aggregates them exactly as before, but decryption now requires an additional interaction round: the coordinator sends the aggregate ciphertext to the $\tau$ authorized parties, and receives $\tau$ partial decryptions in return. Hence, again ignoring one-time setup, the online phase requires $|\mathcal{P}| + 2\tau$ transmissions. Threshold decryption therefore improves trust distribution, but it also introduces extra communication and synchronization overhead relative to the single-decryptor setting.

\paragraph{Security intuition}
Under standard security assumptions, the ciphertexts $c_p^t$ do not reveal the encoded values $\tilde{d}_p^t$ to the coordinator or to an external observer without decryption capability. The coordinator can nevertheless compute the encrypted aggregate through \eqref{eq:app_paillier_agg}. In the single-decryptor setting, only the trusted decryptor can reveal the aggregate plaintext. In the threshold setting, the coordinator learns the aggregate only after collecting sufficiently many partial decryptions from authorized parties, while no single party can decrypt alone. Hence, Paillier enables recovery of the aggregate without exposing any individual plant's plaintext message to the coordinator.

\paragraph{Technical remark}
The main cryptographic details of Paillier encryption, ciphertext aggregation, and threshold decryption are considerably more technical than the high-level description above. The purpose of this appendix section is only to clarify the message flow and the location of decryption capability in the secure aggregation setting considered in this work. Additional technical details can be found in \citet{Liu2024}.

\subsection{Shamir Secret Sharing for Secure Aggregation}
\label{app:app-shamir}

\paragraph{Idea}
Shamir secret sharing splits a secret into multiple shares such that any subset of at least $k$ shares can reconstruct the secret, while fewer than $k$ reveal no information \citep{Shamir1979}. For aggregation, the key advantage is \emph{linearity}: shares of individual secrets can be added to obtain shares of the sum.

In the coordinator-mediated workflow considered here, secure aggregation is implemented in two stages: (i) one round in which plants send recipient-labeled shares to the coordinator, which then forwards each share to its intended plant, and (ii) one uplink of a single aggregate share from each plant to the coordinator for final reconstruction.

\paragraph{Setup}
Fix the reconstruction threshold as $k=|\mathcal{P}|$ and choose distinct nonzero public evaluation points $\alpha_1,\dots,\alpha_{|\mathcal{P}|}$ in a finite field $\mathbb{F}$ of sufficiently large size.\footnote{As with Paillier, real values can be fixed-point encoded into field elements.} Each plant $p$ is associated with one point $\alpha_p$.

This strongest-threshold choice is consistent with the coordinator-mediated workflow used here: each plant retains its own self-share locally and sends only the remaining $|\mathcal{P}|-1$ recipient-labeled shares to the coordinator for forwarding. As a result, the coordinator never observes all shares of any one plant's polynomial and therefore cannot reconstruct an individual secret during the forwarding stage.

We describe the protocol for a single scalar (e.g., a single hour $t$); it is applied independently for each $t=1,\dots,T$ to aggregate a vector.

\paragraph{Protocol for secure aggregation of $D^t$}
Let $\tilde{d}_p^t\in\mathbb{F}$ denote plant $p$'s encoded secret at hour $t$.
\begin{enumerate}[label=\Roman*.]
    \item \textbf{Local polynomial construction:} Each plant $p$ samples a random polynomial of degree $k-1$:
    \begin{equation}
    f_p(x) \;=\; \tilde{d}_p^t \;+\; a_{p,1}x \;+\; a_{p,2}x^2 \;+\; \cdots \;+\; a_{p,k-1}x^{k-1},
    \label{eq:app_shamir_poly}
    \end{equation}
    where coefficients $a_{p,1},\dots,a_{p,k-1}\in\mathbb{F}$ are chosen uniformly at random. The constant term is the secret $\tilde{d}_p^t$.

    \item \textbf{Coordinator-mediated share redistribution:} Plant $p$ evaluates its polynomial at all public points $\alpha_r$, $r\in\mathcal{P}$. It retains its own local share
    \[
    s_{p\to p} \;=\; f_p(\alpha_p),
    \]
    and sends each remaining share
    \[
    s_{p\to r} \;=\; f_p(\alpha_r), \qquad r\neq p,
    \]
    to the coordinator together with the recipient label $r$. The coordinator then forwards each received share to its intended plant.
    
    \item \textbf{Local share summation (forming a share of the global sum):}
    After receiving forwarded shares $\{s_{p\to r}\}_{p\neq r}$ and combining them with its retained self-share $s_{r\to r}$, plant $r$ computes
    \begin{equation}
    S_r \;=\; \sum_{p\in\mathcal{P}} s_{p\to r}
    \;=\; \sum_{p\in\mathcal{P}} f_p(\alpha_r).
    \label{eq:app_sumshare}
    \end{equation}
    Define the \emph{sum polynomial} $F(x) := \sum_{p\in\mathcal{P}} f_p(x)$. Then \eqref{eq:app_sumshare} is exactly
    \[
    S_r \;=\; F(\alpha_r),
    \]
    so each plant now holds one Shamir share of the aggregate.

    \item \textbf{Upload to coordinator and reconstruction:}
    Each plant $r$ sends its single value $S_r$ to the coordinator. Since $k=|\mathcal{P}|$, the coordinator reconstructs the aggregate only after collecting all values $\{(\alpha_r,S_r)\}_{r\in\mathcal{P}}$. It then interpolates $F(\cdot)$ and recovers the constant term:
    \begin{equation}
    F(0) \;=\; \sum_{p\in\mathcal{P}} \tilde{d}_p^t,
    \label{eq:app_recover_sum}
    \end{equation}
    which is the desired aggregate for hour $t$. Repeating this procedure for $t=1,\dots,T$ yields the full vector $D$.
\end{enumerate}

\paragraph{Qualitative and practical remarks}
After the forwarding stage, each plant already possesses one point $(\alpha_r,S_r)$ on the sum polynomial $F(\cdot)$, and the coordinator can recover $F(0)$ once it has collected all $|\mathcal{P}|$ aggregate shares. Thus, plants need only upload one value $S_r$ each in the second stage. In regard to implementation, the first stage requires $|\mathcal{P}|(|\mathcal{P}|-1)$ plant-to-coordinator transmissions and $|\mathcal{P}|(|\mathcal{P}|-1)$ coordinator-to-plant forwarding transmissions, followed by $|\mathcal{P}|$ uplink messages carrying the aggregate shares $\{S_r\}$. This communication overhead can be substantial when secure aggregation is invoked repeatedly inside an iterative distributed optimization scheme.

\paragraph{Security intuition}
For any fixed plant $p$, fewer than $k$ shares of the polynomial $f_p(\cdot)$ reveal no information about the secret $\tilde{d}_p^t$ because the remaining coefficients in \eqref{eq:app_shamir_poly} are random \citep{Shamir1979}. In the coordinator-mediated implementation above, the coordinator observes only the $|\mathcal{P}|-1$ forwarded shares of each individual polynomial and not the locally retained self-share. Since $k=|\mathcal{P}|$, the coordinator does not have enough shares to reconstruct any individual plant's secret during the forwarding stage. By linearity, the uploaded values $\{S_r\}$ are instead shares of the sum polynomial $F(\cdot)$, whose constant term is the aggregate. Thus, the coordinator can reconstruct only the aggregate once it gathers all aggregate shares, while the individual values $\tilde{d}_p^t$ remain hidden.

\section{ADMM Subproblems and Residuals}
\label{app:app-admm-subproblems}

This appendix provides the explicit plant and coordinator subproblems corresponding to the ADMM iterations in subsection~\ref{subsec:distributed-secure-agg} from Equation \eqref{eq:plant-x-update}--\eqref{eq:nu-update-icc}, along with the residual expressions used for stopping and hyper-parameter tuning.

\subsection{Scaled augmented Lagrangian and ADMM updates}
\label{app:scaled-al-admm}

Recall the scaled augmented Lagrangian \eqref{eq:icc-al}:
\[
\mathcal{L}_\rho\bigl(\{x_p\},\{u_p\},\{\lambda_p\},\nu\bigr)
=
\sum_{p\in\mathcal{P}}
\Bigl[
f_p(x_p)
+\tfrac{\rho}{2}\,\bigl\|P^{u}_{p}x_p - u_p + \lambda_p\bigr\|_2^2
\Bigr]
+\tfrac{\rho}{2}\,
\Bigl\|\textstyle\sum_{p\in\mathcal{P}}u_p - D + \nu\Bigr\|_2^2 
\]

Given iterates $\{u_p^{k}\}$, $\{\lambda_p^{k}\}$, and $\nu^{k}$, the vanilla ADMM iterations (for $k=0,1,2,\dots$) are:

\paragraph{Plant $x$--update}
Each plant $p$ solves
\begin{equation}
\begin{aligned}
x_p^{k+1}
\;:=\;
\arg\min_{x_p\in\mathcal{F}_p}
\Bigl\{
f_p(x_p)
+\tfrac{\rho}{2}\,\bigl\|P^{u}_{p}x_p - u_p^{k} + \lambda_p^{k}\bigr\|_2^2
\Bigr\},
\qquad \forall p\in\mathcal{P}.
\end{aligned}
\label{eq:app-x-update}
\end{equation}
In our setting, \eqref{eq:app-x-update} is a mixed-integer quadratic program (MIQP) due to the plant scheduling binaries in $x_p$ and the quadratic penalty term. It is solved locally by each plant using an off-the-shelf MIP solver.

\paragraph{Coordinator $u$--update}
The ICC updates the shipment copies $\{u_p\}$ by solving
\begin{equation}
\begin{aligned}
\{u_p^{k+1}\}_{p\in\mathcal{P}}
\;:=\;
\arg\min_{\{u_p\}}
\Bigl\{
\tfrac{\rho}{2}\sum_{p\in\mathcal{P}}
\bigl\|P^{u}_{p}x_p^{k+1}-u_p+\lambda_p^{k}\bigr\|_2^2
+\tfrac{\rho}{2}\Bigl\|\sum_{p\in\mathcal{P}}u_p - D + \nu^{k}\Bigr\|_2^2
\Bigr\}
\end{aligned}
\label{eq:app-u-update}
\end{equation}

\paragraph{Dual updates}
The scaled dual variables update as
\begin{equation}
\lambda_p^{k+1}
\;:=\;
\lambda_p^{k}
+\bigl(P^{u}_{p}x_p^{k+1}-u_p^{k+1}\bigr),
\qquad \forall p\in\mathcal{P},
\label{eq:app-lambda-update}
\end{equation}
and
\begin{equation}
\nu^{k+1}
\;:=\;
\nu^{k}
+\Bigl(\sum_{p\in\mathcal{P}}u_p^{k+1}-D\Bigr)
\label{eq:app-nu-update}
\end{equation}

\subsection{Closed-form coordinator update}
\label{app:closed-form-u}

The coordinator problem \eqref{eq:app-u-update} admits a closed-form solution. Define the shifted local shipments
\[
\bar{z}_p^{\,k}
:=P^{u}_{p}x_p^{k+1}+\lambda_p^{k},
\qquad
\bar{z}^{\,k}
:=\frac{1}{|\mathcal{P}|}\sum_{p\in\mathcal{P}}\bar{z}_p^{\,k},
\qquad
s^{k}:=D-\nu^{k}
\]
Then the minimizer of \eqref{eq:app-u-update} is given component-wise by
\begin{equation}
u_p^{k+1}
\;=\;
\bar{z}_p^{\,k}-\delta^{k},
\qquad
\delta^{k}
:=\frac{1}{|\mathcal{P}|+1}\Bigl(|\mathcal{P}|\,\bar{z}^{\,k}-s^{k}\Bigr),
\qquad \forall p\in\mathcal{P}.
\label{eq:app-u-closed}
\end{equation}

\begin{proof}
Let $\mathcal{P}=\{1,\dots,P\}$ with $P:=|\mathcal{P}|$, and let $u_p\in\mathbb{R}^n$ for each $p\in\mathcal{P}$. Consider the unconstrained convex quadratic problem
\begin{equation}
\min_{\{u_p\}_{p\in\mathcal{P}}}\;
\frac{\rho}{2}\sum_{p=1}^{P}\bigl\|\bar{z}_p^{\,k}-u_p\bigr\|_2^{2}
+\frac{\rho}{2}\Bigl\|\sum_{p=1}^{P}u_p-s^{k}\Bigr\|_2^{2}.
\label{eq:subprob}
\end{equation}
Since \eqref{eq:subprob} is unconstrained, the KKT conditions reduce to the first-order
stationarity conditions. Define $S:=\sum_{p=1}^P u_p$. The objective can be written as
\[
F(\{u_p\})=\frac{\rho}{2}\sum_{p=1}^{P}\|u_p-\bar{z}_p^{\,k}\|_2^2
+\frac{\rho}{2}\|S-s^k\|_2^2
\]
For any fixed $p\in\mathcal{P}$, differentiating with respect to $u_p$ yields
\[
\nabla_{u_p}F
=\rho\,(u_p-\bar{z}_p^{\,k})+\rho\,(S-s^k).
\]
Stationarity $\nabla_{u_p}F=0$ therefore gives, for all $p\in\mathcal{P}$,
\begin{equation}
u_p-\bar{z}_p^{\,k}+(S-s^k)=0
\qquad\Longleftrightarrow\qquad
u_p=\bar{z}_p^{\,k}-(S-s^k).
\label{eq:up_shift}
\end{equation}
Hence all optimal blocks share a common shift. Let
\begin{equation}
\delta := S-s^k.
\label{eq:delta_def}
\end{equation}
Then \eqref{eq:up_shift} becomes $u_p=\bar{z}_p^{\,k}-\delta$ for all $p\in\mathcal{P}$, and summing
over $p$ gives
\[
S=\sum_{p=1}^{P}u_p=\sum_{p=1}^{P}\bigl(\bar{z}_p^{\,k}-\delta\bigr)
=\sum_{p=1}^{P}\bar{z}_p^{\,k}-P\delta
\]
Using $S=s^k+\delta$ from \eqref{eq:delta_def}, we obtain
\[
s^k+\delta=\sum_{p=1}^{P}\bar{z}_p^{\,k}-P\delta
\qquad\Longrightarrow\qquad
(P+1)\delta=\sum_{p=1}^{P}\bar{z}_p^{\,k}-s^k,
\]
and therefore
\begin{equation}
\delta
=\frac{1}{P+1}\Bigl(\sum_{p=1}^{P}\bar{z}_p^{\,k}-s^k\Bigr)
=\frac{1}{P+1}\Bigl(P\,\bar{z}^{\,k}-s^k\Bigr),
\label{eq:delta_closed}
\end{equation}
where $\bar{z}^{\,k}:=\frac{1}{P}\sum_{p=1}^{P}\bar{z}_p^{\,k}$. Substituting \eqref{eq:delta_closed}
into $u_p=\bar{z}_p^{\,k}-\delta$ yields the closed-form minimizer
\[
u_p^{k+1}=\bar{z}_p^{\,k}-\delta^k,\qquad
\delta^k:=\frac{1}{|\mathcal{P}|+1}\Bigl(|\mathcal{P}|\,\bar{z}^{\,k}-s^k\Bigr),
\qquad \forall p\in\mathcal{P}
\]
Finally, the objective in \eqref{eq:subprob} is a strictly convex quadratic function of
$\{u_p\}_{p\in\mathcal{P}}$ (since $\rho>0$), hence the stationary point above is the unique global
minimizer.
\end{proof}
\noindent
Equation~\eqref{eq:app-u-closed} shows that every plant shares the same offset $\delta^{k}$; thus the ICC only needs the aggregate $\sum_{p}\bar{z}_p^{\,k}$ (equivalently $\bar{z}^{\,k}$) to compute $\delta^{k}$, after which $u_p^{k+1}$ can be formed locally.

\subsection{Two-block form and residuals}
\label{app:residuals}

To express the residuals compactly, rewrite \eqref{eq:icc-consensus} in the two-block form
\begin{equation}
\begin{aligned}
\min_{\{x_p\in\mathcal{F}_p\},\,\{u_p\}}
\quad & \sum_{p\in\mathcal{P}} f_p(x_p) \\
\text{s.t.}\quad
& P^{u}_{p}x_p - u_p = 0, \qquad \forall p\in\mathcal{P},\\
& \sum_{p\in\mathcal{P}}u_p - D = 0
\end{aligned}
\label{eq:app-two-block}
\end{equation}
Define stacked variables $x := (x_p)_{p\in\mathcal{P}}$ and $u := (u_p)_{p\in\mathcal{P}}$, and write
\begin{equation}
A x + B u = c,
\label{eq:app-AxBu}
\end{equation}
where $c := (0,\dots,0,D)$ and the block matrices $A$ and $B$ encode the local equalities $P^{u}_{p}x_p-u_p=0$ and the global equality $\sum_p u_p - D=0$. (The explicit block structure follows directly from \eqref{eq:app-two-block}.)

\paragraph{Primal residual.}
At iteration $k+1$, define
\[
r_{p}^{k+1} := P^{u}_{p}x_p^{k+1}-u_p^{k+1}, \qquad \forall p\in\mathcal{P},
\qquad
r_{0}^{k+1} := \sum_{p\in\mathcal{P}}u_p^{k+1}-D
\]
We report the aggregated primal residual norm as
\begin{equation}
\|r^{k+1}\|_2
:=
\sqrt{\sum_{p\in\mathcal{P}}\|r_{p}^{k+1}\|_2^2 + \|r_{0}^{k+1}\|_2^2 }.
\label{eq:app-primal-residual}
\end{equation}

\paragraph{Dual residual.}
We monitor dual progress via the change in local primal copies available with the plants:
\begin{equation}
\|s^{k+1}\|_2
:=
\rho\,
\sqrt{
\sum_{p\in\mathcal{P}}\|u_p^{k+1}-u_p^{k}\|_2^2}
\label{eq:app-dual-residual}
\end{equation}

When the shipment copies stabilize (and so does their aggregate), $\|s^{k+1}\|_2$ approaches zero, indicating dual convergence.

\section{Upper-bound heuristics for allocating the aggregate mismatch}
\label{app:app-ub-heuristics}

Algorithm~\ref{alg:icc-ub-aggregate} constructs plant-wise perturbations
$\{\epsilon_p^{(h,k)}\}_{p\in\mathcal{P}}$ that redistribute the aggregate mismatch
$\epsilon := \sum_{p\in\mathcal{P}} z_p - D$ while preserving feasibility of the network balance, i.e.,
$\sum_{p\in\mathcal{P}}\epsilon_p^{(h,k)}=\epsilon^k$. Each heuristic $h\in\mathcal{H}$ induces plant-side candidate shipment targets (cf.\ Algorithm~\ref{alg:plant-coord-ub-eval}), which are subsequently certified by the plant-side UB subproblems (Algorithm~\ref{alg:plant-upperbound}).
This appendix summarizes natural extensions of the template heuristics in Algorithm~\ref{alg:icc-ub-aggregate}, and documents an instance-dependent heuristic used in our experiments.

\subsection{Extending \textsc{AssignToTwoPlants} to \textsc{AssignTo}\texorpdfstring{$m$}\textsc{Plants}}
\label{app:assign-to-m-plants}

The heuristics \textsc{AssignToOnePlant} and \textsc{AssignToTwoPlants} are special cases of a general family that allocates $\epsilon^k$ to a subset of plants.
For any integer $m\in\{1,\dots,|\mathcal{P}|-1\}$ and subset $\mathcal{S}\subset\mathcal{P}$ with $|\mathcal{S}|=m$, define
\begin{equation}
\textsc{AssignTo}m\textsc{Plants}(\mathcal{S}):\qquad
\begin{aligned}
\epsilon_{p}^{(h,k)} &:=
\begin{cases}
\omega_p \odot \epsilon^k, & \text{if } p\in\mathcal{S},\\
0, & \text{if } p\notin\mathcal{S},
\end{cases}\\[2pt]
\text{with}\quad &\sum_{p\in\mathcal{S}} \omega_p = \mathbf{1}
\quad (\text{component-wise}).
\end{aligned}
\end{equation}

Here, $\omega_p\in[0,1]^m$ is a (possibly component-wise) split weight, $\odot$ is the Hadamard product, and the constraint
$\sum_{p\in\mathcal{S}}\omega_p=\mathbf{1}$ ensures that $\sum_{p\in\mathcal{P}}\epsilon_p^{(h,k)}=\epsilon^k $.
The cases in Algorithm~\ref{alg:icc-ub-aggregate} correspond to:
(i) $m=|\mathcal{P}|$ with $\omega_p \equiv 1/|\mathcal{P}|$ (equal split),
(ii) $m=1$ with $\omega_{p^\star}\equiv 1$, and
(iii) $m=2$ with $\omega_{p_1}\equiv \omega_{p_2}\equiv 1/2$.

In principle, one may enumerate $m$ up to $|\mathcal{P}|-1$ (hence ``\textsc{AssignToNminus1Plants}'') and/or restrict to a candidate set of subsets based on engineering considerations (e.g., geographic proximity, shared product capability, or historical feasibility). This yields a spectrum of increasing heuristic richness (and compute load) while preserving the privacy property that the ICC requires only the aggregate mismatch $\epsilon^k $ to construct candidates.

\subsection{Non-uniform split ratios and capacity-guided splits}
\label{app:nonuniform-splits}

The equal split rule is often conservative when plants are heterogeneous. A straightforward generalization is to use non-uniform weights over a selected subset of plants.
For $m=2$, this includes ratios such as $20{:}80$, $30{:}70$, and their reversals; more generally, for $m>2$, one can use a library of weight vectors
$\{\omega^{(\ell)}\}_{\ell\in\mathcal{L}}$ satisfying $\omega^{(\ell)}\ge 0$ and $\sum_{p\in\mathcal{S}}\omega_p^{(\ell)}=\mathbf{1}$.
Because $\epsilon^k$ is a vector indexed by $(i,r,w)$, the weights may be either:
(i) \emph{scalar} per plant (applied uniformly across coordinates), or
(ii) \emph{component-wise} (weights vary across $(i,r,w)$), enabling product/region/epoch-specific reallocation.

Instance-dependent splits can exploit exogenous plant attributes that are known to the coordinator, or that the plants are willing to share consistently throughout the coordination process without privacy concerns. For example, a capacity-guided split can be formed via plant-specific scalars $c_p>0$ (e.g., effective throughput capacity) and weights

\begin{equation}
\omega_p = \frac{c_p}{\sum_{q\in\mathcal{S}} c_q},\qquad p\in\mathcal{S},
\end{equation}
or via coordinate-wise capacities $c_{p,j}$ aligned with the shipment index $j\leftrightarrow(i,r,w)$.
These heuristics remain aggregate-only at the ICC: they require only $\epsilon^k$ plus exogenous parameters $(c_p)$ or $(c_{p,j})$, and never require access to individual $z_p$.

\subsection{Instance-dependent heuristic used: electricity-price-weighted split}
\label{app:electricity-split}

In our setting, the shipment vector is indexed by product--region--day $(i,r,w)$ with $w\in\{1,\dots,7\}$, while each plant schedules hourly operation over a 168-hour horizon with time-varying electricity prices. We therefore implemented an instance-dependent heuristic that biases mismatch reallocation toward plants with \emph{lower electricity prices}.

Let $j$ index shipment coordinates in the stacked representation ($j\leftrightarrow(i,r,w)$). For each plant $p$ and coordinate $j$, let $E_{p}[j]$ denote the plant's (known/forecasted) average electricity price relevant to that shipment coordinate (in our implementation, derived from the underlying hourly price profile over the corresponding day $w$ and the plant's operating horizon). Define the per-coordinate minimum price
\begin{equation}
E_{\min}[j] := \min_{p\in\mathcal{P}} E_p[j],
\end{equation}

and a stabilized ``attractiveness'' score
\begin{equation}
a_p[j] := \frac{1}{E_p[j]-E_{\min}[j]+\delta},\qquad \delta>0.
\end{equation}

We then normalize per coordinate to obtain weights
\begin{equation}
\omega_p[j] :=
\begin{cases}
\displaystyle \frac{a_p[j]}{\sum_{q\in\mathcal{P}} a_q[j]}, & \sum_{q} a_q[j] > \texttt{tiny},\\[2mm]
\displaystyle \frac{1}{|\mathcal{P}|}, & \text{otherwise,}
\end{cases}
\end{equation}

where \texttt{tiny} is a small tolerance to robustly handle ties (e.g., identical prices across plants).
The resulting heuristic (\textsc{ElectricitySplitEpsilon}) allocates the mismatch component-wise as

\begin{equation}
\epsilon_p^{(h,k)}[j] = \omega_p[j]\ \epsilon^k[j],\qquad \forall p\in\mathcal{P},\ \forall j,
\end{equation}
and induces plant-side candidate shipments by shifting the plant's iterate by the allocated mismatch (cf.\ Algorithm~\ref{alg:plant-coord-ub-eval}).
By construction, $\sum_{p\in\mathcal{P}}\epsilon_p^{(h,k)}=\epsilon^k$ component-wise, and plants with lower electricity prices receive a larger share of the mismatch adjustment.

This heuristic is attractive in our setting for two reasons: (i) it is \emph{aggregation-compatible} (the ICC requires only $\epsilon^k$ plus exogenous price summaries), and (ii) it aligns the UB candidate generation with an operational proxy for marginal cost, thereby increasing the likelihood that at least one candidate yields a feasible, status-quo-improving upper bound.

\section{Payoff distribution using game theory}
\label{app:app-game-theory}

\subsection{Feasibility of budget-balanced transfers}
\label{app:app-feasible-transfers}

Fix a coordinated feasible operating point $\{x_p^{\text{ub}}\}_{p\in\mathcal{P}}$ derived using some upper-bound heuristic \( h \), and let $ J_p^h$ or simply, 
$J_p := f_p(x_p)$ denote plant $p$'s coordinated cost at this point. Let $\tilde f_p$ denote its status-quo cost.
We seek transfers $\{z_p\}_{p\in\mathcal{P}}$ such that:
\begin{align}
J_p - z_p &\le \tilde f_p, \qquad \forall p\in\mathcal{P}, \label{eq:app-IR}\\
\sum_{p\in\mathcal{P}} z_p &= 0. \label{eq:app-BB}
\end{align}

\begin{theorem}
\label{thm:feasible-transfers}
If
\begin{equation}
\sum_{p\in\mathcal{P}} J_p \;\le\; \sum_{p\in\mathcal{P}} \tilde f_p,
\label{eq:app-TG-nonneg}
\end{equation}
then there exist transfers $\{z_p\}_{p\in\mathcal{P}}$ satisfying \eqref{eq:app-IR}--\eqref{eq:app-BB}.
Moreover, if the inequality in \eqref{eq:app-TG-nonneg} is strict, then the transfers can be chosen so that
$J_p-z_p<\tilde f_p$ for all $p$.
\end{theorem}

\begin{proof}
Partition $\mathcal{P}$ into winners and losers relative to the status quo:
\[
\mathcal{A} := \{p\in\mathcal{P}:\ J_p \le \tilde f_p\},
\qquad
\mathcal{B} := \{p\in\mathcal{P}:\ J_p > \tilde f_p\}.
\]
For $p\in\mathcal{A}$ define the (nonnegative) surplus
$s_p := \tilde f_p - J_p \ge 0$, and for $p\in\mathcal{B}$ define the deficit
$d_p := J_p - \tilde f_p > 0$.
Let
\[
S := \sum_{p\in\mathcal{A}} s_p,
\qquad
D := \sum_{p\in\mathcal{B}} d_p.
\]
Condition \eqref{eq:app-TG-nonneg} is equivalent to
\[
\sum_{p\in\mathcal{P}}(\tilde f_p - J_p) = S - D \ge 0,
\quad\text{hence}\quad S \ge D.
\]

Construct transfers as follows. For each $p\in\mathcal{B}$, set
\[
z_p := d_p,
\]
so that $J_p - z_p = \tilde f_p$ for all $p\in\mathcal{B}$. For plants in $\mathcal{A}$, choose nonnegative
payments $\{\eta_p\}_{p\in\mathcal{A}}$ such that $\sum_{p\in\mathcal{A}}\eta_p = D$ (which is possible since $S\ge D$),
and set
\[
z_p := -\eta_p,\qquad \forall p\in\mathcal{A}.
\]
Then budget balance holds:
\[
\sum_{p\in\mathcal{P}} z_p
=
\sum_{p\in\mathcal{B}} d_p - \sum_{p\in\mathcal{A}} \eta_p
=
D - D
=
0.
\]
Moreover, for any $p\in\mathcal{A}$ we have $0\le \eta_p \le s_p$ (by appropriate choice), implying
\[
J_p - z_p = J_p + \eta_p \le J_p + s_p = \tilde f_p.
\]
Thus \eqref{eq:app-IR} holds for all plants. If \eqref{eq:app-TG-nonneg} is strict, then $S>D$ and we can choose
$\eta_p < s_p$ for all $p\in\mathcal{A}$, yielding strict improvement for every plant.
\end{proof}

\subsection{Nash bargaining solution under budget balance}
\label{app:nash-bargaining}

Fix a coordinated feasible operating point with costs $\{J_p\}_{p\in\mathcal{P}}$, where $J_p:=f_p(x_p)$,
and disagreement (status-quo) costs $\{\tilde f_p\}_{p\in\mathcal{P}}$.
Let $z_p\in\mathbb{R}$ be a transfer credited to plant $p$, with budget balance
\begin{equation}
\sum_{p\in\mathcal{P}} z_p = 0.
\label{eq:app-budget-balance}
\end{equation}
The realized cost of plant $p$ is $J_p - z_p$, hence the surplus relative to disagreement is
\begin{equation}
\Delta_p(z_p) := \tilde f_p - (J_p - z_p) = (\tilde f_p - J_p) + z_p.
\label{eq:app-surplus}
\end{equation}
Nash bargaining selects transfers by maximizing the Nash product of surpluses:
\begin{equation}
\begin{aligned}
\max_{\{z_p\}}\quad & \prod_{p\in\mathcal{P}} \Delta_p(z_p)\\
\text{s.t.}\quad & \sum_{p\in\mathcal{P}} z_p = 0,\\
& \Delta_p(z_p) > 0,\qquad \forall p\in\mathcal{P}.
\end{aligned}
\label{eq:app-nash-product}
\end{equation}

\begin{proposition}
\label{prop:nash-closed-form}
Assume the total savings are positive:
\begin{equation}
TG := \sum_{p\in\mathcal{P}}(\tilde f_p - J_p) > 0.
\label{eq:app-TG}
\end{equation}
Then \eqref{eq:app-nash-product} has a unique optimizer $\{z_p^\star\}$ satisfying
\begin{equation}
\Delta_p(z_p^\star)=\frac{TG}{|\mathcal{P}|},\qquad \forall p\in\mathcal{P},
\label{eq:app-equal-surplus}
\end{equation}
and the corresponding transfers are
\begin{equation}
z_p^\star
=
(J_p-\tilde f_p)
+\frac{TG}{|\mathcal{P}|},
\qquad \forall p\in\mathcal{P}.
\label{eq:app-z-star}
\end{equation}
\end{proposition}

\begin{proof}
Maximizing \eqref{eq:app-nash-product} is equivalent to maximizing
$\sum_{p}\log(\Delta_p(z_p))$. Let $a_p:=\tilde f_p-J_p$ so that $\Delta_p(z_p)=a_p+z_p$.
The problem is strictly concave over $\{z:\ a_p+z_p>0,\ \sum_p z_p=0\}$, hence KKT conditions are necessary and sufficient.
The Lagrangian is
\[
\mathcal{L}(z,\mu)=\sum_{p\in\mathcal{P}}\log(a_p+z_p)-\mu\Big(\sum_{p\in\mathcal{P}} z_p\Big).
\]
Stationarity yields $(a_p+z_p)^{-1}=\mu$ for all $p$, hence $a_p+z_p=c$ for a common constant $c>0$.
Summing and using $\sum_p z_p=0$ gives $|\mathcal{P}|c=\sum_p a_p = TG$, i.e., $c=TG/|\mathcal{P}|$,
which proves \eqref{eq:app-equal-surplus}. Solving $a_p+z_p^\star=TG/|\mathcal{P}|$ for $z_p^\star$ yields
\eqref{eq:app-z-star}. Budget balance holds by construction, and $TG>0$ ensures $\Delta_p(z_p^\star)>0$ for all $p$.
\end{proof}

\section{ADMM implementation details}
\label{app:admm-implement-details}

The hyperparameters in Algorithm~\ref{alg:icc-admm-masked} are
\(\beta, \tau, \theta, \rho^{0}, \rho^{f}, \varepsilon_{\mathrm{pri}}, \varepsilon_{\mathrm{dual}},\) and \(K_{\max}\).

Across all simulation days, we fix \(K_{\max}=30\) and use identical primal and dual stopping tolerances,
\((\varepsilon_{\mathrm{pri}},\varepsilon_{\mathrm{dual}})=(0.01, 0.01)\).
For the primal update damping, we predominantly set \(\beta=0.75\), and reduce it to \(\beta=0.5\) on two days to place greater weight on the newly computed primal iterate.

Penalty updates are performed every \(\rho^{f}=1\) iteration. The initial penalty parameter \(\rho^{0}\) varies across days, with minimum \(1.68\), maximum \(2.43\), mean \(1.96\), and standard deviation \(0.195\); overall, \(\rho^{0}\) remains tightly concentrated around \(2\).
For penalty adaptation, we use \((\tau,\theta)=(1.5,10)\) on most days, and occasionally switch to an aggressive update regime \((\tau,\theta)=(300,1.5)\) to rapidly increase the penalty and promote primal feasibility.

Consequently, the ADMM configuration used for the majority of days is
{\small
\[
\beta=0.75, \tau=1.5, \theta=10, \rho^{0}=1.9, \rho^{f}=1,
(\varepsilon_{\mathrm{pri}},\varepsilon_{\mathrm{dual}})=(0.01,0.01), K_{\max}=30.
\]
}
On days 23 and 31, we instead use
{\small
\[
\beta=0.5,\tau=300,\theta=1.5, \rho^{0}=1.9, \rho^{f}=1,
(\varepsilon_{\mathrm{pri}},\varepsilon_{\mathrm{dual}})=(0.01,0.01), K_{\max}=30,
\]
}
thereby simultaneously (i) accelerating penalty growth to tighten primal feasibility and (ii) reducing \(\beta\) to increase the influence of the latest primal update.

This intervention is supported by the behavior shown in Fig.~\ref{fig:trueubgap}: on days 23 and 31 the relative gap between the ADMM-derived best upper bound and the true social welfare cost is unusually high (approximately \(18\%\)--\(20\%\)). We therefore employ the aggressive penalty-increase strategy on days we encounter instability and poor convergence of the algorithm.

\section{Supplementary ADMM plots and results}
\label{app:app-admm-plots}

\subsection{Day-wise state-matched UB gap}
\label{app:app-state-matched-ub-gap}

Figure~\ref{fig:trueubgap} reports the day-wise state-matched upper-bound (UB) gap achieved by the ADMM-derived feasible schedules under exhaustive heuristic evaluation. For each simulation day, the gap is computed relative to the centralized social-welfare optimum initialized from the same plant states realized by the ADMM-UB trajectory on that day. Summary statistics are provided in the figure legend and in Table~\ref{tab:ub_gap_stats} in the main text.

\begin{figure}[htbp]
  \centering
  \includegraphics[width=0.9\linewidth]{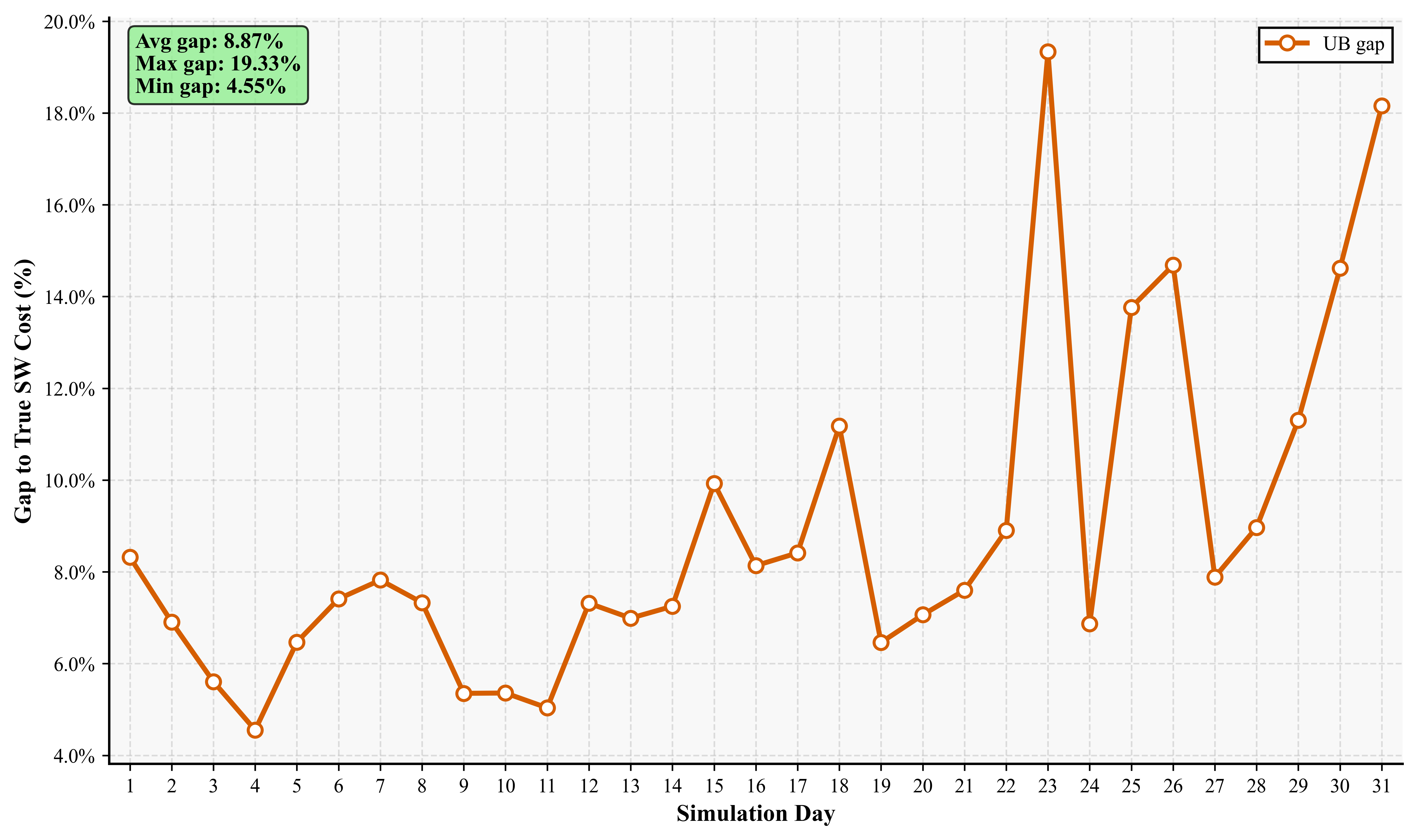}
  \caption{Day-wise state-matched UB gap between the ADMM-derived best feasible upper bound and the centralized social-welfare optimum initialized from the same ADMM-realized plant states.}
  \label{fig:trueubgap}
\end{figure}

\subsection{Heuristic upper-bound evolution on day 4}
\label{app:app-day4-heuristic-ub-evolution}

Figure~\ref{fig:iccobservation} compares the network upper-bound costs obtained by the eight UB heuristics on day~4. The best feasible upper bound is obtained by \texttt{epsilon\_split\_1\_2} at iteration~28, with a network cost of \$124,213.75. This heuristic allocates the aggregate mismatch only across $\mathrm{ASU}_1$ and $\mathrm{ASU}_2$, excluding $\mathrm{ASU}_3$. This behavior is consistent with the electricity-price profile on day~4, where $\mathrm{ASU}_3$ has a substantially higher daily average electricity price than the other plants. Although $\mathrm{ASU}_1$ and $\mathrm{ASU}_2$ have comparable daily average prices, the plant scheduling problems are solved on an hourly grid; hence, intra-day price variation, inventory states, mode-transition constraints, and shipment feasibility can make a joint perturbation over plants~1 and~2 preferable to assigning the mismatch to a single plant. The heuristics involving $\mathrm{ASU}_3$, as well as the all-plant and electricity-weighted splits, remain dominated on this day. The heuristics \texttt{epsilon\_split\_P1} and \texttt{epsilon\_electricity\_split} do not produce feasible upper bounds over the displayed iterations; therefore, although they appear in the legend, no corresponding UB trajectory is visible.
\begin{figure}[htbp]
  \centering
  \includegraphics[width=0.9\linewidth]{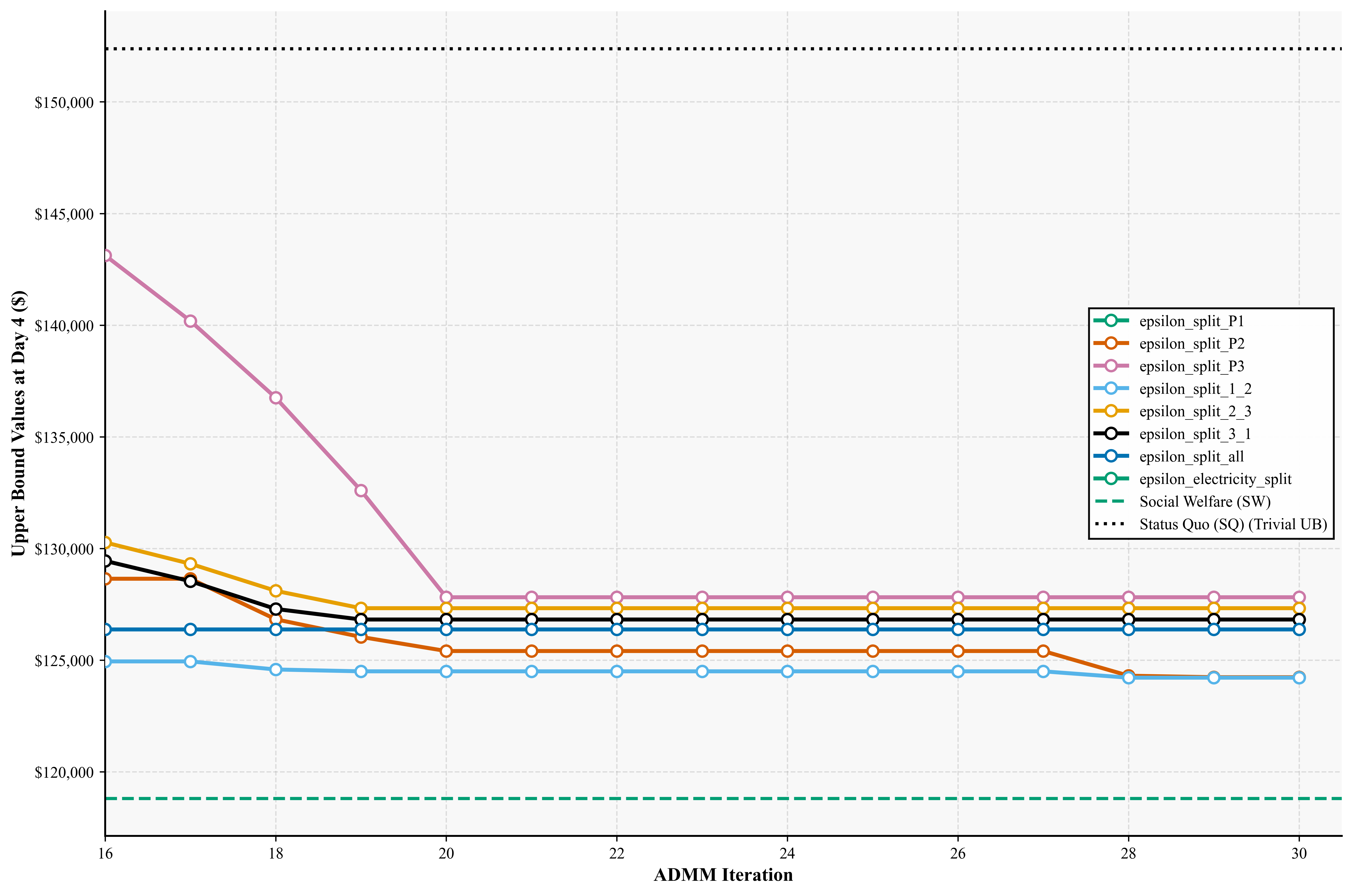}
  \caption{Evolution of heuristic-derived network upper-bound costs on day~4.}
  \label{fig:iccobservation}
\end{figure}

\subsection{Performance of heuristics on the case study}
\label{app:app-heur-performance}

\begin{table}[htbp]
\centering
\caption{Top five UB heuristics ranked by number of days won. A heuristic wins a day if it attains the smallest state-matched UB gap, computed between the ADMM-derived feasible upper-bound cost and the centralized social-welfare optimum initialized from the same ADMM-realized plant states. Gaps are reported in \%.}
\label{tab:top5_heuristics}

\resizebox{\columnwidth}{!}{%
\begin{tabular}{lrrrrrr}
\toprule
\textbf{Heuristic name} &
\textbf{\shortstack{Win\\days}} &
\textbf{\shortstack{Min\\gap (\%)}} &
\textbf{\shortstack{Max\\gap (\%)}} &
\textbf{\shortstack{Avg.\\gap (\%)}} &
\textbf{\shortstack{Var.\\(\%$^{2}$)}} &
\textbf{\shortstack{Solved\\days}} \\
\midrule
$\epsilon\,\mathrm{split}\_{1\_2}$            & 19 & 4.5532 & 19.7458 &  8.2297 & 0.1273 & 24 \\
$\epsilon\,\mathrm{split}\_{2}$             &  6 & 4.5709 & 19.3347 &  9.5529 & 0.1137 & 26 \\
$\epsilon\,\mathrm{split}\_{3}$             &  2 & 7.5847 & 20.1917 & 13.4054 & 0.0840 & 27 \\
$\epsilon\,\mathrm{electricity\_split}$     &  2 & 6.9054 & 19.3462 & 10.3019 & 0.1759 &  6 \\
$\epsilon\,\mathrm{split} \_ {3\_1}$            &  1 & 6.7457 & 14.0326 & 10.1983 & 0.0361 & 23 \\
\bottomrule
\end{tabular}%
}
\end{table}

Table~\ref{tab:top5_heuristics} reports the five best-performing UB heuristics ranked by the number of simulation days on which they achieve the smallest state-matched UB gap to the corresponding centralized benchmark. A ``win day'' denotes a day on which a heuristic yields the lowest feasible upper-bound cost among all evaluated candidates, with the gap computed relative to the centralized social-welfare optimum initialized from the same ADMM-realized plant states. Thus, the table summarizes both the relative quality and the consistency of the heuristic candidates across the rolling-horizon simulation. Two complementary trends emerge.

First, the clear winner is $\epsilon\,\mathrm{split}\_{1\_2}$, which wins on 19 of 31 days and achieves the lowest overall average gap among the listed heuristics ($8.23\%$). This behavior is consistent with the structure of our engineered case study: $\mathrm{ASU}_1$ and $\mathrm{ASU}_2$ have lower average electricity prices over the month (approximately $0.06$ and $0.04$~\$/kWh, respectively), whereas $\mathrm{ASU}_3$ has the highest average price (approximately $0.09$~\$/kWh), as shown in Figure~\ref{fig:elpmonth}. Consequently, reallocating the aggregate mismatch $\epsilon^k$ away from $\mathrm{ASU}_3$ and toward $\mathrm{ASU}_1$--$\mathrm{ASU}_2$ is typically more cost-effective. In addition, $\mathrm{ASU}_1$ is constructed to have the largest capacity in our setup, further reinforcing why heuristics that load $\mathrm{ASU}_1$ (and, jointly, $\mathrm{ASU}_2$) tend to win more often.

Second, the number of \emph{solved days} (days on which a heuristic yields a feasible UB schedule for all plants) does not necessarily correlate with win frequency. For instance, $\epsilon\,\mathrm{split}\_{3}$ is feasible on the most days among the top five (27 solved days) yet wins only twice and has a relatively large average gap ($13.41\%$). This indicates that feasibility alone is insufficient: although assigning mismatch to $\mathrm{ASU}_3$ is often operationally feasible, it is rarely cost-competitive due to $\mathrm{ASU}_3$'s persistently higher electricity prices. Similarly, $\epsilon\,\mathrm{split}\_{2}$ is feasible on 26 days and wins 6, providing a strong baseline among the instance-independent rules that route mismatch to a single plant.

Finally, the instance-dependent rule $\epsilon\,\mathrm{electricity\_split}$ wins on two days but is feasible on only six days. This suggests that price-aware allocation can be effective when its implied shipment targets are compatible with plant-level constraints and inventories, but can be brittle when feasibility is tight. Overall, these results motivate maintaining a \emph{portfolio} of heuristics within ICC--ADMM: simple instance-independent rules offer broad feasibility coverage, while targeted instance-dependent rules can deliver additional improvements when the system state allows.

\section{Adversarial empirical risk minimization model}
\label{app:app-adv-model}

We consider a single attacked plant, with the plant index suppressed. The adversarial empirical risk minimization (ERM) model is written for a fixed adversarial setting \((D,K)\) and a fixed randomized leakage realization \(j\), as defined in the main text. Let
\[
\mathcal D_D=\{1,\dots,D\}
\]
denote the set of attacked simulation days. For each day \(d\in\mathcal D_D\), the adversary observes a randomly sampled subset \(\mathcal A_j^d\) of ADMM iterations from the final ten iterations of that day's coordination transcript, with
\[
|\mathcal A_j^d|=K,
\qquad
\mathcal A_j^d\subseteq \{K_{\max}-9,\dots,K_{\max}\}.
\]
The collection of revealed day--iteration pairs is denoted by
\[
\Omega_{D,K}^{(j)}
:=
\{(d,a): d\in\mathcal D_D,\ a\in\mathcal A_j^d\}.
\]
The adversary's goal is to estimate unknown per-mode fixed-cost coefficients and per-mode--product variable-cost coefficients by fitting a convex ERM model that is consistent with the revealed iterates indexed by \(\Omega_{D,K}^{(j)}\) and the assumed side information.

\paragraph{Index sets}
Let $\mathcal M$ denote the set of operational modes (e.g., \textsc{Off}, \textsc{Liquid\_Prod}, \textsc{Liquid\_Startup}), $\mathcal I$ the product set, $\mathcal R$ the region set, and $\mathcal W=\{1,\dots,W\}$ the lookahead (horizon) days. Define the hourly index set $\mathcal H=\{1,\dots,24W\}$, $\mathcal H_0=\{0\}\cup\mathcal H$, and day-boundary hours $\mathcal H_{24}=\{24,48,\dots,24W\}$; for $h\in\mathcal H_{24}$, let $w=h/24\in\mathcal W$.

\paragraph{Revealed quantities}
For each revealed day--iteration pair \((d,a)\in\Omega_{D,K}^{(j)}\), the adversary observes the mode schedule
\(y_{m,h}^{d,a}\in\{0,1\}\) for all \((m,h)\in\mathcal M\times\mathcal H\), the plant primal update
\(z_{i,r,w}^{d,a}\in\mathbb R\) for all \((i,r,w)\in\mathcal I\times\mathcal R\times\mathcal W\), the ADMM penalty
\(\rho^{d,a}\in\mathbb R_{+}\), and the electricity price
\(\alpha_{h}^{EP,d}\in\mathbb R_{+}\) for all \(h\in\mathcal H\). The adversary also observes the day-boundary inventory
\(IV_{i,0}^{d}\in\mathbb R\) for all \(i\in\mathcal I\) and \(d\in\mathcal D_D\), as well as product-wise production bounds
\(\underline{PD}_{i}\in\mathbb R_{+}\) and \(\overline{PD}_{i}\in\mathbb R_{+}\) for all \(i\in\mathcal I\). Finally, for each \((d,a)\in\Omega_{D,K}^{(j)}\) and \((i,r,w)\in\mathcal I\times\mathcal R\times\mathcal W\), it observes the ADMM center term
\begin{equation}
\label{eq:q_def}
q_{i,r,w}^{d,a}
:=
\bigl(u^{d,a}-\lambda^{d,a}\bigr)_{i,r,w}\in\mathbb R,
\end{equation}
and two aggregate cost ratios,
\(CR_{fixed\_to\_variable}^{agg}\) and \(CR_{production\_to\_startup\_mode}^{agg}\).

\paragraph{Unknowns}
The adversary estimates nonnegative cost coefficients
\(\widehat{\delta}_m\ge 0\) for fixed costs by mode and
\(\widehat{\gamma}_{m,i}\ge 0\) for variable costs by mode and product. These coefficients define the estimated private objective-parameter vector
\[
\widehat c
:=
\Bigl(
\{\widehat{\delta}_m\}_{m\in\mathcal M},
\{\widehat{\gamma}_{m,i}\}_{(m,i)\in\mathcal M\times\mathcal I}
\Bigr),
\]
with entries ordered consistently with the true coefficient vector \(c_{\mathrm{orig}}\). In the numerical attack experiments, this vector corresponds to the reported estimate \(c_{\mathrm{est}}^{(j)}\) for randomized leakage realization \(j\). Thus, the reconstruction error in the main text compares \(c_{\mathrm{est}}^{(j)}=\widehat c\) with the true private coefficient vector \(c_{\mathrm{orig}}\). The adversary also estimates latent operational profiles for each revealed day--iteration pair \((d,a)\in\Omega_{D,K}^{(j)}\), including
\(\widehat{PD}_{i,h}^{d,a}\in\mathbb R_{+}\),
\(\widehat{IV}_{i,h}^{d,a}\in\mathbb R\),
\(\widehat z_{i,r,w}^{d,a}\ge 0\),
\(\widehat E_{d,a,h}\ge 0\), and cost scalars
\(\widehat C_{d,a}\), \(\widehat C^{\mathrm{fix}}_{d,a}\), and
\(\widehat C^{\mathrm{var}}_{d,a}\), all nonnegative.

\paragraph{Adversarial ERM problem (convex QP)}
Given a correctness weight $\Gamma>0$, the adversary solves:
\begin{subequations}
\label{eq:adv_qp}
\begin{align}
\min \quad
& \sum_{(d,a)\in\Omega_{D,K}^{(j)}}
\Bigg[
\widehat{C}_{d,a}
+\frac{\rho^{d,a}}{2}\!\sum_{(i,r,w)\in\mathcal I\times\mathcal R\times\mathcal W}\!
\bigl(\widehat z_{i,r,w}^{d,a}-q_{i,r,w}^{d,a}\bigr)^{2}
\nonumber\\
& \qquad +\Gamma\!\sum_{(i,r,w)\in\mathcal I\times\mathcal R\times\mathcal W}\!
\bigl(\widehat z_{i,r,w}^{d,a}-z_{i,r,w}^{d,a}\bigr)^{2}
\Bigg] \label{eq:adv_obj}\\[1mm]
\text{s.t.}\quad
& \widehat E_{d,a,h}
=\sum_{m\in\mathcal M}\widehat\delta_m\,y_{m,h}^{d,a}
+\sum_{m\in\mathcal M}\sum_{i\in\mathcal I}\widehat\gamma_{m,i}\,\widehat{PD}_{i,h}^{d,a},
\nonumber\\
& \qquad \forall (d,a,h)\in\Omega_{D,K}^{(j)}\times\mathcal H,
\label{eq:adv_power}\\
& \widehat C_{d,a}=\sum_{h\in\mathcal H}\alpha_h^{EP,d}\,\widehat E_{d,a,h},  \qquad  \forall (d,a)\in\Omega_{D,K}^{(j)}, \label{eq:adv_cost_total}\\
& \widehat C^{\mathrm{fix}}_{d,a}=\sum_{h\in\mathcal H}\sum_{m\in\mathcal M}\alpha_h^{EP,d}\,y_{m,h}^{d,a}\,\widehat\delta_m, \qquad  \forall (d,a)\in\Omega_{D,K}^{(j)}, \label{eq:adv_cost_fix}\\
& \widehat C^{\mathrm{var}}_{d,a}=\sum_{h\in\mathcal H}\sum_{m\in\mathcal M}\sum_{i\in\mathcal I}\alpha_h^{EP,d}\,\widehat\gamma_{m,i}\,\widehat{PD}_{i,h}^{d,a},  \qquad \forall (d,a)\in \Omega_{D,K}^{(j)}, \label{eq:adv_cost_var}\\
& \sum_{r\in\mathcal R}\widehat z_{i,r,w}^{d,a}
=\sum_{\tau=h-23}^{h}\widehat{PD}_{i,\tau}^{d,a}
-\Bigl(\widehat{IV}_{i,h}^{d,a}-\widehat{IV}_{i,h-24}^{d,a}\Bigr),
\nonumber\\
& \qquad \forall (d,a,i,h)\in \Omega_{D,K}^{(j)}\times\mathcal I\times\mathcal H_{24},
\label{eq:adv_inv_balance}\\
& \widehat{IV}_{i,0}^{d,a}=IV_{i,0}^{d}, \qquad \forall (d,a,i)\in \Omega_{D,K}^{(j)}\times\mathcal I, \label{eq:adv_inv_init}\\
& \underline{PD}_{i}\,y_{\textsc{LiquidProd},h}^{d,a}\le \widehat{PD}_{i,h}^{d,a}\le \overline{PD}_{i}\,y_{\textsc{LiquidProd},h}^{d,a},
\nonumber\\
& \qquad \forall (d,a,i,h)\in \Omega_{D,K}^{(j)} \times\mathcal I\times\mathcal H,
\label{eq:adv_PD_bounds}\\
& \widehat\delta_{\textsc{Off}}=0,\qquad
\widehat\gamma_{\textsc{Off},i}=0,\qquad
\widehat\gamma_{\textsc{LiquidSU},i}=0,\quad \forall i\in\mathcal I,
\label{eq:adv_mode_structure}\\
& \sum_{(d,a)\in\Omega_{D,K}^{(j)}}\widehat C^{\mathrm{fix}}_{d,a}
=CR_{\textsc{FixToVar}}^{agg}\sum_{(d,a)\in\Omega_{D,K}^{(j)}}\widehat C^{\mathrm{var}}_{d,a},
\label{eq:adv_ratio_fix_var}\\
& \sum_{(d,a)\in\Omega_{D,K}^{(j)}}\widehat C^{\mathrm{fix}, \textsc{LiquidProd}}_{d,a} = CR_{\textsc{ProdToSuMode}}^{agg} \sum_{(d,a)\in\Omega_{D,K}^{(j)}}\widehat C^{\mathrm{fix}, \textsc{LiquidSU}}_{d,a}
\label{eq:adv_ratio_prod_su} 
\end{align}
\end{subequations}

Model~\eqref{eq:adv_qp} fits a cost-consistent latent plant trajectory $(\widehat{PD},\widehat{IV},\widehat z)$ that (i) remains close to the revealed primal update $z^{d,a}$, (ii) is consistent with the ADMM center term $q^{d,a}$ under the penalty $\rho^{d,a}$, and (iii) satisfies day-boundary inventory accounting and capacity bounds under the revealed mode schedule. The ratio constraints \eqref{eq:adv_ratio_fix_var}--\eqref{eq:adv_ratio_prod_su} encode aggregate side information that improves identifiability of fixed-cost components.

\section{ICC--ADMM algorithms}
\label{app:app-icc-algorithms}

\begin{algorithm}[!htbp]
\caption{\textsc{LocalPlantADMMStep}$(p,k)$}
\label{alg:local-plant-solve}
\DontPrintSemicolon
\SetKwInOut{Input}{Input}
\SetKwInOut{InState}{Initial iterates}
\SetKwInOut{State}{Internal State}
\SetKwInOut{Param}{Hyperparameter}
\SetKwInOut{Output}{Output}

\Input{Plant index $p$, iteration $k$}
\InState{ $( \{u_p^0,\lambda_p^0, \nu_p^0,z_p^{*,0} \}  =\mathbf{0}_m, \rho^0_p)$}
\State{  $S_p^{\text{init}}, (u_p^k,\lambda_p^k, \nu_p^k, \rho^k_p, z_p^{\star,k} )$; neighbor masking seeds $\{s_{pq}\}_{q\in\mathcal{P}\setminus\{p\}}$ }
\Param{$\beta$}
\Output{$(\widetilde{z}_p^{k+1},\widetilde{\overline{z}}_p^{k})$.}

\tcp{\emph{Local MIQP solve (plant-side subproblem)}}
$x_p^{k+1} \leftarrow \arg\min\limits_{x_p\in\mathcal{F}_p(S_p^{\text{init}})}
\Big\{ f_p(x_p) + \tfrac{\rho^k_p}{2}\|P_p^u x_p - u_p^{k} + \lambda_p^{k}\|_2^2 \Big\}$\;

$z_p^{k+1} \leftarrow P_p^u x_p^{k+1}$\;

$z_p^{\star,k+1} \leftarrow \beta z_p^{\star,k} + (1-\beta)z_p^{k+1}$\tcp*[r]{damping}

$\overline{z}_p^{k} \leftarrow z_p^{\star,k+1} + \lambda_p^{k}$\tcp*[r]{shifted local shipments}

\For{$q\in \mathcal{N}(p)=\mathcal{P}\setminus\{p\}$ }{ 
$m_{pq}^k \leftarrow \mathrm{PRG}(s_{pq},k)$ \;
}
$mask_p^k \leftarrow \sum_{q\in\mathcal{N}(p):\,p<q} m_{pq}^k
-\sum_{q\in\mathcal{N}(p):\,q<p} m_{qp}^k$ \tcp*[r]{mask generation}

$\widetilde{z}_p^{k+1} \leftarrow z_p^{\star,k+1} + mask_p^k$\;
$\widetilde{\overline{z}}_p^{k} \leftarrow \overline{z}_p^{k} + mask_p^k$\;

$z_p^{\star,k} \leftarrow z_p^{\star,k+1} $ \tcp*[r]{store for next iteration}

\Return{$(\widetilde{z}_p^{k+1},\widetilde{\overline{z}}_p^{k})$}\;
\end{algorithm}

\begin{algorithm}[!htbp]
\caption{\textsc{ICCUBHeuristics}$(D,\{\widetilde{z}_p^{k+1}\}_{p\in\mathcal{P}},k)$}
\label{alg:icc-ub-aggregate}
\DontPrintSemicolon
\SetKwInOut{Input}{Inputs}
\SetKwInOut{Output}{Outputs}

\Input{
Total demand $D\in\mathbb{R}^{m}$; masked plant shipments $\{\widetilde{z}_p^{k+1}\}_{p\in\mathcal{P}}$; iteration \(k\).}
\Output{
A set of candidate offset allocations $\{\epsilon_p^{(h,k)}\}_{p\in\mathcal{P}}$ for each candidate heuristic $h$ such that $\sum_{p\in\mathcal{P}}\epsilon_p^{(h,k)}=\epsilon^k$.}

\tcp{\emph{Aggregate mismatch (masks cancel under summation)}}
$Q^k \leftarrow \sum_{p\in\mathcal{P}} \widetilde{z}_p^{k+1}$\;
$\epsilon^k \leftarrow Q^k - D$\;

\tcp{\emph{Construct a list of instantiated heuristic candidates}}
$\mathcal{H} \leftarrow \{\textsc{SplitEpsilonEqual}\}$\;
$\mathcal{H} \leftarrow \mathcal{H}\ \cup\ \{\textsc{AssignToOnePlant}(p): p\in\mathcal{P}\}$\;
$\mathcal{H} \leftarrow \mathcal{H}\ \cup\ \{\textsc{AssignToTwoPlants}(p_1,p_2): (p_1,p_2)\in\mathcal{S}_2\}$\tcp*[r]{e.g., $\mathcal{S}_2=\{(p_1,p_2)\in\mathcal{P}^2: p_1<p_2\}$ or a subset}

\tcp{\emph{Generate candidate allocations}}
\ForEach{$h\in\mathcal{H}$}{
    \uIf{$h=\textsc{SplitEpsilonEqual}$}{
        $\epsilon_p^{(h,k)} \leftarrow \epsilon^k/|\mathcal{P}|,\quad \forall p\in\mathcal{P}$\;
    }
    \uElseIf{$h=\textsc{AssignToOnePlant}(p^\star)$}{
        $\epsilon_{p^\star}^{(h,k)} \leftarrow \epsilon^k$;\quad
        $\epsilon_{p}^{(h)} \leftarrow 0,\ \forall p\in\mathcal{P}\setminus\{p^\star\}$\;
    }
    \uElseIf{$h=\textsc{AssignToTwoPlants}(p_1,p_2)$}{
        $\epsilon_{p_1}^{(h,k)} \leftarrow \epsilon^k/2$;\quad
        $\epsilon_{p_2}^{(h,k)} \leftarrow \epsilon^k/2$;\quad
        $\epsilon_{p}^{(h,k)} \leftarrow 0,\ \forall p\in\mathcal{P}\setminus\{p_1,p_2\}$\;
    }
    \Else{
    \tcp{Additional heuristic rules cf.~\ref{app:app-ub-heuristics}.}
    }
}
\Return{$\{\epsilon_p^{(h,k)}\}_{p\in\mathcal{P},\,h\in\mathcal{H}}$}\;
\end{algorithm}

\begin{algorithm}[!htbp]
\caption{\textsc{PlantCoordinatorUpdate\&UBEval}$(p,k,\bar z^{\,k},D,\{\epsilon_p^{(h,k)}\}_{h\in\mathcal{H}})$}
\label{alg:plant-coord-ub-eval}
\DontPrintSemicolon
\SetKwInOut{Input}{Inputs}
\SetKwInOut{State}{Local state}
\SetKwInOut{Output}{Outputs}

\Input{
Plant index $p$; iteration $k$; broadcast mean $\bar z^{\,k}\in\mathbb{R}^{m}$; total demand $D\in\mathbb{R}^{m}$;
candidate offsets $\{\epsilon_p^{(h,k)}\}_{h\in\mathcal{H}}$ for plant $p$.}
\State{
Local variables $\nu_p^{k}\in\mathbb{R}^{m}$, $\overline z_p^{k}\in\mathbb{R}^{m}$, $mask_p^{k}\in\mathbb{R}^{m}$ (from \textsc{LocalPlantStep}); $\mathcal{F}_p(S_p^{\text{init}})$.}
\Output{
masked message $\widetilde u_p^{k+1}$;
and UB-evaluation summaries $\{(J_{p}^{(h,k)},\,\texttt{feas}_{p}^{(h,k)})\}_{h\in\mathcal{H}}$.}

\tcp{ \emph{Coordinator step cf.\ \eqref{eq:icc-u-update} performed locally}}
$s_p^{k} \leftarrow D-\nu_p^{k}$\tcp*[r]{local copy of coordinator dual}
$\delta_p^{k} \leftarrow \frac{1}{|\mathcal{P}|+1}\Bigl(|\mathcal{P}|\,\bar z^{\,k}-s_p^{k}\Bigr)$\;
$u_p^{k+1} \leftarrow \overline z_p^{k} - \delta_p^{k}$\;

\tcp{\emph{Evaluate UB candidates}}
\ForEach{$h\in\mathcal{H}$}{
    $\hat d_{p}^{(h,k)} \leftarrow z_p^{*,k+1} - \epsilon_{p}^{(h,k)}$\tcp*[r]{candidate shipping for plant $p$}
    \uIf{$\hat d_{p}^{(h,k)} \notin \mathbb{R}_+^{m}$}{
        $\texttt{feas}_{p}^{(h,k)} \leftarrow 0$;\quad
        $J_{p}^{(h,k)} \leftarrow +\infty$\;
    }
    \Else{
        \tcp{Problem-dependent UB solve; fix shipments to $\hat d_p^{(h)}$, check feasibility and possibly compute a cost}
        $(J_{p}^{(h,k)}, \, \texttt{feas}_{p}^{(h,k)}, S_p^{\text{new},(h,k)}) \leftarrow \textsc{PlantUpperBound}(p,h,k,\hat d_{p}^{(h,k)},S_p^{\text{init}})$\;
    }
}

\tcp{\emph{Mask the primal shipment copy before sending to ICC }}
$\widetilde u_p^{k+1} \leftarrow u_p^{k+1} + mask_p^{k}$\;

\Return{$\widetilde u_p^{k+1},\{(J_{p}^{(h,k)},\texttt{feas}_{p}^{(h,k)})\}_{h\in\mathcal{H}})$}\;
\end{algorithm}

\begin{algorithm}[!htbp]
\caption{\textsc{PlantUpperBound}$(p,h,k,\hat d_p^{(h,k)},S_p^{\text{init}})$}
\label{alg:plant-upperbound}
\DontPrintSemicolon
\SetKwInOut{Input}{Inputs}
\SetKwInOut{Output}{Outputs}

\Input{
Plant index $p$; candidate shipment target $\hat d_p^{(h,k)}\in\mathbb{R}^{m}$; initial state $S_p^{\text{init}}$.}
\Output{
Feasibility flag $\texttt{feas}_p^{(h,k)}\in\{0,1\}$; upper-bound cost $J_p^{(h,k)}\in\mathbb{R}\cup\{+\infty\}$; updated state $S_p^{\text{new}}$.}

\tcp{\emph{necessary checks}}
\uIf{$\hat d_p ^{(h,k)}\notin \mathbb{R}_+^{m}$}{
    $\texttt{feas}_p^{(h,k)} \leftarrow 0$;\quad $J_p \leftarrow ^{(h,k)}+\infty$;\quad \Return{$(J_p,\texttt{feas}_p,S_p^{\text{init}})$}\;
}

\tcp{\emph{UB solve by enforcing shipments}}
Solve the following plant feasibility problem:
\begin{equation}
\label{eq:ub-plant-problem}
\begin{aligned}
J_p^{(h,k)} \;:=\; \min_{x_p}\quad & f_p(x_p) \\
\text{s.t.}\quad & P_p^u x_p = \hat d_p^{(h,k)}, \\
& x_p \in \mathcal{F}_p(S_p^{\text{init}}), \\
& x_p \in \mathbb{R}_+^{r_p}\times\{0,1\}^{b_p}.
\end{aligned}
\end{equation}

\uIf{problem \eqref{eq:ub-plant-problem} is feasible}{
    $\texttt{feas}_p^{(h,k)} \leftarrow 1$\;
    Obtain an optimizer $x_p^{\mathrm{ub}}$ (or any feasible solution) from the solver\;
    \tcp{\emph{Update state implied by the feasible schedule}}
    $S_p^{\text{new},{(h,k)}} \leftarrow \textsc{StateTransition}(S_p^{\text{init}},x_p^{\mathrm{ub}})$\;
}
\Else{
    $\texttt{feas}_p^{(h,k)} \leftarrow 0$;\quad $J_p^{(h,k)} \leftarrow +\infty$;\quad $S_p^{\text{new},{(h,k)}} \leftarrow S_p^{\text{init}}$\;
}

\Return{$( J_p^{(h,k)}, \, \texttt{feas}_p^{(h,k)},\, S_p^{\text{new},{(h,k)}})$}\;
\end{algorithm}

\begin{algorithm}[!htbp]
\caption{\textsc{ICCUBEval}$(k, \{(J_{p}^{(h,k)},\texttt{feas}_{p}^{(h,k)})\}_{p\in\mathcal{P},\,h\in\mathcal{H}}, \{\tilde f_p\}_{p\in\mathcal{P}})$}
\label{alg:icc-ub-eval}
\DontPrintSemicolon
\SetKwInOut{Input}{Inputs}
\SetKwInOut{State}{Local State}
\SetKwInOut{Output}{Outputs}

\Input{
Current iteration $k$;
per-plant UB summaries $\{(J_{p}^{(h,k)},\texttt{feas}_{p}^{(h,k)})\}$ for each $p\in\mathcal{P}$ and $h\in\mathcal{H}$, where
$\texttt{feas}_{p}^{(h,k)}\in\{0,1\}$ and $J_{p}^{(h,k)}\in\mathbb{R}\cup\{+\infty\}$;
status-quo plant costs $\{\tilde f_p\}_{p\in\mathcal{P}}$.}

\State{Incumbent best upper bound $J^{\star}\in\mathbb{R}\cup\{+\infty\}$ with incumbent label $(h^{\star},k^{\star}, \texttt{feas}^{\star} )$}

\Output{
Updated incumbent best upper bound $J^{\star}$; corresponding heuristic $h^{\star}$; global feasibility flag $\texttt{feas}^{\star}\in\{0,1\}$; and iteration index $k^{\star}$ at which $J^{\star}$ was attained.}

\tcp{\emph{Evaluate feasibility and aggregate UB cost for each heuristic at iteration $k$}}
\ForEach{$h\in\mathcal{H}$}{
    $\texttt{feas}^{(h,k)} \leftarrow \prod\limits_{p\in\mathcal{P}} \texttt{feas}_{p}^{(h,k)}$\tcp*[r]{AND across plants}
    \uIf{$\texttt{feas}^{(h,k)} = 1$}{
        $J^{(h,k)} \leftarrow \sum\limits_{p\in\mathcal{P}} J_{p}^{(h,k)}$\tcp*[r]{aggregate UB cost}
    }
    \Else{
        $J^{(h,k)} \leftarrow +\infty$\;
    }
}

\tcp{\emph{Best feasible heuristic at iteration $k$ (if any)}}
$h^{k} \leftarrow \arg\min\limits_{h\in\mathcal{H}} J^{(h,k)}$\;
$J^{k} \leftarrow J^{(h^{k},k)}$\;

\tcp{\emph{Global feasibility at iteration $k$: must be feasible and improve on network status quo}}
$\texttt{feas}^{k} \leftarrow
\mathbb{I}\!\left[\,J^{k}<+\infty \ \wedge\  J^{k}\le \sum_{p\in\mathcal{P}}\tilde f_p\,\right]$\;

\tcp{\emph{Update incumbent best upper bound only if improved at iteration $k$}}
\If{$\texttt{feas}^{k}=1$ \textbf{and} $J^{k} < J^{\star}$}{
    $J^{\star} \leftarrow J^{k}$\;
    $h^{\star} \leftarrow h^{k}$\;
    $k^{\star} \leftarrow k$\;
    $\texttt{feas}^{\star} \leftarrow 1$\tcp*[r]{improved incumbent found at iter $k$}
}

\Return{$(J^{\star},h^{\star},k^{\star},\texttt{feas}^{\star})$}\;
\end{algorithm}

\begin{algorithm}[!htbp]
\caption{\textsc{PlantUpdateDual\&State\&Residuals}$(p,k,\bar u^{\,k+1},h^\star,\texttt{feas}^\star, k^\star)$}
\label{alg:plant-dual-state-residual}
\DontPrintSemicolon
\SetKwInOut{Input}{Inputs}
\SetKwInOut{AState}{Local ADMM state}
\SetKwInOut{BState}{Local UB state}
\SetKwInOut{Output}{Outputs}

\Input{
Plant index $p$; iteration $k$; broadcast mean shipment $\bar u^{\,k+1}\in\mathbb{R}^{m}$;
selected heuristic $h^\star$; global feasibility flag $\texttt{feas}^\star\in\{0,1\}$.}

\AState{
Local iterates $(u_p^{k+1},u_p^{k},\lambda_p^{k},\nu_p^{k})$;
plant shipments ($z_p^{k+1}, z_p^{\star,k+1}$) from \textsc{LocalPlantADMMStep};
feasible UB candidate state and schedules, associated with $h$, ($x_{p}^{\mathrm{ub},(h)}$, $S_p^{(h)}, J_p^{(h)} \; \text{if} \; \texttt{feas}_{p}^{(h)} ) $.}

\BState{Incumbent best upper bound $ \{ S_p^{\text{new}}, x_p^{\star},J_p^{\star} \} $}

\Output{Scalar residual pieces $(\alpha_p^{k+1},\gamma_p^{k+1})$.}

\tcp{\emph{Dual ascent (scaled ADMM)}}
$\nu_p^{k+1} \leftarrow \nu_p^{k} + \Bigl(|\mathcal{P}|\,\bar u^{\,k+1}-D\Bigr)$\;
$\lambda_p^{k+1} \leftarrow \lambda_p^{k} + \Bigl(z_p^{\star,k+1}-u_p^{k+1}\Bigr)$\;

\tcp{\emph{Plant commits state update implied by the selected UB heuristic}}
\If{$\texttt{feas}^\star = 1 \ \texttt{and} \ k^\star = k$ }{
    \tcp{If plants computed UB candidates, they locally store the best feasible one per $h$.
    When $h^\star$ is selected, plant $p$ commits the corresponding state update.}
    $S_p^{\text{new}} \leftarrow S_p^{\text{new},(h^\star,k)}$ \tcp*[r]{updates boundary conditions}
    $x_p^{\star} \leftarrow  x_p^{(h^\star,k)} $ \tcp*[r]{plant updates decisions}
    $J_p^{\star} \leftarrow  J_p^{(h^\star, k)} $ \tcp*[r]{plant updates UB cost}
}

\tcp{\emph{Residual pieces (sent to ICC as scalars)}}
$\alpha_p^{k+1} \leftarrow \bigl\|z_p^{k+1}-u_p^{k+1}\bigr\|_2^2$\tcp*[r]{local primal residual piece}
$\gamma_p^{k+1} \leftarrow  (\rho_p^k)^2 \bigl\|u_p^{k+1}-u_p^{k}\bigr\|_2^2$\tcp*[r]{local dual residual piece}

\tcp{\emph{Store iterate for next iteration}}
$u_p^{k} \leftarrow u_p^{k+1}$\;

\Return{$\bigr( (S_p^{\text{new}}, x_p^{\star}, J_p^{\star}), (\lambda_p^{k+1}, \nu_p^{k+1}), (\alpha_p^{k+1},\gamma_p^{k+1}) \bigl) $}\;
\end{algorithm}

\begin{algorithm}[!htbp]
\caption{\textsc{PlantRhoUpdateAndDualRescale}$(p,k,\|r^{k+1}\|_2,\|s^{k+1}\|_2)$}
\label{alg:plant-rho-update-rescale}
\DontPrintSemicolon
\SetKwInOut{Input}{Inputs}
\SetKwInOut{State}{Local state}
\SetKwInOut{Param}{Hyperparameters}
\SetKwInOut{Output}{Outputs}

\Input{Plant index $p$; iteration $k$; global residual norms $\|r^{k+1}\|_2$ and $\|s^{k+1}\|_2$ broadcast by the ICC.}
\State{Local penalty $\rho_p^{k}>0$; scaled dual variables $\lambda_p^{k+1}$ and $\nu_p^{k+1}$ after the dual-ascent step.}
\Param{Residual-balancing parameters $(\tau,\theta)$ announced once by the ICC, with $\tau>1$ and $\theta>1$; penalty update frequency $\rho^{f}\in\mathbb{Z}_{\ge 1}$.}
\Output{Updated $(\rho_p^{k+1},\lambda_p^{k+1},\nu_p^{k+1})$.}

\tcp{\emph{Default: no change}}
$\rho_p^{k+1} \leftarrow \rho_p^{k}$\;

\tcp{\emph{Update $\rho$ only every $\rho^{f}$ iterations}}
\If{$k \bmod \rho^{f} = 0$}{
    \uIf{$\|r^{k+1}\|_2 > \tau\,\|s^{k+1}\|_2$}{
        $\rho_p^{k+1} \leftarrow \theta\,\rho_p^{k}$\;
    }
    \ElseIf{$\|s^{k+1}\|_2 > \tau\,\|r^{k+1}\|_2$}{
        $\rho_p^{k+1} \leftarrow \rho_p^{k}/\theta$\;
    }
}

\tcp{\emph{Rescale scaled duals if $\rho$ changed (scaled ADMM)}}
\If{$\rho_p^{k+1}\neq \rho_p^{k}$}{
    $c_{\rho,p}^{k+1} \leftarrow \rho_p^{k}/\rho_p^{k+1}$\;
    $\lambda_p^{k+1} \leftarrow c_{\rho,p}^{k+1}\,\lambda_p^{k+1}$\;
    $\nu_p^{k+1} \leftarrow c_{\rho,p}^{k+1}\,\nu_p^{k+1}$\;
}

\Return{$(\rho_p^{k+1},\lambda_p^{k+1},\nu_p^{k+1})$}\;
\end{algorithm}

\newpage

\FloatBarrier
\clearpage


\section{Air separation unit model}
\label{app:asumodel}

In this section, we present the mixed-integer linear programming (MILP) formulation for the operation of a set of air separation unit (ASU) plants producing liquid nitrogen (LIN), liquid oxygen (LOX), and liquid argon (LAR). The ASUs differ in production capacities, storage capacities, and operating characteristics. The plants are required to satisfy specified daily production targets over a weekly planning horizon, while production rates and operating decisions are modeled at an hourly resolution. We assume that ASUs do not rely on third-party external purchases, since the assigned production targets are feasible within the available plant capacities. Before presenting the complete scheduling formulation, we first describe the internal operating structure of the ASU model, following \citet{Karwan2007}.

An ASU uses air as the raw material and separates it into its constituent components through a sequence of cryogenic unit operations. The incoming air is first compressed and purified to remove contaminants, then cooled to cryogenic temperatures and partially liquefied before being fed to the distillation system. The separated nitrogen, oxygen, and argon streams are subsequently processed through liquefaction equipment to obtain the required liquid products. Depending on the status of the major equipment units and whether they are active, inactive, or being prepared for operation, the plant may operate in different modes. In this work, we consider three operating modes: OFF, Liquid Start-up, and Liquid Production. The objective of the scheduling problem is to determine the optimal hourly operating mode and production rate of each ASU so as to minimize operating cost under the considered electricity pricing scheme while meeting the prescribed production targets.

Electricity prices are treated as deterministic inputs to the optimization model. Specifically, the 24-hour day-ahead electricity prices are assumed to be known, while prices beyond the day-ahead period are represented using forecasts. Although production targets are specified on a daily basis, the scheduling decisions must be made at an hourly resolution because of equipment-level operational constraints and hourly variations in electricity prices. Consequently, electricity consumption is directly linked to the operating mode of the ASU and, in the Liquid Production mode, to the selected production rate. The feasible production region in the Liquid Production mode is approximated using the convex hull of historical production data. Figure~\ref{fig:convhull} illustrates this approximation for representative production rates when the liquefier is active, resulting in a three-dimensional convex polytope for the considered liquid products. This approximation can be further refined using the Convex Region Surrogate (CRS) framework described by \citet{Zhang2015crs}. Mode transitions are restricted by liquefier start-up and shut-down requirements, and additional minimum up-time and down-time constraints may be imposed to limit excessive switching and equipment wear. These transition restrictions are represented using the state-transition graph shown in Figure~\ref{fig:transition}.

Table~\ref{tab:table_1} reports the transition-time parameters associated with 
the liquefier state-transition graph. Since the optimization horizon is 
discretized at an hourly resolution, transition times with $\theta=1$ hour are 
implicitly captured within a single time interval and therefore do not introduce 
additional inter-temporal restrictions. The proposed framework can be extended to 
incorporate additional operational details, such as ramp-up trajectories, 
start-up product losses, and product-purity requirements. In this work, we model 
the ASU plant in a manner similar to the plant configuration studied by \citet{Zhang2015}, where the integration of a cryogenic energy storage system with an 
existing ASU plant is assessed. However, the cryogenic energy storage system is 
not considered in the present study. The complete ASU operating environment is 
shown in Figure~\ref{fig:asuenv}. Given daily production targets over a weekly 
planning horizon, the ASU operator must determine the hourly operating decisions 
for each plant.

\begin{itemize}
    \item the mode of operation,
    \item the production quantity for each product,
    \item the amount of power to be purchased from the electricity market,
    \item the quantity of liquids to be stored.
\end{itemize}

\begin{figure}[h!]
  \centering
   \includegraphics[width=0.7\textwidth]{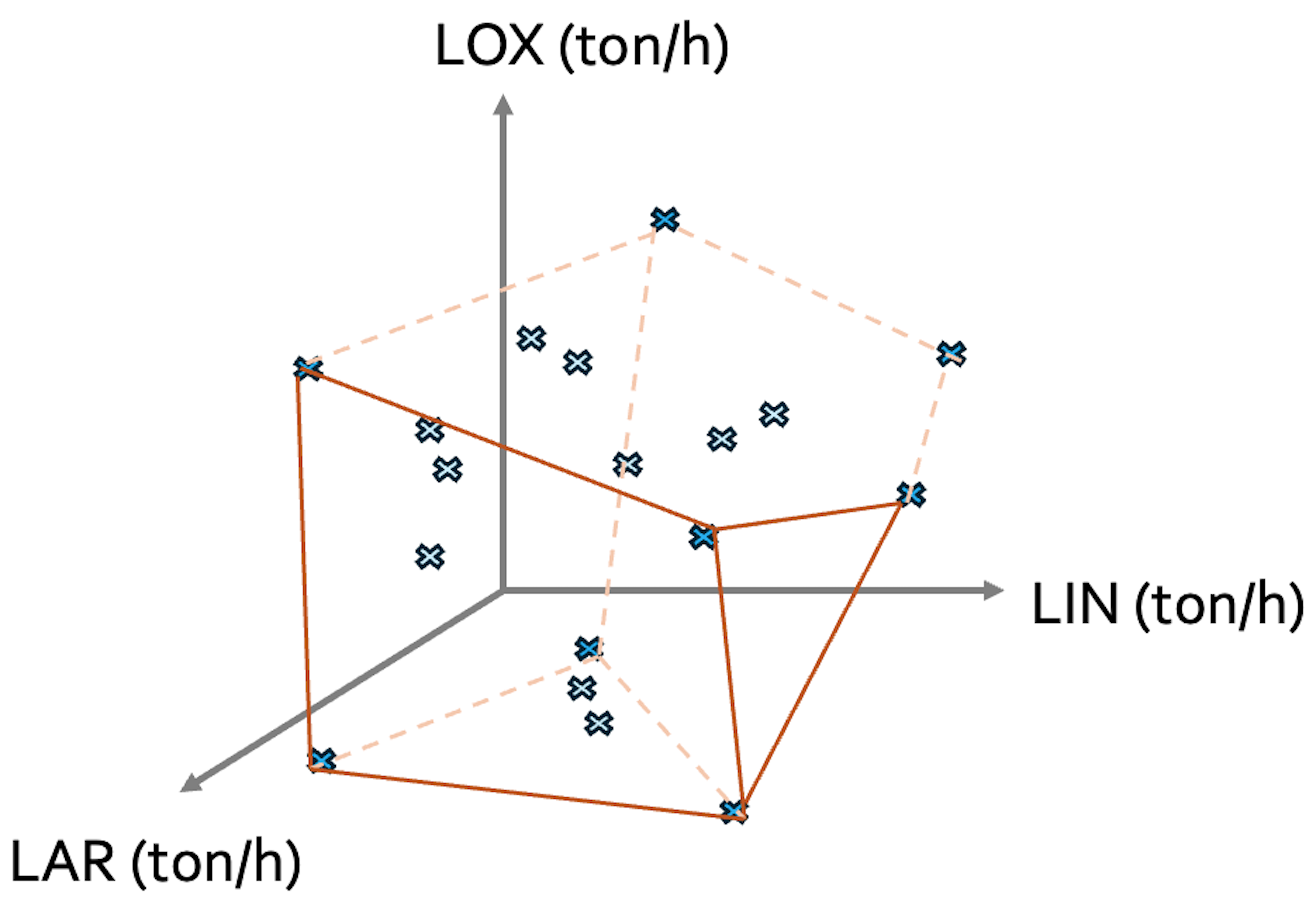}
  \caption{Exemplar approximation of feasible region using convex hull.}
  \label{fig:convhull}
\end{figure}

\begin{figure}[h!]
  \centering
  \includegraphics[width=\textwidth]{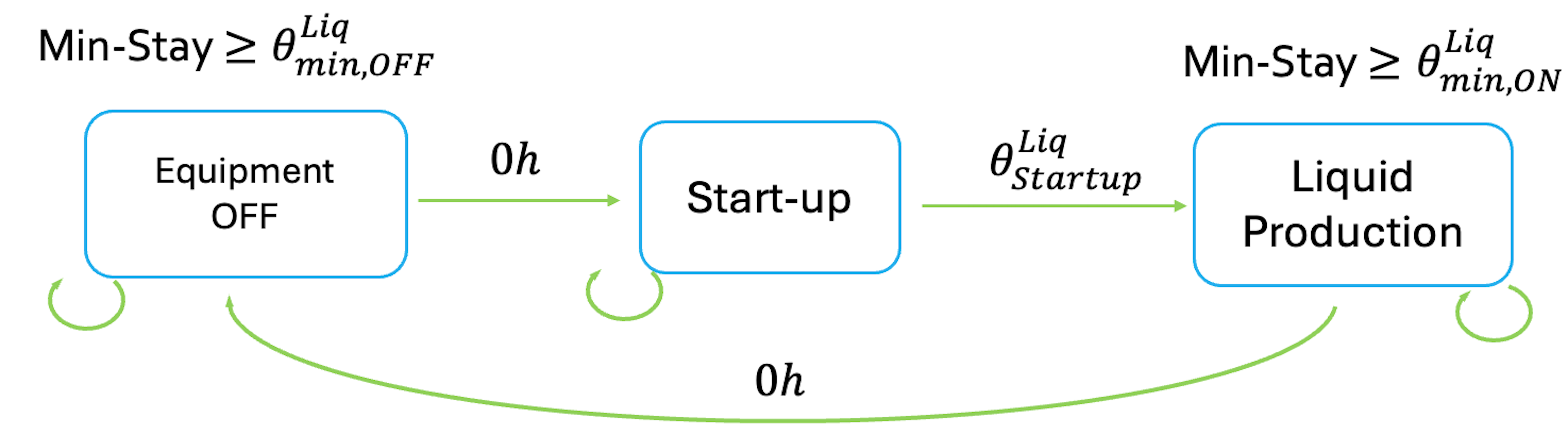}
  \caption{State graph of the internal liquefier in the ASU. The arrows indicate the allowed transitions and their required transition times.}
  \label{fig:transition}
\end{figure}

\begin{table}[htbp]
\centering
\caption{Description of theta values used in the state graph}
\label{tab:table_1}
\begin{tabularx}{\textwidth}{cXXX}
\toprule
& \textbf{$\theta_{\text{min,OFF}}^{\text{Liq}}$} 
& \textbf{$\theta_{\text{Startup}}^{\text{Liq}}$} 
& \textbf{$\theta_{\text{min,ON}}^{\text{Liq}}$} \\
\midrule
\textbf{Description} 
& Minimum time the liquefier must be OFF 
& Time it takes to start liquid production once turned ON 
& Minimum time the liquefier must be ON \\
\bottomrule
\end{tabularx}
\end{table}

\begin{figure}[h!]
  \centering
  \includegraphics[width=\textwidth]{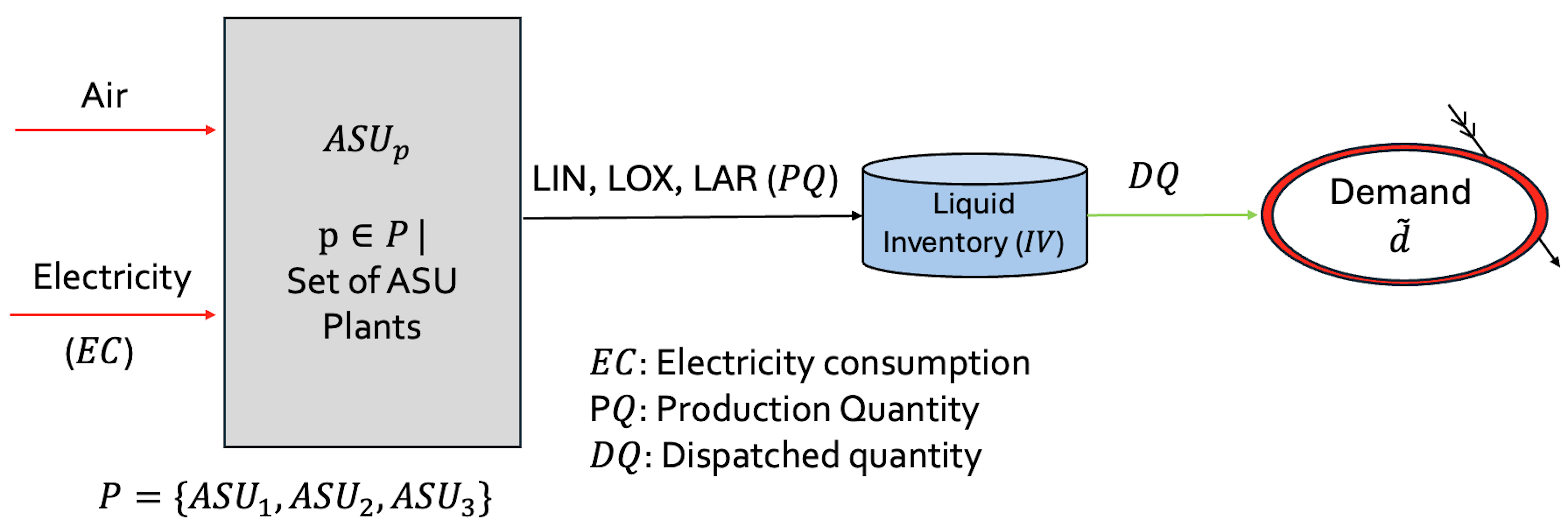}
  \caption{External ASU environment.}
  \label{fig:asuenv}
\end{figure}

\begin{figure}[h!]
  \centering
  \includegraphics[width=\textwidth]{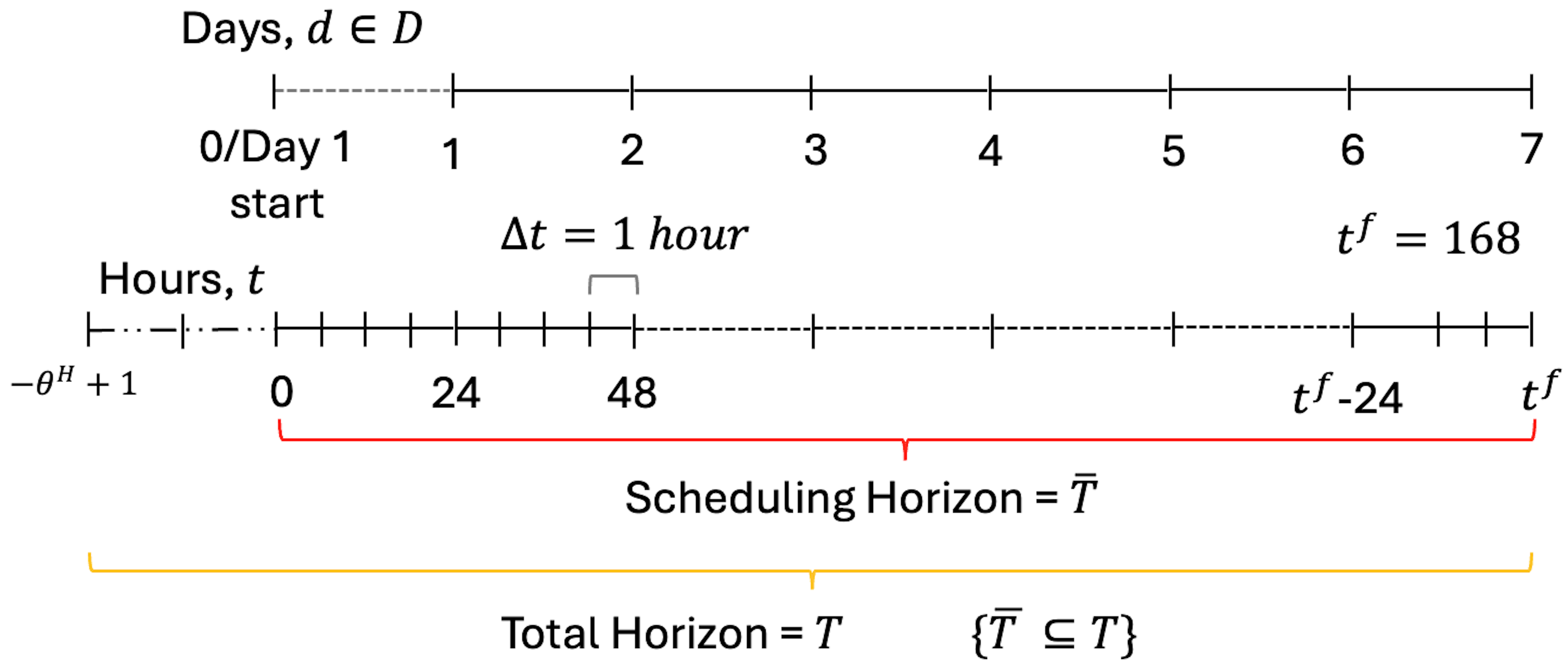}
  \caption{Daily and hourly discretization of weekly time horizon.}
  \label{fig:horizon}
\end{figure}

To formulate the scheduling problem, we discretize the 7-day planning horizon 
into uniform time intervals of 1 hour, resulting in $7 \times 24 = 168$ hourly 
decision periods. A decision indexed by time $t$ represents the operating 
decision applied during the interval $(t-1,t]$. The hourly grid representation of 
the 7-day horizon is shown in Figure~\ref{fig:horizon}. The decision-making 
horizon is denoted by $\bar T=\{1,2,\ldots,t^f\}$ and is embedded within the 
extended time horizon 
$T=\{-\theta^H+1,-\theta^H+2,\ldots,0,1,\ldots,t^f\}$, which includes historical 
time periods required to enforce mode-transition constraints. The historical 
mode information is provided as a boundary condition to ensure that future 
operating decisions are consistent with the prior state of each plant. Since ASU 
production targets are specified on a daily basis, we introduce the day index 
$d \in D=\{1,2,\ldots,7\}$ in addition to the hourly index $t$. The discretized 
optimization setting from the operators' perspective is shown in 
Figure~\ref{fig:situation}. The set of liquid products is denoted by 
$\bar I=\{\mathrm{LIN},\mathrm{LOX},\mathrm{LAR}\}$, and the daily production 
target for product $i \in \bar I$ at plant $p$ on day $d$ is denoted by 
$\tilde d_{pid}$. Finally, let $\mathcal P=\{\mathrm{ASU}_1,\mathrm{ASU}_2,
\mathrm{ASU}_3\}$ denote the set of ASU plants, with $p \in \mathcal P$. We next 
present the mathematical formulation of the scheduling model.

\subsection{Mass balance}

We formulate the inventory balance constraints using the ASU operating 
environment shown in Figure~\ref{fig:asuenv}. For each plant $p \in P$ and 
liquid product $i \in \bar I$, the inventory level evolves according to the 
difference between liquid production and dispatch decisions. The inventory 
balance is given in Equation~\ref{eq:iv}. Specifically, 
Equation~\ref{eq:iva} updates the inventory at the end of time period $t$ as the 
inventory available at the end of the previous period, plus the liquid production 
quantity $PQ_{pit}$ generated by the ASU, minus the dispatched quantity 
$DQ_{pit}$. The dispatch variables are subsequently linked to the daily target 
requirements. Equation~\ref{eq:ivb} imposes the product- and plant-specific 
storage capacity limits, ensuring that the inventory level does not exceed the 
available tank capacity.

\begin{figure}[h!]
  \centering
  \includegraphics[width=\textwidth]{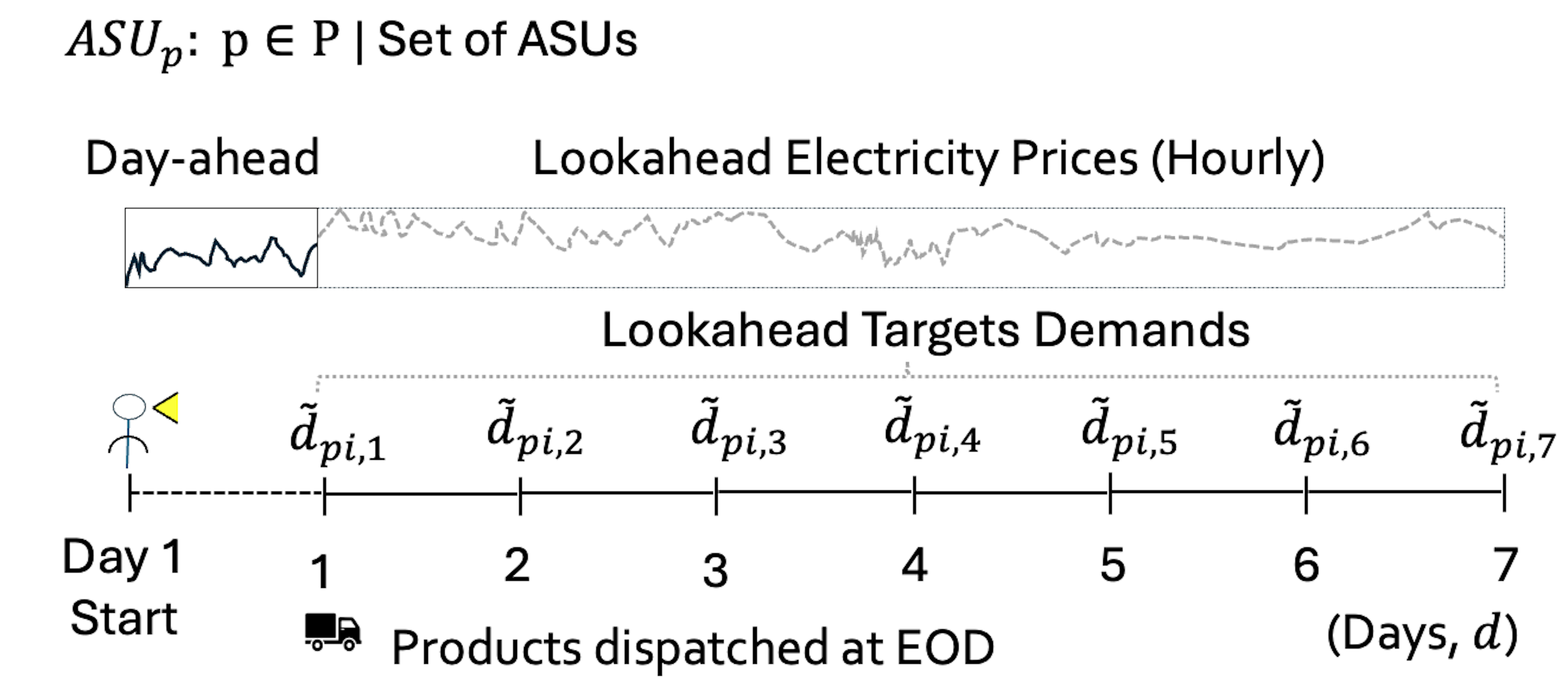}
  \caption{External state observed by an air separation unit operator.}
  \label{fig:situation}
\end{figure}

\begin{subequations}\label{eq:iv}
\begin{align}
IV_{pit} &= IV_{pi, t-1} + PQ_{pit} - DQ_{pit}  && \forall p \in P, i \in \bar{I}, t \in \bar{T} \label{eq:iva} \\
IV_{pit} &\leq IV_{pi}^u && \forall p \in P, i \in \bar{I}, t \in \bar{T} \label{eq:ivb}
\end{align}
\end{subequations}

Next, we impose the daily liquid product target to be met by each ASU in Equation~\ref{eq:target}. Equation~\ref{eq:targeta} enforces that the sum of hourly dispatched quantities over each day is equal to the target quantity $\tilde d_{pid}$. Equation~\ref{eq:targetb} ensures that liquid products are dispatched only at the end of each day over the weekly planning horizon.

\begin{subequations}\label{eq:target}
\begin{align}
\sum_{t=24(d-1)+1}^{24d} DQ_{pit} &= \tilde d_{pid}
&& \forall p \in P,\ i \in \bar I,\ d \in D, \label{eq:targeta}\\
DQ_{pit} &= 0
&& \forall p \in P,\ i \in \bar I,\ d \in D,\ 
t=24(d-1)+1,\ldots,24d-1.
\label{eq:targetb}
\end{align}
\end{subequations}

\subsection{Energy balance}

The considered ASU model has no auxiliary power source other than the regional energy market from which they receive locational marginal prices. Also, it has no ability to produce energy inhouse using for example, an CES system as proposed in \citet{Zhang2015}. Hence, the amount of power purchased would be directly equal to the amount consumed. Equation~\ref{eq:energy} states that electricity consumed by ASU in any time-period $E C_t$ would be equal to the amount of energy purchased from the energy market $E P_t$.
\begin{equation}\label{eq:energy}
EC_{pt} = EP_{pt} \quad \forall p \in P, t \in \bar{T}
\end{equation}
The electricity consumption in-turn would depend on the mode of operation in each plant. We model the consumption as linear function of production quantity as described in next section as part of disjunction.

\subsection{Internal operations}

In this section, we formulate the internal operating model of the ASU. As 
illustrated by the state-transition graph in Figure~\ref{fig:transition}, an ASU 
can physically occupy only one operating mode at any given time. Let
\[
M=\{m_{\mathrm{OFF}},m_{\mathrm{LS}},m_{\mathrm{LP}}\}
\]
denote the set of operating modes, corresponding to OFF, Liquid Start-up, and 
Liquid Production, respectively. We introduce the binary variable $y_{pmt}$, 
which is equal to 1 if plant $p \in P$ operates in mode $m \in M$ at time 
$t \in \bar T$, and 0 otherwise. Since each plant can operate in exactly one mode 
at any time, the following assignment constraint is imposed:
\begin{equation}\label{eq:mode_assignment}
\sum_{m \in M} y_{pmt} = 1
\qquad \forall p \in P,\ t \in \bar T.
\end{equation}

Among the considered operating modes, nonzero production is possible only in the 
Liquid Production mode. The OFF and Liquid Start-up modes therefore correspond to 
zero production. For each plant $p \in P$, the feasible production region in the 
Liquid Production mode is approximated using the convex hull of historical 
production samples. Let $S_p$ denote the number of historical samples available 
for plant $p$, and let each sample be represented by the production vector
\[
\mathbf{q}_p^s = \left(Q_{pi}^s\right)_{i \in I},
\qquad s=1,\ldots,S_p.
\]
The historical production data set for plant $p$ is then given by
\[
HPQ_p =
\left\{
\mathbf{q}_p^1,\mathbf{q}_p^2,\ldots,\mathbf{q}_p^{S_p}
\right\}.
\]
Let $J_p$ denote the set of vertices of the convex hull of $HPQ_p$, and let 
$v_{pji}$ denote the production quantity of product $i \in I$ at vertex 
$j \in J_p$. The hourly production quantity $PQ_{pit}$ is represented using 
convex-combination weights $\lambda_{pjt}$ as follows:
\begin{subequations}\label{eq:production}
\begin{align}
PQ_{pit} &= \sum_{j \in J_p} \lambda_{pjt} v_{pji}
&& \forall p \in P,\ i \in I,\ t \in \bar T,
\label{eq:productiona}\\
\sum_{j \in J_p} \lambda_{pjt} &= y_{p,m_{\mathrm{LP}},t}
&& \forall p \in P,\ t \in \bar T,
\label{eq:productionb}\\
\lambda_{pjt} &\ge 0
&& \forall p \in P,\ j \in J_p,\ t \in \bar T.
\label{eq:productionc}
\end{align}
\end{subequations}

Equation~\ref{eq:productiona} expresses the production vector as a convex 
combination of the vertices of the feasible production region. 
Equation~\ref{eq:productionb} links the convex-hull representation to the 
Liquid Production mode. When $y_{p,m_{\mathrm{LP}},t}=1$, the production vector 
lies within the convex hull of historical Liquid Production samples. Conversely, 
when $y_{p,m_{\mathrm{LP}},t}=0$, all convex-combination weights are forced to 
zero, and hence $PQ_{pit}=0$ for all products $i \in I$.

The electricity consumption of each ASU is modeled as the sum of a mode-dependent 
fixed electricity requirement and a production-dependent variable component. The 
fixed component captures the electricity required to maintain operation in a 
given mode, while the variable component captures the marginal electricity 
consumption associated with liquid production. The electricity consumption is 
given by
\begin{equation}\label{eq:electricity}
EC_{pt}
=
\sum_{m \in M} \delta_{pm} y_{pmt}
+
\sum_{i \in I} \gamma_{pi}^{\mathrm{LP}} PQ_{pit}
\qquad
\forall p \in P,\ t \in \bar T.
\end{equation}

Here, $\delta_{pm}$ denotes the fixed electricity requirement of plant $p$ when 
operating in mode $m$, and $\gamma_{pi}^{\mathrm{LP}}$ denotes the marginal 
electricity consumption associated with producing product $i$ in the Liquid 
Production mode. The parameterization reflects the physical behavior of the ASU 
across operating modes. Specifically, the OFF mode is assumed to have no 
electricity consumption, so
\[
\delta_{p,m_{\mathrm{OFF}}}=0
\qquad \forall p \in P.
\]
The Liquid Start-up mode requires electricity to bring the unit to an operable 
state, but no product is generated in this mode. Therefore, it is represented by 
a nonzero fixed electricity requirement and no production-dependent component. 
In contrast, the Liquid Production mode has both a nonzero fixed electricity 
requirement and a production-dependent variable electricity component. 
Equivalently, the mode-dependent electricity structure can be summarized as
\[
\delta_{p,m_{\mathrm{OFF}}}=0, 
\qquad
\delta_{p,m_{\mathrm{LS}}}>0,
\qquad
\delta_{p,m_{\mathrm{LP}}}>0
\qquad \forall p \in P,
\]
with the variable electricity coefficient applied only in the Liquid Production 
mode through $\gamma_{pi}^{\mathrm{LP}}$.

\subsection{Mode-transition constraints}

The feasible transitions among the liquefier operating modes are represented
using the state-transition graph shown in Figure~\ref{fig:transition}. Let
\[
M=\{m_{\mathrm{OFF}},m_{\mathrm{LS}},m_{\mathrm{LP}}\}
\]
denote the set of operating modes, corresponding to OFF, Liquid Start-up, and
Liquid Production, respectively. The set of allowable transitions is defined as
\[
\mathcal{TR}
=
\{(m_{\mathrm{OFF}},m_{\mathrm{LS}}),
  (m_{\mathrm{LS}},m_{\mathrm{LP}}),
  (m_{\mathrm{LP}},m_{\mathrm{OFF}})\},
\]
while the set of prohibited transitions is
\[
\mathcal{DTR}
=
\{(m_{\mathrm{OFF}},m_{\mathrm{LP}}),
  (m_{\mathrm{LP}},m_{\mathrm{LS}}),
  (m_{\mathrm{LS}},m_{\mathrm{OFF}})\}.
\]
We define the binary variable $y_{pmt}$ to be equal to 1 if plant $p$ is in mode
$m$ at time $t$, and 0 otherwise. Similarly, $z_{pmm't}$ is a binary variable
that is equal to 1 if plant $p$ transitions from mode $m$ to mode $m'$ at time
$t$, and 0 otherwise.

Since each plant can occupy exactly one mode at any time, we impose
\begin{equation}
\sum_{m \in M} y_{pmt}=1
\qquad
\forall p \in P,\ t \in \bar T.
\label{eq:mode_assignment}
\end{equation}
The evolution of the operating mode is governed by a transition-flow balance.
For each mode, the change in the corresponding mode-selection variable is equal
to the number of incoming transitions minus the number of outgoing transitions:
\begin{equation}
y_{pmt}-y_{pm,t-1}
=
\sum_{\substack{(m',m)\in \mathcal{TR}}}
z_{pm'm,t-1}
-
\sum_{\substack{(m,m')\in \mathcal{TR}}}
z_{pmm',t-1}
\qquad
\forall p \in P,\ m \in M,\ t \in \bar T.
\label{eq:transition_flow}
\end{equation}
At most one mode transition is allowed in any time period:
\begin{equation}
\sum_{(m,m')\in \mathcal{TR}} z_{pmm't}
\leq 1
\qquad
\forall p \in P,\ t \in \bar T.
\label{eq:one_transition}
\end{equation}
Transitions that are not allowed by the state graph are explicitly prohibited:
\begin{equation}
z_{pmm't}=0
\qquad
\forall p \in P,\ (m,m')\in \mathcal{DTR},\ t \in T.
\label{eq:forbidden_transition}
\end{equation}

The liquefier transition logic is further constrained by minimum OFF time,
start-up duration, and minimum ON time requirements. The parameters
$\theta_{\mathrm{min,OFF}}^{\mathrm{Liq}}$,
$\theta_{\mathrm{Startup}}^{\mathrm{Liq}}$, and
$\theta_{\mathrm{min,ON}}^{\mathrm{Liq}}$ are expressed in integer numbers of
hourly time periods. The following constraints are imposed only when the lagged
time indices belong to the extended horizon $T$.

If the liquefier transitions from Liquid Production to OFF, it must remain in
the OFF mode for at least $\theta_{\mathrm{min,OFF}}^{\mathrm{Liq}}$ time
periods:
\begin{equation}
y_{p,m_{\mathrm{OFF}},t}
\geq
\sum_{r=0}^{\theta_{\mathrm{min,OFF}}^{\mathrm{Liq}}-1}
z_{p,m_{\mathrm{LP}},m_{\mathrm{OFF}},t-1-r}
\qquad
\forall p \in P,\ t \in \bar T.
\label{eq:min_off}
\end{equation}
Once the liquefier is turned on from the OFF mode, it must remain in the
Liquid Start-up mode for the prescribed start-up duration
$\theta_{\mathrm{Startup}}^{\mathrm{Liq}}$:
\begin{equation}
y_{p,m_{\mathrm{LS}},t}
\geq
\sum_{r=0}^{\theta_{\mathrm{Startup}}^{\mathrm{Liq}}-1}
z_{p,m_{\mathrm{OFF}},m_{\mathrm{LS}},t-1-r}
\qquad
\forall p \in P,\ t \in \bar T.
\label{eq:startup_duration}
\end{equation}
After the start-up duration has elapsed, the liquefier must transition from
Liquid Start-up to Liquid Production:
\begin{equation}
z_{p,m_{\mathrm{LS}},m_{\mathrm{LP}},t}
\geq
z_{p,m_{\mathrm{OFF}},m_{\mathrm{LS}},
t-\theta_{\mathrm{Startup}}^{\mathrm{Liq}}}
\qquad
\forall p \in P,\ t \in \bar T.
\label{eq:startup_to_production_transition}
\end{equation}
Equivalently, the plant is forced to enter the Liquid Production mode after the
completion of start-up:
\begin{equation}
y_{p,m_{\mathrm{LP}},t}
\geq
z_{p,m_{\mathrm{OFF}},m_{\mathrm{LS}},
t-\theta_{\mathrm{Startup}}^{\mathrm{Liq}}-1}
\qquad
\forall p \in P,\ t \in \bar T.
\label{eq:startup_to_production_mode}
\end{equation}
Finally, once the liquefier enters Liquid Production, it must remain in the
Liquid Production mode for at least
$\theta_{\mathrm{min,ON}}^{\mathrm{Liq}}$ time periods:
\begin{equation}
y_{p,m_{\mathrm{LP}},t}
\geq
\sum_{r=0}^{\theta_{\mathrm{min,ON}}^{\mathrm{Liq}}-1}
z_{p,m_{\mathrm{LS}},m_{\mathrm{LP}},t-1-r}
\qquad
\forall p \in P,\ t \in \bar T.
\label{eq:min_on}
\end{equation}

\subsection{Boundary conditions}

Solving the scheduling model requires specifying the initial state of each ASU
and the desired terminal inventory levels. The inventory boundary conditions are
given by
\begin{subequations}\label{eq:inventory_boundary}
\begin{align}
IV_{pi0} &= IV_{pi}^{\mathrm{initial}}
&& \forall p \in P,\ i \in \bar I,
\label{eq:inventory_boundarya}\\
IV_{pi,t^f} &\geq IV_{pi}^{\mathrm{final}}
&& \forall p \in P,\ i \in \bar I.
\label{eq:inventory_boundaryb}
\end{align}
\end{subequations}
Equation~\ref{eq:inventory_boundarya} fixes the inventory available at the
beginning of the scheduling horizon, while
Equation~\ref{eq:inventory_boundaryb} ensures that the inventory remaining at
the end of the horizon is no less than the prescribed terminal level.

The initial operating mode of each ASU must also be specified. Let
$y_{pm}^{\mathrm{initial}}$ denote the initial mode of plant $p$, where
$y_{pm}^{\mathrm{initial}}=1$ if plant $p$ is initially in mode $m$, and 0
otherwise. These parameters satisfy
\[
\sum_{m \in M} y_{pm}^{\mathrm{initial}} = 1
\qquad \forall p \in P.
\]
The initial mode boundary condition is
\begin{equation}\label{eq:initial_mode}
y_{pm0} = y_{pm}^{\mathrm{initial}}
\qquad
\forall p \in P,\ m \in M.
\end{equation}

Finally, the switching constraints require knowledge of recent operating history,
since the minimum OFF time, start-up duration, and minimum ON time restrictions
may depend on transitions that occurred before the beginning of the optimization
horizon. Let $z_{pmm't}^{\mathrm{historical}}$ be equal to 1 if plant $p$
transitioned from mode $m$ to mode $m'$ at historical time $t$, and 0 otherwise.
The historical transition boundary condition is
\begin{equation}\label{eq:historical_transition}
z_{pmm't} = z_{pmm't}^{\mathrm{historical}}
\qquad
\forall p \in P,\ (m,m') \in \mathcal{TR},\
-\theta^H+1 \leq t \leq -1.
\end{equation}
Here,
\[
\mathcal{TR}
=
\left\{
(m_{\mathrm{OFF}},m_{\mathrm{LS}}),
(m_{\mathrm{LS}},m_{\mathrm{LP}}),
(m_{\mathrm{LP}},m_{\mathrm{OFF}})
\right\},
\]
and
\[
\theta^H =
\max
\left\{
\theta_{\mathrm{min,OFF}}^{\mathrm{Liq}},
\theta_{\mathrm{Startup}}^{\mathrm{Liq}},
\theta_{\mathrm{min,ON}}^{\mathrm{Liq}}
\right\}.
\]
Therefore, the historical horizon extends only as far back as the largest
transition-related lag required by the liquefier state-transition constraints.

\subsection{Objective function}

The objective of the scheduling problem is to minimize the total operating cost
of the ASU system over the planning horizon. Since third-party purchases and
external conversion processes are not considered in this formulation, the
operating cost is determined solely by the electricity required to operate the
ASUs. The total operating cost for plant $p$ is therefore given by
\begin{equation}\label{eq:plant_cost}
TOC_p =
\sum_{t \in \bar T} \varepsilon_{pt} EC_{pt}
\qquad \forall p \in P,
\end{equation}
where $\varepsilon_{pt}$ denotes the unit electricity price for plant $p$ at time
$t$, and $EC_{pt}$ denotes the electricity consumption of plant $p$ at time $t$.

All physical quantities, including production, inventory, dispatch, and
electricity consumption variables, are constrained to be nonnegative where
appropriate. The production targets assigned to the ASUs are assumed to be
feasible within the available plant capacities; therefore, external product
purchases are not included in the formulation. The ASUs are also assumed to be
supported by sufficient transportation capacity, so fleet-sizing and
vehicle-routing decisions are not modeled explicitly.

Finally, we note that the objective function used in the coordinated demand
response case study in the main text also includes shipping costs associated
with serving different customer regions. These costs are omitted here because
the present formulation describes the static status-quo scheduling problem, in
which demand is aggregated and plant operations are optimized only with respect
to time-varying electricity prices.

\section*{CRediT authorship contribution statement}

\noindent\textbf{Akshdeep Singh Ahluwalia:}
Conceptualization, Methodology, Software, Validation, Formal analysis,
Investigation, Data curation, Visualization, Writing -- original draft,
Writing -- review \& editing.

\medskip
\noindent\textbf{Zachary Wilson:}
Conceptualization, Supervision, Writing -- review \& editing.

\medskip
\noindent\textbf{Jeffrey E. Arbogast:}
Conceptualization, Supervision, Writing -- review \& editing.

\medskip
\noindent\textbf{Can Li:}
Conceptualization, Supervision, Project administration, Funding acquisition,
Writing -- review \& editing.

\section*{Declaration of competing interest}
Can Li reports that financial support for this research was provided by Air Liquide. Zachary Wilson and Jeffrey E. Arbogast are employees of Air Liquide. Akshdeep Singh Ahluwalia declares that he has no known competing financial interests or personal relationships that could have appeared to influence the work reported in this paper.

\section*{Funding}
This work was supported by Air Liquide through the 2023 Air Liquide Scientific Challenge.

\section*{Data availability}
The data used in this study are synthetic and are available from the corresponding author upon reasonable request. The computational implementation is not publicly available.

\section*{Disclaimer}

The results and conclusions presented in this work are based on
theoretical research using synthetic data and a simulated industrial gas supply chain. The case study is not representative of any actual Air Liquide supply chain.

\section*{Declaration of generative AI and AI-assisted technologies in the writing process}

During the preparation of this work, the authors used OpenAI's ChatGPT in order to improve the language and readability of the manuscript. After using this tool, the authors reviewed and edited the content as needed and take full responsibility for the content of the publication.

\bibliographystyle{elsarticle-harv} 
\bibliography{export}



\end{document}